\documentclass[10pt]{article}

\usepackage[lined,boxed,ruled,vlined,linesnumbered,nofillcomment]{algorithm2e}
\usepackage{bm,bbm}
\usepackage{tabularx,ragged2e,booktabs,caption,framed}
\usepackage[utf8]{inputenc}
\usepackage{float,url,hyperref,epic,eepic}
\usepackage{xcolor, enumerate}
\usepackage{graphicx}
\usepackage{amsmath, amsfonts, amssymb, amscd} 
\usepackage{calc}
\usepackage[numbers,square]{natbib}

\usepackage{listings}
\makeatletter
\newcommand{\pushright}[1]{\ifmeasuring@#1\else\omit\hfill$\displaystyle#1$\fi\ignorespaces}
\newcommand{\pushleft}[1]{\ifmeasuring@#1\else\omit$\displaystyle#1$\hfill\fi\ignorespaces}
\makeatother

\DeclareMathOperator{\poly}{poly}

\newcommand{\ttau}{{\tilde{\tau}}}
\newcommand{\tilt}{{\tilde{t}}}

\DeclareMathOperator{\diag}{diag}
\DeclareMathOperator{\Ring}{R}
\DeclareMathOperator{\cop}{cpr}
\DeclareMathOperator{\sgn}{sgn}
\DeclareMathOperator{\odty}{odt}
\DeclareMathOperator{\sig}{sig}

\DeclareMathOperator{\adj}{adj}
\DeclareMathOperator{\SYM}{sym}
\DeclareMathOperator{\ord}{ord}
\DeclareMathOperator{\can}{can}
\DeclareMathOperator{\red}{red}
\DeclareMathOperator{\tsym}{\textsc{sym}} 

\DeclareMathOperator{\Sym}{\Upsilon}

\newcommand{\symtk}{\tsym_{2^k}}
\newcommand{\canp}{\can_p}

\newcommand{\sigp}{\sig_p}
\newcommand{\sigt}{\sig_2}
\newcommand{\tSym}{\tilde{\Sym}}

\newcommand{\copt}{\cop_2}
\newcommand{\copp}{\cop_p}

\newcommand{\scale}{\mathbb{I}}
\DeclareMathOperator{\excess}{exs}

\newcommand{\tMQ}{\tilde{\MQ}}

\newcommand{\tMH}{\tilde{\MH}}
\newcommand{\tMU}{\tilde{\MU}}
\newcommand{\SymP}{\Upsilon_p}
\newcommand{\SymT}{\Upsilon_2}

\newcommand{\pexcess}{\excess_p}

\newcommand{\zRz}[1]{\bbZ/#1\bbZ}
\newcommand{\zqz}{\zRz{q}}
\newcommand{\zpz}{\zRz{p}}

\newcommand{\ztkz}{\zRz{2^k}}

\newcommand{\zpkz}{\zRz{p^k}}

\newcommand{\modtk}{\bmod{2^k}}
\newcommand{\eqq}{\overset{q}{\sim}}
\newcommand{\eqR}{\overset{\bbR}{\sim}}

\newcommand{\eqp}{\overset{p^*}{\sim}}

\newcommand{\eqt}{\overset{2^*}{\sim}}
\newcommand{\eqx}[1]{\overset{#1}{\sim}}

\newcommand{\dSym}{\det(\Sym)}
\newcommand{\PSym}{\bbP_{\Sym}}

\newcommand{\ordp}{\ord_p}
\newcommand{\ordt}{\ord_2}
\newcommand{\sgnp}{\sgn_p}

\DeclareMathOperator{\scalep}{\scale_p}
\DeclareMathOperator{\scalet}{\scale_2}

\newcommand{\MTP}{\mathtt{T}^{+}}
\newcommand{\MTM}{\mathtt{T}^{-}}

\newcommand{\MA}{\mathtt{A}} \newcommand{\MB}{\mathtt{B}}
 \newcommand{\MD}{\mathtt{D}}
 
\newcommand{\MH}{\mathtt{H}} \newcommand{\MI}{\mathtt{I}}
 
\newcommand{\MM}{\mathtt{M}} 
\newcommand{\MQ}{\mathtt{Q}} \newcommand{\MS}{\mathtt{S}}
 \newcommand{\MU}{\mathtt{U}}
\newcommand{\MV}{\mathtt{V}} 
\newcommand{\MX}{\mathtt{X}}

\newcommand{\Vb}{\mathbf{b}}

\newcommand{\Vd}{\mathbf{d}} 
 
\newcommand{\Vx}{\mathbf{x}} 
\newcommand{\Vy}{\mathbf{y}}

\newcommand{\Vv}{\mathbf{v}}
\newcommand{\Vw}{\mathbf{w}}

\DeclareMathOperator{\type}{type}

\DeclareMathOperator{\SGN}{\textsc{sgn}}
\DeclareMathOperator{\SGNI}{\textsc{sgn}^\times}

\DeclareMathOperator*{\argmax}{arg\,max}
\DeclareMathOperator{\maj}{maj}

\DeclareMathOperator{\gl}{GL}
\DeclareMathOperator{\SL}{SL}
\DeclareMathOperator{\gln}{GL_n}

\DeclareMathOperator{\sln}{SL_n}

\newcommand{\legendre}[2]{%
\left( \frac{#1}{#2} \right)%
}

\DeclareMathOperator{\gen}{Gen}
\newcommand{\Lattice}{L}

\newtheorem{fact}{\textsc{Fact}}

\newcommand{\bbR}{\mathbb{R}}

\newcommand{\bbQ}{\mathbb{Q}}
\newcommand{\bbZ}{\mathbb{Z}}
\newcommand{\bbP}{\mathbb{P}}

\newcommand{\fS}{\mathcal{S}}

\iftrue

\renewcommand{\fS}{\mathcal{S}}

\fi

\DeclareSymbolFont{bbold}{U}{bbold}{m}{n}
\DeclareSymbolFontAlphabet{\mathbbold}{bbold}

\newcommand{\RET}{\textbf{return }}
\newcommand{\THEN}{\textbf{then}}
\newcommand{\FOR}{\textbf{for }}

\newcommand{\ENDIF}{\textbf{end if}}
\newcommand{\DO}{\textbf{do}}
\newcommand{\PROC}{\textbf{proc}}
\newcommand{\END}{\textbf{end }}
\newcommand{\NOT}{\textbf{not }}
\DeclareMathOperator{\TYPE}{type}
\DeclareMathOperator{\EVEN}{even}
\DeclareMathOperator{\ODD}{odd}
\DeclareMathOperator{\NOPS}{nops}
\newcommand{\AND}{\textbf{ and }}
\newcommand{\IN}{\textbf{in }}
\newcommand{\IF}{\textbf{if }}
\newcommand{\OR}{\textbf{ or }}
\newcommand{\INT}{\textbf{ intersect }}
\DeclareMathOperator{\FXI}{\textsc{fxi}}

\usepackage{setspace}

\newsavebox{\theorembox} 
\newsavebox{\lemmabox}
\newsavebox{\claimbox}
\newsavebox{\corollarybox} \newsavebox{\propositionbox}
\newsavebox{\examplebox} \newsavebox{\conjecturebox}
\newsavebox{\algbox} \newsavebox{\qbox} \newsavebox{\problembox}
\newsavebox{\definitionbox} \newsavebox{\assumptionbox}
\newsavebox{\hypothesisbox} \newsavebox{\obsbox} 
\savebox{\theorembox}{\noindent\bf Theorem} 
\savebox{\lemmabox}{\noindent\bf Lemma}
\savebox{\claimbox}{\noindent\bf Claim}
\savebox{\corollarybox}{\noindent\bf Corollary}
\savebox{\propositionbox}{\noindent\bf Proposition}
\savebox{\examplebox}{\noindent\bf Example}
\savebox{\conjecturebox}{\noindent\bf Conjecture}
\savebox{\algbox}{\noindent\bf Algorithm}
\savebox{\qbox}{\noindent\bf Question}
\savebox{\definitionbox}{\noindent\bf Definition}
\savebox{\problembox}{\noindent\bf Problem}
\savebox{\assumptionbox}{\noindent\bf Assumption}
\savebox{\hypothesisbox}{\noindent\bf Hypothesis}
\savebox{\obsbox}{\noindent\bf Observation}
\newtheorem{theorem}{\usebox{\theorembox}}
\newtheorem{lemma}[theorem]{\usebox{\lemmabox}}

\newtheorem{claim}{\usebox{\claimbox}}

\newtheorem{definition}{\usebox{\definitionbox}}
\newcommand{\qed}{\;\;\;\Box} 
\newenvironment{proof}{{\bf Proof:}}{\hfill\(\qed\)\newline}

\title{Generating a Quadratic Forms from a Given Genus}
\author{
	Chandan Dubey\\
	\texttt{chandan.dubey@inf.ethz.ch}
	\and
	Thomas Holenstein\\
	\texttt{thomas.holenstein@inf.ethz.ch}
}
\date{Institut f\"ur Theoretische Informatik, ETH Z\"urich}

\begin{document}
\maketitle

\begin{abstract}
Given a non-empty genus in $n$ dimensions with 
determinant $d$, we give a randomized algorithm
that outputs a quadratic form from this genus. The time complexity
of the algorithm is $\poly(n,\log d)$; assuming Generalized 
Riemann Hypothesis (GRH). 
\end{abstract}

\renewcommand{\baselinestretch}{0.95}\normalsize
\tableofcontents
\renewcommand{\baselinestretch}{1}\normalsize
\newpage

\section{Introduction}

Let $\Ring$ be a commutative ring with unity and $\Ring^\times$ 
be the set of units (i.e., invertible elements) of $\Ring$. A 
quadratic form over the ring $\Ring$ in $n$-formal variables 
$x_1,\cdots,x_n$ in an expression 
$\sum_{1\leq i, j \leq n}a_{ij}x_ix_j$, where $a_{ij}=a_{ji} 
\in \Ring$. A quadratic form can equivalently be represented by
a symmetric matrix $\MQ^n=(a_{ij})$ such that 
$Q(x_1,\cdots,x_n)=(x_1,\cdots,x_n)'\MQ(x_1,\cdots,x_n)$. The
quadratic form is called integral if $\Ring=\bbZ$ and the 
determinant of the quadratic form $Q$ is defined as $\det(\MQ)$.
In this paper, we concern ourselves with integral quadratic
forms, henceforth referred only as {\em quadratic forms}.

One of the classical problems in the study of quadratic
forms is their classification into equivalence classes. 
Two quadratic forms $\MQ_1, \MQ_2$ are said to be equivalent
over a ring $\Ring$ if there exists a transformation 
$\MU \in \gln(\Ring)$ such that $\MQ_1 = \MU'\MQ_2\MU$.
For example, $\MQ_1$ and $\MQ_2$ are $q$-equivalent, for an
integer $q$ (denoted, $\MQ_1 \eqq \MQ_2$), if there exists a
matrix $\MU \in \gln(\zqz)$ such that 
$\MQ_1 \equiv \MU'\MQ_2\MU \bmod q$. Intuitively, 
$q$-equivalence means
that there exists an invertible linear change of variables over 
$\zqz$ that transforms one form to the other. Gauss \cite{Gauss86}
gives a complete classification of binary quadratic forms 
(i.e., $n=2$).

Two quadratic forms are said to be in the 
same genus if they are equivalent over the reals $\bbR$ and also
over $\zqz$ for all positive integers $q$. In this paper, we 
consider the following problem: given a description of a non-empty
genus, produce a quadratic form from that genus. 
A discussion of the problem can
be found in Conway and Sloane \cite{CS99}, page 403. The
best algorithm for this problem is based on Minkowski 
Reduced forms and takes $O(d^{n^2})$ time for genus in dimension
$n$ with determinant $d$.

The skeleton of our
algorithm is similar to the algorithm given by Hartung 
\cite{Hartung08}. His thesis uses an equivalent but different
approach based on Cassels \cite{Cassels78}. Unfortunately, there 
are several gaps in his construction. There are also mistakes when
dealing with prime~2. But, the most
severe problem with the algorithm is that its time
complexity is proportional to $n^n$ i.e., it is not
polynomial. A discussion can be found in Section \ref{sec:Hartung}.

We mention here, a connection of our problem to 
lattices as studied in the Computer Science community.
A {\em full-rank} lattice $\Lattice$ in $\bbR^n$ is a discrete subgroup of
$\bbR^n$ which is the set of all integer linear combinations of 
$n$-linearly independent vectors, say $\Vb_1,\ldots,\Vb_n$ i.e., 
$\Lattice=\{\sum_{i=1}^n z_i \Vb_i \mid z_1,\ldots,z_n \in \bbZ\}$. 
The matrix $\MB=[\Vb_1,\ldots,\Vb_n]$ is called
the {\em basis} of the lattice and the matrix $\MQ=\MB'\MB$ is called
a {\em Gram} matrix of the lattice. A lattice is {\em integral} if its
Gram matrix has only integer entries. It is not difficult to see
that the Gram matrix of a lattice defines a positive definite 
quadratic form.

Two lattices
are called {\em isomorphic} if one can be transformed into
another by an orthogonal linear transformation.
A fundamental question, called the Lattice 
Isomorphism Problem (LIP), is to decide if two given
Gram matrices come from isomorphic lattices. 
In other words, given two Gram matrices
$\MQ_1$ and $\MQ_2$ one has to decide if there exists a 
unimodular matrix $\MU$ such that $\MQ_2=\MU'\MQ_1\MU$.
For Gram matrices in dimension $n$ and determinant $d$,
the problem can be solved using Minkowski Reduced Forms
(see Section 10, Chapter 10 \cite{CS99}) in time
$O(d^{n^2})$. Other exhaustive search algorithms are known,
see \cite{Dietmann03,Siegel72}. Recently, Regev and 
Haviv \cite{HR13}
gave an algorithm with time complexity which is 
$n^{O(n)}$ times the size of the input.

The shortest vector problem (SVP) is the problem of finding the
shortest non-zero vector in a given lattice.
The current best known hardness for SVP is given
by Regev-Haviv \cite{HR07} and is based on tensoring lattice bases
in the hope of amplifying the length of the shortest vector. 
This approach fails in general. For large enough dimension $n$, there are 
self dual lattices with shortest vector $\Omega(\sqrt{n})$. The
usual tensoring among these lattices fails to amplify the length
of the shortest vector (Lemma~2.4, \cite{HR07}). It is not known
how one can construct self-dual lattice with shortest vector
length $\Omega(\sqrt{n})$ but it can be shown that such lattices
exist in large dimensions. The proof of existence (see page~48, 
\cite{MH73}) uses the
Smith-Minkowski-Siegel mass formula; which computes 
the average number vectors of a certain length in a genus.
One way to generate a self-dual lattice with shortest vector 
$\Omega(\sqrt{n})$ is to sample a lattice according to a certain
distribution from a specific genus (see Milnor-Husemoller 
\cite{MH73}). Our result falls short in the following way. Given
this specific genus, we can construct one lattice but we do not
know how to sample according to the distribution specified in 
\cite{MH73}. In this respect, our work can be seen as an important 
first step towards construction of self-dual lattices with shortest
vector $\Omega(\sqrt{n})$.

\paragraph*{Our Contributions.}

Let $d$ be the determinant of a genus in dimension $n$. 
We present a $\poly(n,\log d)$ Las Vegas algorithm
that outputs a quadratic form in the genus with constant probability.
Our construction technique is inspired by the proof of 
Smith-Minkowski-Siegel mass formula given by Siegel 
\cite{Siegel35} and uses similar notations as Conway-Sloane
\cite{CS99}. 

A significant feature of our work is the simplification achieved by
not using $p$-adic numbers, a staple in the analysis of 
integral quadratic forms \cite{CS99,Kitaoka99,Kneser02,Siegel35}.

\section{Preliminaries}

Integers and ring elements are denoted by lowercase letters,
vectors by bold lowercase letters and matrices by typewriter
uppercase letters. The $i$'th component of a vector $\Vv$ is
denoted by $v_i$. We use the notation $(v_1,\cdots,v_n)$ for
a column vector and the transpose of matrix $\MA$ is denoted by
$\MA'$. The matrix
$\MA^n$ will denote a $n\times n$ square matrix.
The scalar product of two vectors will be denoted 
$\Vv'\Vw$ and equals $\sum_i v_iw_i$. The standard Euclidean norm of the 
vector $\Vv$ is denoted by $||\Vv||$ and equals $\sqrt{\Vv'\Vv}$.

If $\MQ_1^n, \MQ_2^m$ are matrices, then the {\em direct product}
of $\MQ_1$ and $\MQ_2$ is denoted by $\MQ_1\oplus\MQ_2$ and is
defined as $\diag(\MQ_1,\MQ_2)=\begin{pmatrix}
\MQ_1&0\\0&\MQ_2
\end{pmatrix}$.
Given two matrices $\MQ_1$ and $\MQ_2$ with the same number
of rows, $[\MQ_1,\MQ_2]$ is the matrix which is obtained
by concatenating the two matrices columnwise.
A matrix is called unimodular
if it is an integer $n\times n$ matrix with determinant $\pm 1$.
If $\MQ^n$ is a $n\times n$ integer matrix and $q$ is a positive integer
then $\MQ \bmod{q}$ is defined as the matrix with all entries of $\MQ$
reduced modulo $q$.

Let $\Ring$ be a commutative ring with unity and $\Ring^\times$ be the
set of units (i.e., invertible elements) of $\Ring$.
If $\MQ \in \Ring^{n\times n}$ is a square matrix, the 
{\em adjugate} of $\MQ$ is defined as the transpose of 
the cofactor matrix and is denoted by $\adj(\MQ)$. The matrix $\MQ$ 
is invertible if and only if $\det(\MQ)$ is a unit of $\Ring$. 
In this case, $\adj(\MQ)=\det(\MQ)\MQ^{-1}$. The set of invertible 
$n\times n$ matrices over $\Ring$ is denoted by $\gln(\Ring)$. The 
subset of matrices with determinant $1$ will be denoted by 
$\sln(\Ring)$.

\begin{fact}\label{fact:gln}
A matrix $\MU$ is in $\gln(\Ring)$ iff $\det(\MU) \in \Ring^\times$. 
\end{fact}

The set of odd primes is denoted by $\bbP$. 
We define $\bbQ/(-1)\bbQ=\zRz{(-1)}:=\bbR$. 
For every prime $p$ and positive integer $k$,
we define the ring $\zpkz=\{0,\cdots,p^k-1\}$, where product
and addition is defined modulo $p^k$.

Let $p$ be a prime, and $a, b$ be integers. Then, 
$\ordp(a)$ is the largest integer exponent of $p$ such that $p^{\ordp(a)}$
divides $a$. We let $\ordp(0) = \infty$. The $p$-coprime part of $a$ is then 
$\copp(a)=\frac{a}{p^{\ordp(a)}}$. Note that $\copp(a)$ is,
by definition, a unit of $\zpz$. 
For $\frac{a}{b}$, a rational number, we define
$\ordp(\frac{a}{b})=\ordp(a)-\ordp(b)$. The $p$-coprime part of $\frac{a}{b}$
is denoted as
$\copp(\frac{a}{a})$ and equals $\frac{a/p^{\ordp(a)}}{b/p^{\ordp(b)}}$.
For a positive integer $q$,
one writes $a\equiv b \bmod{q}$, if $q$ 
divides $a-b$. By $x:=a \bmod{q}$, we mean that
$x$ is assigned the unique value $b \in \{0,\cdots,q-1\}$ such that 
$b \equiv a \bmod{q}$.
An integer $t$ is 
called a {\em quadratic residue} modulo $q$ if $\gcd(t,q)=1$ and
$x^2\equiv t \bmod{q}$ has a solution. 

\begin{definition}\label{def:Legendre} 
Let $p$ be an odd prime, and $t$ be a positive integer with
$\gcd(t,p)=1$.
Then, the Legendre-symbol of $t$ with 
respect to $p$ is defined as follows.
\begin{displaymath}
\legendre{t}{p} = t^{(p-1)/2} \bmod p = \left\{ \begin{array}{ll}
1 & \textrm{if $t$ is a quadratic residue modulo $p$}\\
-1 & \textrm{otherwise.}
\end{array}\right.
\end{displaymath}
\end{definition}

For the prime~2, there is an extension of Legendre symbol called the
Kronecker symbol. It is defined for odd integers $t$ and 
$\legendre{t}{2}$ equals $1$ if $t\equiv \pm 1 \bmod 8$, and $-1$
if $t \equiv \pm 3 \bmod 8$.

The Law of Quadratic Reciprocity, conjectured by Euler and Legendre
and first proved by Gauss, says that if 
$p_1, p_2$ are distinct primes, and $p$ is an odd prime, then
\begin{align}\label{QuadraticReciprocity}
\legendre{p_1}{p_2} = \left\{
\begin{array}{ll} 
- \legendre{p_2}{p_1} & \text{if $p_1\equiv p_2 \equiv 3 \pmod 4$,}\\
\legendre{p_2}{p_1} & \text{otherwise.}
\end{array}\right. \qquad
\legendre{p\vphantom{2}}{2\vphantom{p}} = \legendre{2}{p}
\end{align}

The $p$-sign of $t$, denoted $\sgnp(t)$, is defined as 
$\legendre{\copp(t)}{p}$
for odd primes $p$ and $\copt(t) \bmod 8$ otherwise. We
also define $\sgnp(0)=0$, for all primes $p$.
Thus,
\[
\sgnp(0) = 0 \qquad 
\sgnp(t>0) \in \left\{
	\begin{array}{ll}
	\{+1,-1\} & \text{if $p$ is odd}\\
	\{1,3,5,7\} & \text{otherwise}
	\end{array}\right.
\]

The following lemma is well known.

\begin{lemma}\label{lem:QR}
Let $p$ be an odd prime. Then, there
are $\frac{p-1}{2}$ quadratic residues and $\frac{p-1}{2}$ 
quadratic non-residues modulo $p$. Also, every quadratic residue
in $\zpz$ can be written as a sum of two quadratic non-residues
and every quadratic non-residue can be written as a sum of two quadratic
residues.
\end{lemma}

An integer $t$ is a square modulo $q$ 
if there exists an integer $x$ such that $x^2\equiv t \pmod{q}$.
The integer $x$ is called the {\em square root} of $t$
modulo $q$. If no such $x$ exists, then $t$ is a non-square modulo $q$.

\begin{definition}\label{def:NonSquare}
Let $p$ be a prime and $\frac{x}{y}$ be a rational number. Then, $\frac{x}{y}$
can be uniquely written as $\frac{x}{y}=p^{\alpha}\frac{a}{b}$, 
where $a,b$ are  units of $\zpz$. We say that $\frac{x}{y}$ is a 
$p$-antisquare if $\alpha$ is odd and $\sgnp(a) \neq \sgnp(b)$.
\end{definition}

The following lemma is folklore and gives the necessary and sufficient
conditions for an integer $t$ to be a square modulo $p^k$. For
completeness, a proof is provided in Appendix \ref{sec:Proofs}.

\begin{lemma}\label{lem:Square}
Let $p$ be a prime, $k$ be a positive integer and 
$t \in \zpkz$ be a non-zero integer. Then, $t$ is a 
square modulo $p^k$ if and only if $\ordp(t)$ is even and
$\sgnp(t)=1$. 
\end{lemma}

\begin{definition}\label{def:Prim}
Let $p^k$ be a prime power. A vector $\Vv \in (\zpkz)^n$ is called 
primitive if there exists a component $v_i$, $i \in [n]$, of $\Vv$ such
that $\gcd(v_i,p)=1$. Otherwise, the vector $\Vv$ is 
non-primitive.
\end{definition}

Our definition of primitiveness of a vector is different but equivalent 
to the 
usual one in the literature. A
vector $\Vv \in (\zqz)^n$ is called primitive over $\zqz$ for a 
composite integer $q$ if it is primitive modulo $p^{\ordp(q)}$ for
all primes that divide $q$.

\paragraph{Randomized Algorithms.} 
Our randomized algorithms are 
Las Vegas algorithms. They either fail
and output nothing, or produce a correct answer. The 
probability of failure is bounded by a constant. Thus, for any
$\delta>0$, it is possible to repeat the algorithm 
$O(\log \frac{1}{\delta})$ times and succeed with probability at least
$1-\delta$. Henceforth, these algorithms will be called
{\em randomized algorithms}.

Our algorithms perform two kinds of operations. Ring operations e.g.,
multiplication, additions, inversions over $\zpkz$ and operations
over integers $\bbZ$ e.g., multiplications, additions, divisions etc
and operations over integers $\bbZ$.
The runtime for all these operations is treated as constant i.e., $O(1)$
and the time complexity of the algorithms is measured in terms of 
ring operations. Note that the complexity cannot be assumed to be
$O(1)$ if the numbers are doubly exponential in $n$. Thus, we make sure
than the numbers generated during the algorithm are bounded by 
$2^{\poly(n,d)}$.
Sometimes, we also need to sample a 
uniform ring element from $\zpkz$. We adapt the convention that sampling
a uniform ring elements also takes $O(1)$ ring operations.

For example, the Legendre symbol of an integer $a$ can be computed
by fast exponentiation in $O(\log p)$ ring operations over $\zpz$ 
while $\ordp(t)$
for $t \in \zpkz$ can be computed by fast exponentiation in $O(\log k)$ ring
operations over $\zpkz$.

Let $\omega$ be the constant, such that multiplying two $n\times n$ 
matrices over $\zpkz$ takes
$O(n^{\omega})$ ring operations. 

\paragraph{Dirichlet's Theorem.}

Let $a,q$ be positive integers such that $\gcd(a,q)=1$.
Dirichlet's theorem states
that there are infinitely many primes of the form $a+zq$, where
$z$ is a non-negative integer. The following theorem gives a
quantitative version of Dirichlet's theorem using Generalized 
Riemann Hypothesis (GRH). A proof of the theorem 
can be found in any analytic number theory book, for 
example \cite{Kowalski04}.

\begin{theorem}\label{thm:ERH}
Let $a, q$ be integers such that $\gcd(a,q)=1$ and $S$ be the set
$\{a+zq\mid z \in \bbZ, a+zq \leq q^3\}$. Then assuming GRH, there exists a 
constant $c$ such that $S$ has $c\frac{|S|}{\log|S|}$ primes.
\end{theorem}

Another implication of GRH is that the 
smallest quadratic non-residue modulo $p$, for odd prime $p$; is
a number less than $3(\ln p)^2/2$, see \cite{Ankeny52,Wedeniwski01}.
Thus, assuming GRH, a quadratic residue modulo $p$ can be found 
deterministically in time $O(\log^3p)$ ring operations over $\zpz$ 
by trying all integers $\leq 3(\ln p)^2/2$.

\paragraph{Quadratic Form.} 
An $n$-ary quadratic form over a ring $\Ring$
is a symmetric matrix $\MQ \in \Ring^{n\times n}$, interpreted as the following
polynomial in $n$ formal variables $x_1,\cdots, x_n$ of uniform degree~2.
\[
\sum_{1\leq i,j \leq n}\MQ_{ij}x_ix_j = 
\MQ_{11}x_1^2 + \MQ_{12}x_1x_2 + \cdots = \Vx'\MQ\Vx
\]
The quadratic form is called {\em integral} if it is defined over the ring
$\bbZ$. It is called positive definite if for all non-zero column vectors
$\Vx$, $\Vx'\MQ\Vx > 0$. This work deals with integral quadratic forms,
henceforth called simply {\em quadratic forms}.
The {\em determinant} of the quadratic 
form is defined as $\det(\MQ)$. 
A quadratic form is called {\em diagonal} if $\MQ$ is a diagonal matrix. 

Given a set of formal variables 
$\Vx=\begin{pmatrix}x_1 & \cdots & x_n\end{pmatrix}'$ one can make a linear 
change of variables to $\Vy=\begin{pmatrix}y_1 & \cdots & y_n\end{pmatrix}'$ 
using a matrix $\MU \in \Ring^{n\times n}$ by setting $\Vy=\MU\Vx$. 
If additionally, 
$\MU$ is invertible over $\Ring$ i.e., $\MU \in \gln(\Ring)$, then this 
change of 
variables is reversible over the ring. We now define the equivalence of
quadratic forms over the ring $\Ring$ (compare with Lattice Isomorphism).

\begin{definition}\label{def:equiv}
Let $\MQ_1^n, \MQ_2^n$ be quadratic forms over a ring $\Ring$. They are called 
$\Ring$-{\em equivalent} if there exists a $\MU \in \gln(\Ring)$ such that 
$\MQ_2=\MU'\MQ_1\MU$.
\end{definition}

If $\Ring=\zqz$, for some positive integer $q$, then two integral
quadratic forms $\MQ_1^n$ and $\MQ_2^n$ will be called $q$-equivalent (denoted,
$\MQ_1\eqq \MQ_2$)
if there exists a matrix $\MU \in \gln(\zqz)$ such that 
$\MQ_2\equiv\MU'\MQ_1\MU\pmod q$.
For a prime $p$, they
are $p^*$-equivalent (denoted, $\MQ_1\eqp \MQ_2$) if  they
are $p^k$-equivalent for every positive integer $k$. Additionally,
$(-1)^*$-equivalence as well as $(-1)$-equivalence mean
equivalence over the reals $\bbR$.

Let $\MQ^n$ be a $n$-ary integral quadratic form, and $q,t $ be positive 
integers. If the equation $\Vx'\MQ\Vx \equiv t \pmod{q}$ has a solution
then we say that $t$ has a $q$-representation in $\MQ$ (or $t$ has 
a representation in
$\MQ$ over $\zqz$). Solutions 
$\Vx \in (\zqz)^n$ to the equation are called $q$-{\em representations} of
$t$ in $\MQ$. We classify the representations into two
categories: {\em primitive} and {\em non-primitive} (see Definition 
\ref{def:Prim}). The following lemma shows that a primitive 
representation can be extended to a invertible transformation.

\begin{lemma}\label{lem:ExtendPrimitive}
Let $p$ be a prime, $k$ be a positive integer and 
$\Vx \in (\zpkz)^n$ be a primitive vector.
Then, an $\MA$ can be found in $O(n^2)$ ring operation
such that $[\Vx,\MA] \in \sln(\zpkz)$.
\end{lemma}
\begin{proof}
The column vector $\Vx=(x_1,\cdots,x_n)$ is primitive, hence there 
exists a $x_i$, $i\in [n]$ such that $x_i$ is invertible over $\zpkz$.
It is easier to write the matrix $\MU$, which equal $[\Vx,\MA]$ where
the row $i$ and $1$ or $[\Vx,\MA]$ are swapped.
\[
\MU = \begin{pmatrix}x_i & 0 \\
\Vx_{-i} & x_i^{-1}\bmod{p^k} \oplus \MI^{n-2}
\end{pmatrix}\qquad \Vx_{-i}=(x_1,\cdots,x_{i-1},x_{i+1},\cdots,x_n)
\]
The matrix $\MU$ has determinant~1 modulo $p^k$ and hence is invertible
over $\zpkz$. The lemma now follows from the fact that the swapped matrix 
is invertible iff the original matrix is invertible.
\end{proof}

For the following result, see Theorem~2, \cite{Jones50}.

\begin{theorem}\label{thm:diagonal}
An integral quadratic form $\MQ^n$ is equivalent to a quadratic
form 
$q_1 \oplus \cdots \oplus q_a \oplus q_{a+1} \oplus \cdots \oplus q_n$ 
over the field of rationals $\bbQ$,
where $a \in [n]$, $q_1, \cdots, q_a$ are positive rational numbers and
$q_{a+1}, \cdots, q_n$ are negative rational numbers.
\end{theorem}

The {\em signature} (also, $(-1)$-signature) of the form 
$\MQ$ (denoted $\sig(\MQ)$, also  $\sig_{(-1)}(\MQ)$) is defined
as the number $2a-n$, where $a$ is the integer in 
Theorem \ref{thm:diagonal}. 

Each rational number $q_i$ in Theorem
\ref{thm:diagonal} can be written uniquely as 
$p^{\alpha_i}a_i$, where $\alpha_i = \ordp(q_i)$ and $a_i = \copp(q_i)$.
Let $m$ be the number of $p$-antisquares among $q_1,\cdots,q_n$. 
Then, we define the $p$-signature of $\MQ$ as follows.
\begin{align}\label{def:PSignature}
\sigp(\MQ)=\left\{
\begin{array}{ll}
p^{\alpha_1}+ p^{\alpha_2} + \cdots + p^{\alpha_n} + 4m  \pmod{8}  & p\neq 2 \\
a_1 + a_2 + \cdots + + a_n + 4m \pmod{8} & p=2 
\end{array}\right.
\end{align}
The $2$-signature is also known as the {\em oddity} and is denoted by
$\odty(\MQ)$. 
Even though there are different ways to diagonalize a quadratic form
over $\bbQ$, the signatures are an invariant for the quadratic form.

For each $p \in \{-1,2\} \cup \bbP$, we define the $p$-excess of $\MQ$ 
as follows.
\begin{align}\label{def:PExcess}
\pexcess(\MQ) = \left\{
	\begin{array}{ll}
	\sigp(\MQ) - n & p \neq 2 \\
	 n - \sigt(\MQ) & p = 2 
	\end{array}\right.
\end{align}

\begin{theorem}\label{thm:oddity}(\cite[page~76]{Cassels78}, \cite[Theorem~29]{Jones50})
Let $\MQ^n$ be an integral quadratic form. Then,
\begin{align}\label{SumFormula}
\begin{array}{ll}
\sum_{p \in \{-1,2\}\cup \bbP} \pexcess(\MQ) &\equiv 0 \pmod 8 \text{ or equivalently,}\\
\sig(\MQ) + \sum_{p \in \bbP} \pexcess(\MQ) &\equiv \odty(\MQ) \pmod 8\;. 
\end{array}
\end{align}
\end{theorem}

The Equation \ref{SumFormula} is also referred to as the {\em 
oddity formula} in the literature.

\paragraph{Diagonalizing a Quadratic Form.}

For the ring $\zpkz$ such that $p$ is odd, there
always exists an equivalent quadratic form which is also diagonal 
(see \cite{CS99}, Theorem 2, page 369).
Additionally, one can explicitly find the invertible change of
variables that turns it into a diagonal quadratic form.
The situation is tricky over the ring $\ztkz$. Here, it might not
be possible to eliminate all mixed terms, i.e., terms of the form
$2a_{ij}x_ix_j$ with $i\neq j$. For example, consider
the quadratic form $2xy$ i.e., 
$\begin{pmatrix}0 & 1 \\ 1 & 0 \end{pmatrix}$ over $\ztkz$. 
An invertible linear change of variables over $\ztkz$ is
of the following form.
\begin{gather*}
\begin{array}{l}x \to a_1x_1 + a_2x_2\\ 
y \to b_1x_1+b_2x_2\end{array} \qquad 
\begin{pmatrix}a_1 & a_2 \\ b_1 & b_2\end{pmatrix} \text{ invertible over
$\ztkz$}
\end{gather*}
The mixed term after this transformation is $2(a_1b_2+a_2b_1)$. 
As $a_1b_2+a_2b_1 \bmod 2$ is the same as the determinant
of the change of variables above i.e., $a_1b_2-a_2b_1$
modulo $2$; it is not possible for a transformation in 
$\gl_2(\ztkz)$ to eliminate the mixed term. 
Instead, one can show that over $\ztkz$ it is possible to get an 
equivalent form where the mixed terms are disjoint i.e., both
$x_ix_j$ and $x_ix_k$ do not appear, where $i,j,k$ are 
pairwise distinct.
One captures this form by the following definition.

\begin{definition}\label{def:BlockDiagonal}
A matrix $\MD^n$ over integers is in a block diagonal form if it is a
direct sum of type I and type II forms; where type I form is an integer
while type II is a matrix of the form 
$\begin{pmatrix}2^{\ell+1}a & 2^{\ell} b \\ 2^{\ell}b & 2^{\ell+1}c\end{pmatrix}$
with $b$ odd.
\end{definition}

The following theorem is folklore and is also implicit in the proof of
Theorem~2 on page~369 in \cite{CS99}. For completeness, we provide a 
proof in Appendix \ref{sec:BlockDiagonal}.

\begin{theorem}\label{thm:BlockDiagonal}
Let $\MQ^n$ be an integral quadratic form, $p$ be a prime, 
and $k$ be a positive integer. 
Then, there is an algorithm that performs $O(n^{1+\omega}\log k)$ ring operations
and produces a matrix $\MU \in \sln(\zpkz)$ such that $\MU'\MQ\MU\pmod{p^k}$, is a
diagonal matrix for odd primes $p$ and a block diagonal matrix
(in the sense of Definition~\ref{def:BlockDiagonal}) for $p=2$.
\end{theorem}

\paragraph{Canonical Forms.}

For a quadratic form $\MQ$ and prime $p$, the set 
$\{\MS^n \mid \MS \eqp \MQ\}$ is the set of 
$p^*$-equivalent forms of $\MQ$ (also called $p^*$-equivalence class of 
$\MQ$).  It is possible to define something called a 
``canonical'' quadratic form for the $p^*$-equivalence class of a given
quadratic form $\MQ$. In
particular, we are interested in 
a function $\canp$ such that for all integral quadratic forms
$\MQ$, $\canp(\MQ) \in \{\MS \mid \MS \eqp \MQ\}$; with the
property that if $\MQ_1 \eqp \MQ_2$ then 
$\canp(\MQ_1) = \canp(\MQ_2)$.
We also consider a related problem of coming up with a 
canonicalization procedure. In particular, we want a polynomial
time algorithm that given $\MQ, p$ and a positive integer $k$,
finds $\MU \in \gln(\zpkz)$ such that 
$\MU'\MQ\MU \equiv \canp(\MQ)\bmod{p^k}$.

It is not difficult to show the existence of a canonical form.
For example, we can go over the $p^*$-equivalence class of
$\MQ$ and output the form which is lexicographically the smallest
one. But, this form gives us no meaningful information about
$\MQ$ or the $p^*$-equivalence class of $\MQ$.

For odd prime $p$,
the $p$-canonical form is implicit in Conway-Sloane 
\cite{CS99} and is also described explicitly by Hartung
\cite{Hartung08, Cassels78}. The canonicalization algorithm in this case
is not complicated and can be claimed to be implicit in 
Cassels \cite{Cassels78}.

The definition of canonical form for the case of prime~2 is 
quite involved and needs careful analysis\footnote{
Cassels (page~117, Section~4, \cite{Cassels78}),
referring to the canonical forms for $p=2$
observes that ``only a masochist is invited to read the rest''.} 
\cite{Jones44, CS99, Watson60}. 
Jones \cite{Jones44} presents the most complete description
of the $2$-canonical form. His method is to come up with a small
$2$-canonical forms and then showing that every quadratic form is
$2^*$-equivalent to one of these. Unfortunately, a few of his 
transformations are existential i.e., he shows that a transformations
with certain properties exists without explicitly finding them.

The following theorems appear in \cite{DH14Can}. 

\begin{theorem}\label{thm:ALG:CanP}
Let $\MQ^n$ be an integral quadratic form, $p$ be a prime and
$k\geq \ordt(\det(\MQ))+k_p$. Then, there is an algorithm (Las Vegas
with constant probability of success for odd primes and deterministic for
the prime~2)
that given $(\MQ^n, p, k)$ performs
$O(n^{1+\omega}\log k+nk^3+n\log p+\log^3p)$ ring operations over $\zpkz$
and outputs $\MU \in \gln(\zpkz)$
such that $\MU'\MQ\MU \equiv \canp(\MQ) \bmod{p^k}$.
\end{theorem}

\begin{theorem}\label{thm:PStarEq}
Let $\MQ^n$ be an integral quadratic form, $p$ 
be a prime and $k = \ordp(\det(\MQ))+k_p$. If $\MD^n$ is a block 
diagonal form which is equivalent
to $\MQ$ over $\zpkz$, then $\MD \eqp \MQ$.
\end{theorem}

\paragraph{Primitive Representations.}
The following theorem gives an algorithmic handle on the 
question of deciding if an integer $t$ has a primitive
$p^*$-representation in $\MQ$. The theorem is implicit
in Siegel \cite{Siegel35} (a proof is provided in
Appendix \ref{sec:Proofs} for completeness).

\begin{theorem}\label{thm:Siegel13}
Let $\MQ^n$ be an integral quadratic form, $t$ be an integer,
$p$ be a prime and $k=\max\{\ordp(\MQ),\ordp(t)\}+k_p$. Then,
if $t$ has a primitive $p^k$-representation in $\MQ$ then 
$t$ has a primitive $p^*$-representation in $\tMQ$ for all
$\tMQ \eqp \MQ$.
\end{theorem}

Next, we give several results from \cite{DH14}.
This paper deals with the following problem. Given a quadratic
form $\MQ$ in $n$-variables, a prime $p$, and integers $k,t$
find a solution of $\Vx'\MQ\Vx\equiv t \bmod{p^k}$, if it exists.
Note that it is easy (i.e., polynomial time tester exists) to test 
if $t$ has a $p^k$-representation in $\MQ$.

\begin{theorem}\label{thm:PolyRep}
Let $\MQ^n$ be an integral quadratic form, $p$ be a prime, 
$k$ be a positive integer, $t$ be an element of $\zpkz$.
Then, there is a 
polynomial time algorithm (Las Vegas for odd primes and
deterministic for the prime~2)
that performs
$O(n^{1+\omega}\log k+nk^3+n\log p)$ ring operations over 
$\zpkz$ and outputs a primitive
 $p^k$-representation of $t$ by $\MQ$,
if such a representation exists. The time complexity can be improved
for the following special cases.
\begin{align*}
\begin{array}{ll}
\text{Type I, $p$ odd} & ~~O(\log k+\log p)\\
\text{Type I, $p=2$} & ~~O(k)\\
\text{Type II} & ~~O(k\log k)
\end{array}
\end{align*}
\end{theorem}

Next, we give necessary and sufficient conditions for a Type II
block to represent an integer $t$. A proof of this result can
also be found in \cite{DH14}.

\begin{lemma}\label{lem:RepresentTTypeII}
Let $\MQ=\begin{pmatrix}2a & b \\ b & 2c \end{pmatrix}$, $b$ odd 
be a type II block, and $t, k$ be positive integers. Then, $\MQ$ represents
$t$ primitively over $\ztkz$ if $\ordt(t)=1$.
\end{lemma}

\paragraph{Additional Notation.}
For convenience, we introduce the following notations, where
$q$ is a positive integer. 
\begin{align}\label{not:kp}
k_p = \left\{\begin{array}{ll} 3 & \text{if $p=2$, and}\\
1 & \text{$p$ odd prime.}
\end{array}\right.
&\qquad 
\overline{q}= q\prod_{p|2q}p^{k_p} \\
\bbP_q &= \{p \mid \ordp(2q)>0\} \nonumber\\
\SGNI &= \{1,3,5,7\} \nonumber\\
\MTP = \begin{pmatrix}2&1\\1&4\end{pmatrix}
&\qquad
\MTM = \begin{pmatrix}2&1\\1&2\end{pmatrix}\nonumber
\end{align}

\section{Technical Overview}

In this section, we give an overview of our algorithm.
Note that this algorithm does not run in polynomial
time. The final version of the algorithm, which is correct and runs
in polynomial time will be presented in Section \ref{sec:QFGen}.

Before describing the algorithm, we need to describe the input
to the algorithm. 
The question of succinct specification of a genus has several 
competing answers \cite{Watson76,Kitaoka99,CS99,OMeara73}. 
In this work, we use the specification by Conway-Sloane, called
the {\em symbol}; with the property that two 
quadratic forms are in the same genus iff they have the same symbol. 
An intuitive overview of the symbol follows (see Section \ref{sec:Symbol}
for the formal version).

It can be shown that two quadratic forms $\MQ_1^n$
and $\MQ_2^n$ are in the same genus iff : (i) $\det(\MQ_1)=\det(\MQ_2)$,
(ii) $\sig(\MQ_1)=\sig(\MQ_2)$
and (iii) $\MQ_1\eqp \MQ_2$, for every prime $p$ that divides 
$2\det(\MQ_1)$. The set of primes 
$\PSym=\{p\mid \ordp(2\det(\Sym)) > 0\}$ is called the set of 
{\em relevant primes} for the genus $\Sym$.
The question now reduces to finding the necessary and sufficient
condition for $p^*$-equivalence. It can be shown that two 
quadratic forms $\MQ_1^n$ and $\MQ_2^n$ are $p^*$-equivalent iff
(a) $\ordp(\det(\MQ_1))=\ordp(\det(\MQ_2))$, and 
(b) $\MQ_1 \overset{p^k}{\sim} \MQ_2$, for 
$k=\ordp(\det(\MQ_1))+k_p$. These two conditions can be written
in an equivalent way, using what is called the $p$-symbol of a 
quadratic form.  The key property is that two quadratic forms are 
$p^*$-equivalent iff they have the same $p$-symbol. For now, we
can think of $\canp(\MQ)$ as the $p$-symbol of $\MQ$. By definition
of $p$-canonical forms, it follows that $\MQ_1 \eqp \MQ_2$ iff
$\canp(\MQ_1)=\canp(\MQ_2)$. 

\paragraph{Symbol.} One can now give an informal description of
the symbol. Intuitively, think of $\SYM_p(\MQ)$ as the tuple 
$(p,\canp(\MQ))$. Then, by the definition of canonical forms, two 
quadratic form are $p^*$-equivalent iff they have the same $p$-symbol.
 The symbol of a 
quadratic form $\MQ$ is the list of tuples
$(p,\canp(\MQ))$, one for each prime that divides $2\det(\MQ)$
along with its signature 
i.e., 
\begin{gather}\label{EQ:SYM}
\SYM(\MQ) = \{(-1, \sig(\MQ))\}
\bigcup_{p| 2\det(\MQ)} \{ (p,\can_p(\MQ))\}
\end{gather}
Note that the determinant of $\MQ$ is missing from the symbol. This is 
because it is possible to find the determinant from the symbol $\SYM(\MQ)$.
This is done by calculating $\ordp(\det(\MQ))$ for each prime $p$.
The relevant primes of $\SYM(\MQ)$ can be read from the first component
in of the tuples in $\SYM(\MQ)$. And, for a relevant prime $p$ it 
can be shown that $\ordp(\det(\MQ))=\ordp(\det(\canp(\MQ)))$.

The notation $\Sym^n$ will denote both 
a genus and the symbol of the
genus, depending on the context. The notation $\Sym_p$ will denote the 
$p$-symbol of the genus $\Sym$.
The set of {\em relevant primes} for the symbol $\Sym$ is denoted by 
$\PSym=\{p \mid \ordp(2\dSym) > 0\}$. 
If $\SYM_p(\MQ)$ is the $p$-symbol of $\MQ$ then
$\scalep(\MQ)$ denotes the set of $p$-scales of $\MQ$. 
The size of the symbol of a genus $\Sym^n$ of determinant $d$ 
is $O(n|\PSym|\log d)$.

\paragraph*{Local Form.} Given a symbol $\Sym^n$ and a positive 
integer $q$, it is possible to construct a quadratic form $\MS^n$
such that $\MS \eqq \MQ$, for every $\MQ \in \Sym$; in 
$\poly(n,\log\det(\Sym), \log q)$ time. Such a form is said to be locally
equivalent to the genus $\Sym$ over $\zqz$ and is denoted by 
$\MS \eqq \Sym$. The local form satisfies the following important 
property. For every prime $p$ that divides $q$, 
$\ordp(\det(\MS))=\ordp(\det(\Sym))$. Note that $\MS$ does not need to 
have determinant $\det(\Sym)$ and may not be equivalent to $\Sym$
over $\bbR$. In particular, $\MS \not\in \Sym$.

\paragraph{Primitive Representation.}
If $t$ has a primitive $p^*$-representation in $\MQ$ then, 
by Theorem \ref{thm:Siegel13}, $t$ has a 
primitive $p^*$-representation in every 
quadratic form in the genus $\gen(\MQ)$. Hence, the primitive
representativeness of an integer $t$ by a quadratic form $\MQ$ 
only depends on the symbol $\SYM(\MQ)$. 

\begin{definition}\label{def:PrimRepInGenus}
We say that an integer $t$ has a primitive representation in a genus
$\Sym$ if $t$ has a primitive $p^*$-representation in $\Sym$ for all
$p \in \{-1,2\} \cup \bbP$.
\end{definition}

\paragraph*{Simple Version of the Algorithm.}
Let $\Sym^n$ be a symbol of non-empty genus. In this section, 
we give the simple version of the algorithm. 
The run time of the algorithm is not polynomial; mainly because the
exponential blowup in the determinant after each recursive step. 
In Section \ref{sec:QFGen}, we show that by carefully selecting 
the embedding $\Vx$ of $t$ and by
simplifying the input before each recursive call, it is
possible to avoid the blowup and show a polynomial bound on the
runtime. 

\vspace*{10pt}
\noindent
\textsc{GenSimple} ({\em input:} $\Sym^n$)
{\em output:} $\MQ^n \in \Sym$
\begin{enumerate}[1.]
\item If $n=1$ then $\MQ=\dSym$.
\item Let $t$ be an integer such that $t$ has a primitive representation 
in $\Sym$. 
\item Let $q=\overline{t^{n-1}\dSym}$. Find $\MS$ such that 
$\MS\eqq \Sym$.
\item Find a primitive $q$-representation $\Vx$ such that $\Vx'\MS\Vx\equiv t\bmod{q}$.
\item Extend $\Vx$ by $\MA$ so that $[\Vx,\MA]\in \gln(\zqz)$.
\item Compute the following quantities.
\begin{align}\label{ALGO:dH}
\begin{array}{ll}
\Vd:=\Vx'\MS\MA \bmod q & \MH := (t\MA'\MS\MA - \Vd'\Vd) \bmod q\\
\end{array}
\end{align}
\item Define the symbol $\tSym^{n-1}$ as follows.
\begin{align}\label{def:tSym}
\tSym_p = \left\{ \begin{array}{ll}
\sig(t)\left(\sig(\Sym)-\sig(t)\right) & \text{if }p=-1\\
\SYM_p(\MH) & \text{if $p$ divides $q$}\\
\SYM_p(\MI^{n-2} \oplus t^{n-2}\dSym) & \text{otherwise}
\end{array}\right.
\end{align}
\item Let $\tMH =$ \textsc{GenSimple}$(\tSym)$.
\item Find $\tMU \in \gl_{n-1}(\zqz)$ such that $\tMH\equiv \tMU'\MH\tMU \bmod q$.
\item Output $
\MQ = \begin{pmatrix}t & \Vd\tMU \\ 
(\Vd\tMU)' & \frac{\tMH+\tMU'\Vd'\Vd\tMU}{t}\end{pmatrix}
$, where the division in the lower right is (usual) rational division.
\end{enumerate}

\subsection{Intuitive Description of the Algorithm}
For simplicity, we assume that $\sig_{(-1)}(\Sym)=n$ i.e., the genus
$\Sym$ is the genus of positive definite quadratic forms or
Gram matrix of lattices.

In order to find a Gram matrix $\MQ$ in the genus $\Sym$, we
start by finding a value $t$ such that $\MQ$ has the following form.
\begin{align}\label{EQ:MQForm}
\MQ = \begin{pmatrix}t & \Vw \\ \Vw' & \tMQ\end{pmatrix} 
\end{align}
It turns out that it suffices to find an integer $t$ which has a 
primitive representation in the genus $\Sym$. The next step
is best explained by thinking in terms of lattices.

The Gram matrix $\MQ$ equals $\MB'\MB$, where 
$\MB=[\Vb_1,\ldots,\Vb_n]$ is a basis of a lattice with Gram matrix
$\MQ$. Because $\MQ$ has the form given in Equation \ref{EQ:MQForm},
$\Vb_1'\Vb_1 = t$. Consider now the (possibly non-integral) lattice
one obtains by projecting $[\Vb_1,\ldots,\Vb_n]$ onto the subspace
orthogonal to $\Vb_1$. 
It is possible to show that the Gram matrix of this lattice in
$(n-1)$-dimensions is given by $\tMQ-\frac{\Vw'\Vw}{t}$.

The matrix $\tMQ$, and the matrix $\Vw'\Vw$ are integral. Thus,
$t\tMQ-\Vw'\Vw$ is a Gram matrix of an integral lattice.
To find $\MQ$, we therefore (a) find the symbol of the lattice
$t\tMQ - \Vw'\Vw$ and recursively find a corresponding lattice, 
and (b) find $\Vw$.
To solve (a), the algorithm above constructs a locally equivalent quadratic
form $\MS$, then finds a representation $\Vx$ of $t$ into $\MS$, and transforms
$\MS$ with a transformation $[\Vx,\MA] \in \gln(\zqz)$ 
which maps $t$ into the top left
\begin{align}
[\Vx,\MA]' \MS [\Vx,\MA] = \begin{pmatrix}
t & \Vd \\ \Vd' & \MA'\MS\MA 
\end{pmatrix}\bmod q
\end{align}
To solve (b), it is possible to show that one can recover $\Vw$
from the vector $\Vd$.

Finally, one can show that because of the way we chose $q$ 
in the algorithm, the expression 
$\frac{\tMH+\tMU'\Vd'\Vd\tMU}{t}$ in the construction of $\MQ$
is actually integral.

\subsection{Comparisons to Hartung's Algorithm.}\label{sec:Hartung}

Our construction in Section \ref{sec:QFGen} is 
similar to the algorithm given by Hartung \cite{Hartung08}.
This algorithm, as we will show, is not polynomial but $O(n^n)$.
There are other severe problems with Hartung's work
(i) several lemmas are incorrect because of insufficient care
while handling prime~2, and (ii) the construction of $t$, in case
of dimension~2,
is short but unfortunately incorrect.
This construction takes us several pages 
(page~\pageref{dim2}-\pageref{dim2end}).
A detailed discussion of the comparison follows.

The first non-trivial step is to construct an integer $t$
which is primitively representatable in the genus $\Sym$. 
Hartung constructs a $t$ such that $t=\wp s$, where $s$
divides $\dSym$ and $\wp$ is an prime which does not
divide $\dSym$. This construction seems correct when 
$n\geq 3$ but incorrect for $n=2$. One of the reasons is
the treatment of the prime~2; which is not thorough. The
prime~2 has been known to create problems, if not handled
correctly \cite{Pall65, Minkowski10, Watson76}. 

For example, Lemma 3.3.1 \cite{Hartung08} is incorrect
for $p=2$ because it is not possible to divide by~2 over
the ring $\ztkz$ at the end of the proof. This leads to an
easy counter example for Lemma 3.3.2, which claims that
a quadratic form $\MQ^{n\geq 3}$ with 
$\det(\MQ) \in (\zpz)^\times$ represents every integer $t$
primitively over $\zpkz$ for all positive integers $k$. A 
counter example is $(\MQ=x^2+y^2+z^2, p=2, k=3, t=7)$. By
exhaustive search, it can be verified that $x^2+y^2+z^2$
does not represent $7$ modulo $8$. This mistake becomes
more severe in the construction of $t$ for $\Sym^{n=2}$.
The construction in this case is
highly non-trivial and needs a separate treatment
(page~\pageref{dim2}-\pageref{dim2end}).

The construction of $t$ for $\Sym^{n\geq 3}$ seems to be 
correct \cite{Hartung08}. 
Our construction, though, gives smaller $t$. Hartung 
needs a prime $\wp$ which does not divide the determinant, 
each time he needs to find a primitively representable integer $t$. 
In contrast, we need
such a prime only once. Note that the construction of such a prime
$\wp$ takes polynomial time if ERH holds.

The most serious issue is that the algorithm by Hartung is not polynomial
time. He argues that each step in the algorithm is polynomial
and because each recursive step reduces the dimension by~1, overall
the algorithm is polynomial. This is not true because 
after finding
$t$ and reducing to one less dimension to a symbol
$\tSym^{n-1}$, 
$\det(\tSym)=t^{n-2}\dSym$ (see Claim \ref{CALG:dtSym}).
The upper found on $t$ is $\dSym$ and so
$\det(\tSym)$ can be as large as $\dSym^{n-1}$;
leading to a blowup $\sim \dSym^{n^n}$ if used $n$ times
recursively. As it is, the time complexity of Hartung's 
algorithm is proportional to $O(n^n)$.
In contrast, our construction represents $t$ in a specific way and uses the
property of this representation to show that the determinant 
blows up by $2^{n^2}$ at most; resulting in a 
polynomial time algorithm (see Section \ref{sec:QFGenPoly}).

\section{Formalizing the Input}

This section describes the input to the algorithm which boils 
down to the question of a succinct representation of the genus.

\subsection{Symbol of a Quadratic Form}\label{sec:Symbol}

There are several equivalent ways of giving a description of 
the $p^*$-equivalence \cite{CS99, Kitaoka99, OMeara73, Cassels78}. 
In this work, we go with a modified version of the Conway-Sloane description,
called the $p$-{\em symbol} of a quadratic form. Our modification
gets rid of the need to use the $p$-adic numbers.
Note that $p$-adic numbers are a 
staple in this area and we are not aware of any work which does not
use them \cite{Kitaoka99, OMeara73, Siegel35}.

The $(-1)$-symbol of a quadratic form is equal to the 
signature of the quadratic form.

\subsubsection{$p$-symbol, $p$ odd prime}\label{sec:PSym}

Let $k=\ordp(\det(\MQ))+1$ and $\MD$ be the diagonal quadratic form 
which is $p^k$-equivalent to $\MQ$ (see Theorem \ref{thm:BlockDiagonal}).
Then, $\MD$ can be written as follows.
\begin{align}\label{def:jordanp}
\MD=\MD_0^{n_0} \oplus p\MD_1^{n_1} \ldots \oplus p^i \MD_i^{n_i} \oplus \ldots 
\qquad i \leq \ordp(\det(\MQ))\;,
\end{align}
where $\MD_0,\ldots,\MD_{k-1}$ are diagonal quadratic forms, 
$\sum_i n_i = n$ and $p$ does not divide 
$\det(\MD_0)\ldots\det(\MD_{k-1})$. Let $\scalep(\MQ)$ is the set of
$p$-orders $i$ with non-zero $n_i$.
Then, the $p$-symbol of $\MQ$ is defined as the
set of {\em scales} $i$ occurring in Equation \ref{def:jordanp} with non-zero
$n_i$, {\em dimensions} $n_{i}=\dim(\MD_i)$ and {\em signs} 
$\epsilon_i=\legendre{\det(\MD_i)}{p}$. 
\begin{equation}\label{eq:PSymbol}
\SYM_p(\MQ)= \left\{(p, i,\legendre{\det(\MD_i)}{p},n_i) \mid 
i \in \scalep(\MQ)\right\}
\end{equation}

The following fundamental result follows from Theorem~9, page~379 
\cite{CS99} and Theorem \ref{thm:PStarEq}. 

\begin{theorem}\label{thm:PSymbol}
For $p \in \{-1\} \cup \bbP$, two quadratic forms 
are $p^*$-equivalent iff they have the same $p$-symbol.
\end{theorem}

\subsubsection{$2$-symbol}\label{sec:2Sym}

Let $k=\ordt(\det(\MQ))+3$ and $\MD$ be the block diagonal form 
which is $2^k$-equivalent to $\MQ$ (Theorem \ref{thm:BlockDiagonal}).
Then, $\MD$ can be written as follows.
\begin{align}\label{def:jordant}
\MD=\MD_0^{n_0} \oplus 2\MD_1^{n_1} \ldots \oplus 2^i \MD_i^{n_i} 
\oplus \ldots \qquad i \leq \ordt(\det(\MQ))\;,
\end{align}
where $\det(\MD_0), \ldots, \det(\MD_i), \ldots$ are odd, $\sum_i n_i=n$
and each $\MD_i$ is in block diagonal form according to Definition 
\ref{def:BlockDiagonal}.
The $2$-symbol of $2^i\MD_i$ are the following quantities.
\begin{equation}\label{def:2invariants}
\left(\begin{array}{lll}
i   & \text{{\em scale} of $\MD_i$} & \\
n_i= \dim(\MD_i) & \text{{\em dimension} of $\MD_i$} &\\
\epsilon_i = \legendre{\det(\MD_i)}{2} & \text{{\em sign} of $\MD_i$} & \\
\type_i=\text{I  or  II} & \text{{\em type} of $\MD_i$} & \text{I, iff 
there is an odd entry on} \\ && \text{~the main diagonal of the }\\
&&\text{~matrix $\MD_i$}\\
\odty_i\in\{0,\ldots,7\} & \text{{\em oddity} of $\MD_i$} & 
\text{it $\type_i=$I, then it 
is equal to} \\ && \text{~the trace of $\MD_i$ read modulo $8$,} \\ 
&& \text{~and is $0$ otherwise.}
\end{array}\right)
\end{equation}
Let the set of scales $i$, with non-zero $n_i$, be denoted 
$\scalet(\MQ)$.
Then, the $2$-symbol of $\MQ$ is written as follows.
\begin{equation}\label{eq:2sym}
\SYM_2(\MQ)= \left\{(2,i,\epsilon_i, n_i, \type_i, \odty_i) 
\mid i \in \scalet(\MQ)\right\}
\end{equation}

In contrast to the $p \in \{-1\} \cup \bbP$ case, 
two $2^*$-equivalent quadratic forms may produce two different 
$2$-symbols. These symbols are then said to be $2$-equivalent.

\bigskip

Consider the following useful generalization of the function $p$-order. 

\begin{definition}\label{def:porder}
Let $p \in \{2\}\cup \bbP$ be a prime, and $\Sym$ be a symbol with
$\SymP = \{(p,i,n_i,\epsilon_i, *, *)\}$, where $*$ is
empty in case $p$ is odd. Then, $\ordp(\Sym)$ is defined as 
$\argmax_i \{i \in \scale_p(\Sym)\}$.
\end{definition}

\subsection{Reduced Symbol}

The set $\cup_{p \in \{-1,2\} \cup \bbP} \SymP$, where $\SymP$ is a 
$p$-symbol, is a complete description of a genus and can be used as
an input to the algorithm. Unfortunately, 
this description is too long because there are infinitely many primes. The
following lemma helps us in giving a shorter description.

\begin{lemma}\label{lem:RelevantPrimes}
Let $\MQ^n$ be a lattice with determinant $d$ 
and $p$ be an odd prime that does not divide $d$. Then,
$\MQ \eqp d\oplus\MI^{n-1}$.
\end{lemma}
\begin{proof}
Let $p$ be an odd prime that does not divide $d$ and 
$\MD=d_1\oplus \cdots \oplus d_n$ be the diagonal matrix which is 
equivalent to $\MQ$ over $\zpz$. Then, $p$ does not divide
$d_1\cdots d_n$ and
\[
\SYM_p(\MQ) = \left\{ \left(p, 0, n, \legendre{\det(\MD)}{p}\right)\right\}
\]
By the definition of the $p$-symbol, $\MD \eqp \MQ$. It follows that there is a
$\MU \in \gln(\zpz)$ such that $\MD\equiv \MU'\MQ\MU \bmod p$.
But then, $\det(\MD)\equiv\det(\MQ)\det(\MU)^2 \bmod{p}$ and 
\[
\legendre{\copp(\det(\MD))}{p} = \legendre{\copp(\det(\MQ)\det(\MU)^2)}{p}
=\legendre{\copp(d)}{p}.
\]
This implies that $d\oplus\MI^{n-1}$ has the same $p$-symbol as
$\MQ$, completing the proof (Theorem \ref{thm:PSymbol}).
\end{proof}

Our input to the algorithm is the symbol of a genus, defined as
follows.

\begin{definition}\label{def:Symbol}
Let $\MQ$ be an integral quadratic form from a given genus.
Then, the symbol the genus is defined as the set
\[
\bigcup_{p \in \{-1\} \cup 
\{p \mid \ordp(2\det(\MQ)) > 0\}} \SYM_p(\MQ)
\]
\end{definition}

From Theorem \ref{thm:PSymbol}, it follows that two quadratic forms are
in the same genus if they are $2^*$-equivalent and have the same 
$p$-symbol for each $p \in \{-1\}\cup \bbP$.

The following theorem is a direct implication of 
\cite[Theorem~9, ][page~379]{CS99}, \cite[Theorem~10, ][page~381]{CS99},
and the definition of the symbol.

\begin{theorem}\label{thm:Symbol}
Let $\MQ_1^n$ and $\MQ_2^n$ be two integral quadratic forms.
Then, the following statements are equivalent.
\vspace{-10pt}
\begin{align*}
\begin{array}{l}
\text{(a)~}\MQ_1 \in \gen(\MQ_2), ~~~~\qquad
\text{(b)~} \det(\MQ_1)=\det(\MQ_2)(=d), 
\MQ_1 \overset{\overline{d}}{\sim} \MQ_2, \text{ and } 
\MQ_1 \overset{\bbR}{\sim} \MQ_2\\
~~~~\text{(c)~}\MQ_1 \eqp \MQ_2, \forall~p \in 
\{-1\}\cup\{p\mid \ordp(2\det(\MQ_1\MQ_2))> 0 \}.\\
\end{array}
\end{align*}
\end{theorem}

Note that Theorem \ref{thm:Symbol} implies that every quadratic
form in a particular genus has the same determinant (for a proof see
page~139, Lemma 4.1, \cite{Cassels78}). The determinant of a genus
can be computed from its symbol.

We simplify the input description further by introducing the notion
of reduced symbol.

\begin{definition}\label{def:ReducedGenus}
A symbol $\Sym$ is {\em reduced} if for every relevant prime 
$p \in \PSym$ the $p$-scale~0 appears in $\scalep(\Sym)$.
\end{definition}

If $i_p = \min\{i \in \scalep(\Sym)\}$ then $\Sym$ is 
$p^{i_p}$-equivalent to a matrix which is identically~0. Thus, if
$\MD$ is a quadratic form with $\MD \eqp \Sym$ then every entry of
$\MD$ is divisible by $p^{i_p}$. Given a genus $\Sym$, we define
the following quantity.
\begin{align}\label{def:gcd}
\gcd(\Sym) = \underset{p \in \PSym}{\prod} p^{i_p} \text{ where, }
i_p = \min\{i \in \scalep(\Sym)\}
\end{align}
The reduced symbol corresponding to the symbol $\Sym$ can now be defined
as follows.
\begin{align}\label{def:gcdSymbol}
\red(\Sym) = \left\{\SYM_p\left(\frac{\MD}{p^{\min\{i \in \scalep(\Sym)\}}}\right)\mid p \in \PSym, \MD \eqp \SymP\right\} \cup \{\SYM_{(-1)}(\Sym)\}
\end{align}

To find a quadratic form from a genus $\Sym$, it suffices to 
find a quadratic form in genus $\red(\Sym)$, the proof of which
is as follows.

\begin{lemma}\label{lem:ReducedGenus}
Let $\Sym$ be a genus, and for each prime $p \in \PSym$,
$i_p$ be the integer $\min\{i \in \scalep(\Sym)\}$.
Then, $\MQ \in \red(\Sym)$ iff 
$\gcd(\Sym)\MQ \in \Sym$.
\end{lemma}
\begin{proof}
Let $p \in \PSym$ be a prime. If $\MS$ is a quadratic form such that
$\SYM(\MS)=\Sym$ then every entry of $\MS$ is divisible by $p^{i_p}$.
We define a quadratic form $\MQ$ as follows.
\[
\MQ=\frac{\MS}{\prod_{p\in\PSym}p^{i_p}}
\]
By definition of $\red(\Sym)$, $\SYM_p(\MQ)=\SYM_p(\red(\Sym))$
for all $p \in \{-1\} \cup \PSym$. Thus, $\MQ \in \red(\Sym)$.

Conversely, if $\MQ \in \red(\Sym)$ then $p^{i_p}\MQ$ has
the same $p$-symbol as $\SymP$. Thus, 
$\left(\prod_{p\in\PSym}p^{i_p}\right)\MQ$ has symbol $\Sym$.
\end{proof}

\subsection{Valid Symbol}

We now define the following
three conditions on the symbol $\Sym$.

\begin{description}
\item[Determinant Condition.] For every prime $p \in \{2\} \cup \bbP$
such that 
$\SymP=\left\{(p,i,\epsilon_i,n_i, *, *) \mid i \in \scalep(\Sym)\right\}$, 
where $*$ is empty for odd primes; 
\begin{equation}\label{Cond:det}
\legendre{\copp(\dSym)}{p} = \prod_{i\in \scalep(\Sym)} \epsilon_i
\end{equation}

\item[Oddity Condition.] The symbol $\Sym$ satisfies the oddity
equation i.e.,
\begin{equation}\label{Cond:Sum}
\sig(\Sym) + \sum_{p \in \bbP} \pexcess(\Sym) \equiv 
\odty(\Sym) \bmod 8
\end{equation}

\item[Jordan Condition.] Let $p$ be an odd prime and 
$\SymP=\left\{(p,i,\epsilon_i,n_i) \mid i \in \scalep(\Sym)\right\}$,
then for each Jordan constituent $(p^i,\epsilon_i,n_i)$, we must have
\begin{equation}\label{Cond:Odd}
\text{if $n_i=0$ or $p=-1$ then $\epsilon=+$}
\end{equation}
For $p=2$, let $\SYM_2(\MQ)= \left\{(2,i,\epsilon_i, n_i, \type_i, s_i) 
\mid i \in \scalet(\Sym)\right\}$, then $\Sym$ satisfies the following
conditions.
\begin{equation}\label{Cond:2}
\begin{array}{l}
\text{for $n_i=0$, $\type_i=$II and $\epsilon_i=+$}\\
\text{for $n_i=1$,} \left\{
	\begin{array}{l}
	\epsilon_i=+ \implies s_i \equiv \pm 1 \bmod 8\\
	\epsilon_i=- \implies s_i \equiv \pm 3 \bmod 8\\
	\end{array}\right. \\
\text{for $n_i=2$,$\type_i=$I} \left\{
	\begin{array}{l}
	\epsilon_i=+ \implies s_i \equiv 0 \text{ or }\pm 2 \bmod 8\\
	\epsilon_i=- \implies s_i \equiv 4 \text{ or }\pm 2 \bmod 8\\
	\end{array}\right. \\
\end{array}
\end{equation}
\end{description}

The set of conditions are taken
from \cite[page~382-383]{CS99}. A symbol $\Sym$ which satisfies these
three conditions will be called {\em valid}.

\section{$q$-equivalent forms, $q$ composite}

Given a valid symbol $\Sym^n$, it is useful to construct a quadratic
form $\MQ^n$ which is $q$-equivalent to $\Sym$ for a given positive
integer $q$ (see Step~3 of \textsc{GenSimple}). 

The following is a helper lemma which shows how to construct a 
quadratic form $\MQ$ such that $\MQ \eqp \Sym$.

\begin{lemma}\label{lem:LocalQF}
There exists a randomized algorithm that takes a symbol $\Sym^n$
of determinant $d$, and a prime $p$ as input; performs
$O(n+\log^3p)$ ring operations over $\bbZ/p^{\ordp(d)+k_p}\bbZ$; and outputs a
block diagonal quadratic form $\MQ^n$ such that 
$\MQ \eqp \Sym$.
\end{lemma}
\begin{proof}
There are three different constructions: for the prime $2$, for relevant
odd prime and for the odd prime that does not divide $\dSym$.
\begin{enumerate}[(1.)]
\item The first and simplest construction deals with odd primes $p$
that do not divide $\dSym$. By Lemma \ref{lem:RelevantPrimes}, 
$\SymP\eqp \dSym\oplus \MI^{n-1}$.
Hence, we set $\MQ=\MI^{n-1}\oplus\dSym$.

\item The second type of primes are odd primes that divide $\dSym$.
Let $\SymP=\{(p,i,\epsilon_i,n_i)\mid i \in \scalep(\Sym)\}$.
We use rejection sampling to find a quadratic non-residue modulo $p$,
say $\tau_p$. Note that generating a random non-zero element from 
$\zpz$ yields a quadratic non-residue with probability $1/2$. 
The matrix $\MQ$, in this case, is generated as follows.
\begin{align}\label{EQ:PDividesDetSym}
\MQ = \underset{i \in \scalep(\Sym)}{\oplus}p^i\MD_i \qquad
\MD_i = \left\{\begin{array}{ll}
\MI^{n_i} & \text{if $\epsilon_i=1$} \\
\MI^{n_i-1}\oplus \tau_p & \text{otherwise}
\end{array}\right.
\end{align}
\item The only remaining case is of the prime~2. 
Let $\SymT=
\{(2,i,\epsilon_i,n_i, \type_i, \odty_i)\mid i \in \scalet(\Sym)\}$.
Then, the
quadratic form $\MQ$ is defined as
$\MQ=\oplus_{i\in\scalet(\Sym)}2^i\MD_i^{n_i}$, where $\MD_i$
is defined as follows.
\begin{align}\label{EQ:MQ2}
\MD_i^{n_i} &= \left\{\begin{array}{ll}
\underbrace{\MTP\oplus\cdots\oplus\MTP}_{n_i/2-1}\oplus\MTM
& \text{if $\epsilon_i=-1, \odty_i=\text{II}$}\\
\underbrace{\MTP\oplus\cdots\oplus\MTP}_{n_i/2}
& \text{if $\epsilon_i=1, \odty_i=\text{II}$}\\
\MI^{n_i-3}\oplus\MD^{n=3} & \text{$\odty_i\in\{0,\cdots,7\}, n_i>3$}
\end{array}\right.
\end{align}
If $n_i=1$ then $\MD_i$
has to be equal to $\odty_i$. For $n_i=2$, we exhaustively
list all possible Type I forms in Table \ref{tab:dim2}. We observe
that two situations are not possible: $\epsilon=+, \odty=4$
and $\epsilon=-, \odty=0$. For $n_i=3$, we list forms
for all possible choices of $\epsilon$ and $\odty$.

\begin{table}
\caption{Exhaustive List of Type I forms for $n=2$}
\centering
\begin{tabular}{| l | l | l |}
\hline
Form  & $\epsilon$ & $\odty$ \\
\hline
$1\oplus 7, 3\oplus 5$ & $+$ & $0$\\
$1\oplus 1, 5\oplus 5$ &  $+$ & $2$\\
$3\oplus 3, 7\oplus 7$ &  $+$ & $6$\\
\cline{1-3}
$3\oplus 7$ 			&  $-$ & $2$\\
$1\oplus 3, 5\oplus 7$ & $-$ & $4$\\
$1\oplus 5$ 			& $-$ & $6$\\
\hline
\end{tabular}
\label{tab:dim2}
\end{table}

The $\MD$ in Equation \ref{EQ:MQ2} is defined as follows. Suppose
we are looking for a type I form in dimension $n_i>3$ with 
$\odty_i \in \{0,\cdots,7\}$ and Legendre-Jacobi symbol $\epsilon_i$. 
In this case, we choose $\MD$ as
the form in Table \ref{tab:dim3} with $\odty=\odty_i-(n_i-3) \bmod 8$
and $\epsilon=\epsilon_i$. 

\begin{table}
\caption{List of Type I forms for $n=3$}
\centering
\begin{tabular}{| l | l | l |}
\hline
$\epsilon$ & $\odty$ & Form\\
\hline
$+$ & $1$ & $1\oplus1\oplus7$\\
& $3$ & $1\oplus1\oplus1$\\
& $5$ & $7\oplus7\oplus7$\\
& $7$ & $1\oplus7\oplus7$\\
\cline{1-3}
$-$ & $1$ &$3\oplus3\oplus3$\\
& $3$ & $3\oplus3\oplus5$\\
& $5$ & $1\oplus1\oplus3$\\
& $7$ & $1\oplus1\oplus5$\\
\hline
\end{tabular}
\label{tab:dim3}
\end{table}

By construction $\SYM_2(\MQ)=\SymT$.
\end{enumerate}

The algorithm needs to generate a quadratic non-residue modulo
$p$ and hence performs $O(n+\log^3p)$ ring operations.
\end{proof}

\begin{theorem}\label{thm:LocalQF}
Let $\Sym^n$ be a symbol and $q$ be a composite integer. Then,
there is a randomized $\poly(n,\log q)$ algorithm that takes
$(\Sym,q)$, along with a factorization of $q$ as input; and
produces a quadratic form $\MQ$ such that $\MQ \eqq \Sym$.
\end{theorem}
\begin{proof}
For each $p\in \PSym$, we use Lemma \ref{lem:LocalQF} to generate
$\MQ_p$ such that $\MQ_p \eqp \Sym$. We now solve the following
system of congruences using the Chinese Remainder Theorem.
\[
\MQ \equiv \MQ_p \pmod{p^{\ordp(q)}}\qquad p \in \bbP_q
\]
By construction, $\MQ \eqq \Sym$. The algorithm runs in time
$\poly(n, \log q)$
\end{proof}

\section{Existence of a Quadratic Form with a given Symbol}\label{sec:QFGen}

In this section, we answer the following question.
Given a symbol $\Sym^n$, how does one verify that the genus 
corresponding to $\Sym$
is non-empty i.e., there exists a quadratic form $\MQ^n$ such that
$\SYM(\MQ)=\Sym$.

\begin{theorem}\label{thm:QFGen}
Let $\Sym^n$ be a valid symbol (i.e., satisfies the determinant, 
oddity and the Jordan conditions); then
there exists an integral quadratic form $\MQ$ such that $\SYM(\MQ)=\Sym$.
\end{theorem}

This is a well known theorem \cite[Theorem~11, ][page~383]{CS99}.
A proof can also be found in \cite{OMeara73}. Our work, not only
shows the existence, but also generates a form in polynomial time.
In this section, we show that the algorithm \textsc{GenSimple},
if successful, generates a quadratic form from the correct genus.

\begin{proof}(Theorem \ref{thm:QFGen})
Let $\Sym^n$ be a valid input symbol. Run the \textsc{GenSimple} 
algorithm on $\Sym$. We show that {\em if} the algorithm outputs a
quadratic form then it must be from the genus $\Sym$.

We prove several claims regarding
the matrices constructed during the algorithm \textsc{GenSimple}.

\begin{claim}\label{CALG:dtSym}
The determinant of the genus $\tSym$ is $t^{n-2}\dSym$. Also, 
$t^{n-2}\dSym$ divides $\det(\MH)$.
\end{claim}
\begin{proof}
Recall that $\MS$ is a
quadratic form which is equivalent to $\Sym$ over $\zqz$. 
Let $\Vd$ be the row vector and $\MH$ be the matrix defined in Equation
\ref{ALGO:dH}. Note that all entries of these matrices i.e., $\Vd, \MH$ 
are integers. Define the matrix 
$\MM=\begin{pmatrix}t & \Vd \\ \Vd' & (\MH+\Vd'\Vd)/t\end{pmatrix}$.
By definition, $\MH+\Vd'\Vd \equiv t\MA'\MS\MA \bmod q$. The integer
$t$ divides $q$. But then, each entry in the matrix $\MH+\Vd'\Vd$
is divisible by $t$ i.e., $(\MH+\Vd'\Vd)/t$ is a matrix over integers.
Thus, $\MM$ is a matrix over integers and the following equality
implies that $\det(\MM)=\det(\MH)/t^{n-2}$.
\begin{align}\label{ALG:M}
\MM=\begin{pmatrix}t & \Vd \\ \Vd' & (\MH+\Vd'\Vd)/t\end{pmatrix}=
\begin{pmatrix}1 & \Vd/t \\ 0 & \MI\end{pmatrix}'
\begin{pmatrix}t & 0 \\ 0 & \MH/t\end{pmatrix}
\begin{pmatrix}1 & \Vd/t \\ 0 & \MI\end{pmatrix}
\end{align}
From Equation \ref{ALGO:dH}, it follows that 
$(\MH+\Vd'\Vd)/t \equiv \MA'\MS\MA \bmod{q/t}$. By definition,
$\MM \equiv (\Vx~\MA)'\MS(\Vx~\MA) \bmod{q/t}$. Hence,
\begin{align}\label{EQALG:dS}
\det((\Vx~\MA)'\MS(\Vx~\MA)) \equiv  \det(\MM) \bmod{q/t} 
= \frac{\det(\MH)}{t^{n-2}} \bmod{q/t}\nonumber\\
t^{n-2}\det(\MS)\det(\Vx~\MA)^2 \equiv \det(\MH) \bmod{qt^{n-3}}\;
\end{align}
Let $p$ be a prime divisor of $q$. Recall $(\Vx~\MA) \in \gln(\zqz)$. 
But then, $p$ does not divide $\det(\Vx~\MA)$. From the fact that
$\MS$ is equivalent to $\Sym$ over $\zqz$, it follows that
$\ordp(\det(\MS))=\ordp(\dSym)$. By definition of $q$,
it follows that $\ordp(q)>\ordp(\dSym t^{n-1})$. But then,
\begin{align}\label{EQALG:MOD}
\ordp(t^{n-2}&\det(\MS)\det(\Vx~\MA)^2)=\ordp(t^{n-2})+\ordp(\det(\MS))\nonumber\\
&=\ordp(t^{n-2})+\ordp(\dSym) < \ordp(q)
\end{align}
From Equation \ref{EQALG:dS} and Equation \ref{EQALG:MOD}, we
conclude that for all primes $p$ that divide $q$, 
$\ordp(\det(\MH))=\ordp(t^{n-2}\det(\MS))=\ordp(t^{n-2}\dSym)$.

It also follows from the definition of
the symbol that $\ordp(\det(\tSym))=\ordp(\det(\MH))$ for 
all relevant primes of symbol $\tSym$. But we showed that 
$\ordp(\det(\MH))=\ordp(t^{n-2}\dSym)$ for every relevant
prime $p$ of $\tSym$. 
Thus, $\det(\tSym)=t^{n-2}\dSym$.
\end{proof}

We now show that $\tSym$ defined in Equation \ref{def:tSym}
is a valid symbol.

\begin{theorem}\label{thm:Existence}\label{thm:tSymExists}
Let $\Sym^{n>1}$ be a valid symbol,
$t$ be a primitively representable integer in $\Sym$, 
$q=\overline{t^{n-1}\dSym}$,
$\MS^n$ be an integral quadratic form with $\MS \eqq \Sym$, 
$\Vx \in (\zqz)^n$ be a primitive vector with $\Vx'\MS\Vx\equiv t \bmod{q}$,
$\MA$ be such that $[\Vx, \MA] \in \gln(\zqz)$, $\Vd:=\Vx'\MS\MA \bmod q$,
$\MH^{n-1}:=(t\MA'\MS\MA - \Vd'\Vd) \bmod q$ and 
$\tSym$ be as defined in Equation \ref{def:tSym}; then $\tSym$
is a valid symbol.
\end{theorem}

\begin{proof}
We divide the proof in three items, one for each condition.
\begin{enumerate}[(i).]
\item (Oddity Condition)
Consider the 
matrix $\MM$ over integers (note that $t$ divides both $q$ and
$\MH+\Vd'\Vd$ and so $(\MH+\Vd'\Vd)/t$ is integral).
\[
\MM=\begin{pmatrix}t & \Vd \\ \Vd' & (\MH+\Vd'\Vd)/t\end{pmatrix}=
\begin{pmatrix}1 & \Vd/t \\ 0 & \MI\end{pmatrix}'
\begin{pmatrix}t & 0 \\ 0 & \MH/t\end{pmatrix}
\begin{pmatrix}1 & \Vd/t \\ 0 & \MI\end{pmatrix}
\]
The matrix $\MV=\begin{pmatrix}1 & \Vd/t \\ 0 & \MI\end{pmatrix}$ is over
rationals and has determinant~1. Thus, $\MM$ is equivalent
to $\diag(t, \MH/t)$ over rationals. By construction,
$\MM \eqq \MS \eqq \Sym$. From \cite[Theorem~3, ][page~372]{CS99},
it follows that for all prime $p$ that divides $q$,
\begin{align}\label{tSymExists:PExcess}
\pexcess(\Sym)=\pexcess(\MM)=\pexcess(t \oplus \MH/t) 
= \pexcess(t)+\pexcess(\MH/t)
\end{align}
By the hypothesis of the theorem, the oddity condition holds for the
symbol $\Sym$. And so,
\begin{align*}
0 &\equiv \sum_{p\in \bbP_q \cup\{-1\}} \pexcess(\Sym) 
  \equiv \sig(\Sym)-n+\sum_{p|q}\left( \pexcess(t)+\pexcess(\MH/t)\right) \\
  &\equiv \sig(\Sym)-(n-1)-\sig(t)+(\sig(t)-1)+\sum_{p|q} \left(\pexcess(t)+\pexcess(\MH/t) \right)\\
  &\equiv \sig(t)(1+\sig(\tSym))-(n-1)-\sig(t)+\sum_{p|q} \pexcess(\MH/t)\\
  &\equiv \sig(t)\sig(\tSym)-(n-1)+\sum_{p|q} \pexcess(t\MH) \\
  &\equiv \sig(t)\sig(\tSym)-(n-1)+\sum_{p|q} \pexcess(t\tSym) \pmod{8}
\end{align*}
From \cite[Theorem~5, ][page~372]{CS99}, it follows that there exists a (rational)
quadratic form $\MX$ which is equivalent to $\tSym$ over rationals.
This also implies that $\tSym$ satisfies the oddity condition (Theorem 
\ref{thm:oddity}).

\item (Determinant Condition)
By definition, $\MH \eqq \tSym$. By item (i), $\MX$ is equivalent to $\tSym$
over rationals. Thus, $\frac{\det(\MH)}{\det(\tSym)}$ is a rational square
modulo $q$ \cite[Theorem~3, ][page~372]{CS99}). From Claim \ref{CALG:dtSym}, 
there exists an integer $x$ such that
\begin{align}\label{tSymExists:Square}
\det(\MH) \equiv t^{n-2}\dSym x^2 \pmod{q}
\end{align}
But then, for all primes $p$ that divide $q$,
\begin{equation*}
\prod_{i\in \scale_p(\tSym)} \epsilon_i = 
\legendre{\copp\left(\det(\MH)\right)}{p} =
\legendre{\copp\left(\det(\tSym)x^2\right)}{p}=
\legendre{\copp\left(\det(\tSym)\right)}{p}
\end{equation*}
This equality show that the determinant condition holds for all primes
$p$ that divide $q$. For all other primes, the determinant condition
holds by construction.

\item (Jordan Condition) The Jordan constituents of the $\MH$ are
 the same as the Jordan constituents of $\tSym$ for all relevant
 primes of $\det(\tSym)$. This is because for all relevent primes $p$ of
 $\tSym$, $\tSym_p=\SYM_p(\MH)$. The quadratic form $\MH$ is integral and so
 its Jordan constituents exist, proving that the Jordan Condition
 is satisfied for $\tSym$.
\end{enumerate}
\end{proof}

\begin{claim}\label{CALG:QisIntegral}
The matrix $\MQ$ is an integral quadratic form with determinant
$\dSym$ and signature $\sig(\Sym)$.
\end{claim}
\begin{proof}
The matrix $\MQ$ is symmetric by construction. By Claim \ref{CALG:dtSym},
the determinant of $\tMH$ equals $t^{n-2}\dSym$. Thus, the following
equality implies that the determinant of $\MQ$ equals $\dSym$.
\begin{align}\label{EQ:dSym}
\MQ = 
\begin{pmatrix}t & \Vd\tMU \\ \tMU'\Vd' & 
\frac{\tMH+\tMU'\Vd'\Vd\tMU}{t}\end{pmatrix}=
\begin{pmatrix}1 & \frac{\Vd\tMU}{t} \\ 0 & \MI\end{pmatrix}'
\begin{pmatrix}t & 0 \\ 0 & \frac{\tMH}{t}\end{pmatrix}
\begin{pmatrix}1 & \frac{\Vd\tMU}{t} \\ 0 & \MI\end{pmatrix}
\end{align}
Let $\tMH$ be the integral quadratic form with symbol $\tSym$,
$\MH:=t\MA'\MS\MA-\Vd'\Vd \bmod{q}$, and
$\tMU$ be such that $\tMU'\MH\tMU\equiv\tMH \bmod{q}$. Then,
\begin{align}\label{CEQ:Lambda}
\begin{pmatrix}1 & 0 \\ 0 & \tMU\end{pmatrix}'
(\Vx~\MA)'\MS(\Vx~\MA)
\begin{pmatrix}1 & 0 \\ 0 & \tMU\end{pmatrix} \bmod q
&=\begin{pmatrix}t & \Vd\tMU \\ (\Vd\tMU)' & \tMU'\MA'\MS\MA\tMU\end{pmatrix}
\bmod q \\
t\tMU'\MA'\MS\MA\tMU - \tMU'\Vd'\Vd\tMU \equiv  \tMU'\MH\tMU 
&\equiv \tMH \bmod{q}\label{EQ:Integer}
\end{align}
The integer $t$ divides $q$. By Equation 
\ref{EQ:Integer}, $\tMH+\tMU'\Vd'\Vd\tMU\equiv 0 \bmod{t}$. But then,
$t$ divides every entry of the matrix
$\tMH+\tMU'\Vd'\Vd\tMU$ and so $\MQ$ is an integral matrix.

Finally, by Equation \ref{EQ:dSym}, $\MQ \eqR t\oplus\frac{\tMH}{t}$. Hence,
\[
\sig(\MQ) = \sig(t) + \sig(\tMH)\sig(t) = \sig(t)+\sig(\tSym)\sig(t)
\overset{(\ref{def:tSym})}{=}\sig(\Sym)
\]
\end{proof}

The proof of Theorem \ref{thm:QFGen} now proceeds as follows.
\begin{enumerate}[(i).]
\addtolength{\itemsep}{-5pt}
\item If the symbol $\Sym$ is valid, then there exists an integer $t$,
which has a primitive representation in $\Sym$ (see Theorem \ref{thm:ExistenceOfRep}).

\item
If $\Sym^n$ satisfies the determinant, oddity and the Jordan
conditions i.e., Equations (\ref{Cond:det})-(\ref{Cond:2}); then
so does $\tSym$. (Theorem \ref{thm:tSymExists})

\item The symbol $\tSym$ is well defined and has a short description. 
In particular, by Claim \ref{CALG:dtSym},
$\det(\tSym)=t^{n-2}\dSym$ and $\tSym$ can equivalently be written
as follows.
\begin{align}\label{def:tSym2}
\tSym=\left(\cup_{p \in \bbP_q}\SYM_p(\MH)\right) \cup 
\{\sig(t)(\sig(\Sym)-\sig(t))\}
\end{align}

\item The output matrix $\MQ$ has determinant $\dSym$. It also 
has the same signature as $\sig(\Sym)$. Thus, $\MQ \eqR \Sym$. 

\item 
If $n=1$ then $\dSym$ is the unique matrix with determinant
equal to the determinant of the symbol $\Sym$. This follows from
the Determinant Condition.

\item
If $n>1$ then, it remains to show that for every relevant
prime of $\Sym$, $\SYM_p(\MQ)=\SymP$.
Consider the following sequence of congruences.
\begin{align*}
t\MQ &= 
\begin{pmatrix}t^2 & t\Vd\tMU \\ t(\Vd\tMU)' 
& \tMH+\tMU'\Vd'\Vd\tMU\end{pmatrix}\\
&\equiv
\begin{pmatrix}t^2 & t\Vd\tMU \\ t(\Vd\tMU)' 
& \tMU'\MH\tMU+\tMU'\Vd'\Vd\tMU\end{pmatrix} \bmod q\\
&\equiv
\begin{pmatrix}1 & 0 \\ 0 & \tMU\end{pmatrix}'
\begin{pmatrix}t^2 & t\Vd \\ t\Vd' 
& \MH+\Vd'\Vd\end{pmatrix}
\begin{pmatrix}1 & 0 \\ 0 & \tMU\end{pmatrix} \bmod q\\
&\overset{(\ref{ALGO:dH})}{\equiv}
\begin{pmatrix}1 & 0 \\ 0 & \tMU\end{pmatrix}'
\begin{pmatrix}t^2 & t\Vd \\ t\Vd' 
& t\MA'\MS\MA\end{pmatrix}
\begin{pmatrix}1 & 0 \\ 0 & \tMU\end{pmatrix} \bmod q\\
&\overset{(\ref{ALGO:dH})}{\equiv}
\begin{pmatrix}1 & 0 \\ 0 & \tMU\end{pmatrix}'
\begin{pmatrix}t\Vx'\MS\Vx & t\Vx'\MS\MA \\ t\MA'\MS\Vx 
& t\MA'\MS\MA\end{pmatrix}
\begin{pmatrix}1 & 0 \\ 0 & \tMU\end{pmatrix} \bmod q\\
&\equiv
\begin{pmatrix}1 & 0 \\ 0 & \tMU\end{pmatrix}'
\begin{pmatrix}\Vx & \MA\end{pmatrix}'t\MS
\begin{pmatrix}\Vx & \MA\end{pmatrix}
\begin{pmatrix}1 & 0 \\ 0 & \tMU\end{pmatrix} \bmod q\\
\MQ &\equiv 
\begin{pmatrix}1 & 0 \\ 0 & \tMU\end{pmatrix}'
\begin{pmatrix}\Vx & \MA\end{pmatrix}'\MS
\begin{pmatrix}\Vx & \MA\end{pmatrix}
\begin{pmatrix}1 & 0 \\ 0 & \tMU\end{pmatrix} \bmod {q/t}\\
\end{align*}
Recall, $\tMU \in \gl_{n-1}(\zqz)$ and $(\Vx~\MA) \in \gln(\zqz)$.
This implies that $\MU=[\Vx,\MA](1\oplus\tMU) \in \gln(\zqz)$. But then,
$\MU\in \gln(\bbZ/\frac{q}{t}\bbZ)$ and $\MQ \overset{q/t}{\sim} \MS$. 
Note that for every prime $p$ that divides $2\dSym$ the following
holds because $\ordp(q/t) > \ordp(\det(\MQ)) + k_p$.
\[
\ordp(\dSym) = \ordp(\det(\MQ)) = \ordp(\det(\MS))
\]
Thus, by definition of $p^*$-equivalence one concludes that 
for every $p$ that divides $2\dSym$, $\MQ \eqp \MS$; 
completing the proof of Theorem \ref{thm:QFGen}.
\end{enumerate}
\end{proof}

\section{Primitive Representation in a Genus}

An important step in the algorithm \textsc{GenSimple}
is to find an integer $t$ which has a 
primitive representation in the genus $\Sym$. 

Recall
Definition \ref{def:PrimRepInGenus}. The following lemma shows
that if $n\geq 2$ then $t$ has a primitive $p^*$-representation
in $\Sym$ for all primes $p$ such that $p$ does not divide
$2\dSym$. A proof of this lemma can already be found in Siegel
\cite{Siegel35}, although in a different setting.

\begin{lemma}\label{lem:Siegel12}
Let $\Sym^{n\geq 2}$ be a valid genus, $t$ be an integer and
$p$ be an odd prime which does not divide $t\dSym$.
Then, $t$ has a primitive $p^*$-representation in $\Sym$.
\end{lemma}
\begin{proof}
Let $p$ be an odd prime which does not divide $t\dSym$. Then, 
by Lemma \ref{lem:RelevantPrimes}, 
$\Sym \eqp \diag(\dSym,1,\cdots,1)$. It suffices to show that
$t$ has a primitive representation in $\diag(\dSym,1)$ over 
$\zpz$ (Theorem \ref{thm:Siegel13}).

By assumption, $\dSym$ and $1$ are invertible modulo $p$.
If $\legendre{t}{p}$ is the same as the Legendre symbol of $\dSym$ or $1$ 
(say, $\dSym$) then $x^2\equiv t\dSym^{-1} \pmod{p}$ has a 
non-trivial solution. 
Otherwise, $\dSym$ and $1$ have the same Legendre symbol, different 
from $t$. But then the result follows from Lemma \ref{lem:QR}.
\end{proof}

\begin{theorem}\label{thm:ExistenceOfRep}
Let $\Sym^{n\geq 2}$ be a valid genus and $\MQ \in \Sym$. A positive 
integer $t$ has a primitive representation in 
$\Sym$ if $t$ has a primitive $p^{K_p}$-representation in $\MQ$ for all
$p$ that divides $2t\dSym$, where 
$K_p= \max\{\ordp(\MQ),\ordp(t)\}+k_p$.
\end{theorem}
\begin{proof}
Follows from Theorem \ref{thm:Siegel13}, Lemma \ref{lem:Siegel12} and
the definition of primitive representations in a genus.
\end{proof}

This simplifies our problem in the algorithmic sense. To find an 
integer $t$ which has a 
primitive representation in the genus $\Sym^{n\geq 2}$, we only need
to check all primes $p$ that divide $2t\dSym$ and only over
the ring $\bbZ/p^{K_p}\bbZ$ for $K_p=\max\{\ordp(\Sym),\ordp(t)\}+k_p$.

For $n>3$, it is comparatively
easy to find a $t$; in fact, it is possible to find a $t$ which divides
$\dSym$. But for dimensions $n=3$ and $n=2$, the proof 
deteriorates to case analyses, especially for dimension~2.
The proofs are constructive in the sense that it is
also possible to find a representation $\Vx$ such that 
$\Vx'\MS\Vx\equiv t\bmod{q}$ in time $\poly(n ,\log \dSym)$.

In this section, we prove the following theorem.

\begin{theorem}\label{thm:FindT}
Let $\Sym^{n\geq 2}$ be a valid reduced genus. Then, there
exists a randomized algorithm that takes $\Sym$ as input; runs in time 
$\poly(|\bbP_{\Sym}|, \log \dSym)$ and outputs
an integer $t$ which has a primitive representation in the genus $\Sym$. 
\end{theorem}

Note that the run time of the algorithm does not depend on $n$. This
is because for $n=4$, we can already find a nice $t$ which divides
$\dSym$ and we ignore the later dimensions. When we want to find
an $\Vx$ such that $\Vx'\MS\Vx\equiv t \bmod{q}$ then we use the same
trick and only represents $t$ using at most $4\times 4$ sub-form of 
$\MS$ (see Section \ref{sec:QFGen}).

But before starting the construction of $t$, we prove two lemmas
which are going to be useful.

\begin{lemma}\label{lem:RepMod2Dim4}
Let $t$ be an odd integer, and 
$2^{i_1}\tau_1 \oplus \cdots \oplus 2^{i_4}\tau_4$ be an
integral quadratic form with $\tau_1, \cdots, \tau_4$ odd and 
$i_1\leq \cdots \leq i_4$. Then, $2^{i_4}t$ has a $2^*$-primitive
representation in $\MD$. Additionally, for every positive integer $k$
there exists a primitive $2^k$-representation $(x_1,\cdots,x_4)$ such 
that $\ordt(x_4)=0$ and $\ordt(2^{i_j}\tau_jx_i^2)\geq i_4$, 
for all $j \in [4]$.
\end{lemma}
\begin{proof}
Let $k=i_4+3$, then  it suffices to show that the following has a 
primitive solution (see Theorem \ref{thm:Siegel13}).
\begin{align}\label{RepMod2Dim4:EQ1}
2^{i_1}\tau_1 x_1^2 + \cdots + 2^{i_4}\tau_4x_4^2 \equiv 2^{i_4}t \pmod{2^{i_4+4}}
\end{align}
We find a primitive solution where $x_4$ is odd. For $j \in [3]$, we set 
$x_j = 2^{\lceil \frac{i_4-i_j}{2} \rceil}y_j$ and divide the entire 
Equation \ref{RepMod2Dim4:EQ1} by $2^{i_4}$. The equation then
reduces to the following.
\begin{align}\label{RepMod2Dim4:EQ2}
\sum_{j \in [3]} 2^{(i_4-i_j) \bmod 2}\tau_j y_j^2 + \tau_4x_4^2 \equiv t \pmod{16}
\end{align}
An exhaustive search shows that for each possible choice of odd
$t$ in $\bbZ/16\bbZ$,
$\tau_1,\cdots, \tau_4 \in \{1,3,5,7\}$ and $i_4-i_j \pmod 2$, the Equation 
\ref{RepMod2Dim4:EQ2} always has a solution, where  $x_4$ is odd.
\end{proof}

\begin{lemma}\label{lem:RepModPDim2}
Let $p$ be an odd prime, 
$\MD=\tau_1 \oplus p^{i}\tau_2$,
where $\tau_1, \tau_2\in (\zpz)^\times$, $i$ even
and $t \in (\zpz)^\times$ such that 
$\legendre{t}{p}\neq \legendre{\tau_1}{p}$. Then, 
$p^it$ has a $p^*$-primitive representation in $\MD$.
\end{lemma}
\begin{proof}
If $\legendre{t}{p}=\legendre{\tau_2}{p}$, then 
$p^{i}t$ has the same symbol as $p^{i}\tau_1$ and then 
$p^{i}t$ has a primitive $p^*$-representation
in $p^i\tau_2$. Otherwise, from the statement of the lemma,
$\legendre{t}{p}\neq \legendre{\tau_1}{p}=\legendre{\tau_2}{p}$.

In this case, $t$ can always be written as 
$t \equiv \tau_1y_1^2+\tau_2y_2^2 \pmod{p}$, where both $y_1$ and $y_2$ are
units of $\zpz$. But then,
\begin{align*}
p^{i}t &\equiv p^{i}\tau_1y_1^2 + p^{i}\tau_2y_2^2 \pmod{p^{i+1}} \\
		&\equiv \tau_1(p^{i/2}y_1)^2 + p^{i}\tau_2y_2^2 \pmod{p^{i+1}}.
\end{align*}
It follows that $p^it$ has a primitive representation by 
$\tau_1\oplus p^i\tau_2$ over $\bbZ/p^{i+1}\bbZ$. By Theorem
\ref{thm:Siegel13}, $p^{i}t$ has a $p^*$-primitive representation in $\MD$.
\end{proof}


\subsection{Representation: $n>3$}

As mentioned earlier, we construct an integer
$t$ such that $t$ divides $\dSym$ and $t$ has a primitive
representation in the input genus $\Sym$. It turns out that
when $n>3$ we do not need to use the fact that the input symbol
$\Sym$ is reduced.


\begin{lemma}\label{lem:TInDim4}
Let $\Sym^{n>3}$ be a genus. Then, there exists an integer
$t$ such that $t$ divides $\dSym$ and $t$ has a primitive
representation in the genus $\Sym$.
\end{lemma}
\begin{proof}
Let us suppose that $p$ is an odd prime that divides $\dSym$. In this
case, we construct a diagonal form using Lemma \ref{lem:LocalQF} as follows.
\begin{align}\label{TInDim4:Local}
\SymP \eqp p^{i_1}\tau_1 \oplus p^{i_2}\tau_2 \oplus \cdots
~~~~~~ \tau_1,\tau_2, \cdots \in (\zpz)^{\times}, i_1\leq i_2 \leq \cdots
\end{align}
An integer can be equivalently written as 
$\prod_{p\in\{-1,2\}\cup \bbP} p^{e_p}$, where $e_p$ is the $p$-order
of the integer. The construction
of the integer $t$ is as follows.
\begin{enumerate}[(i).]
\item For every odd prime that does not divide $\dSym$, 
$e_p$ is identically~0. Also, if $\sig(\Sym)>-n$, then we set $e_{-1}=0$.
Otherwise, $e_{-1}=1$.

\item For every odd prime $p$ that divide $\dSym$ our first step
is to compute the value of $e_p \bmod 2$. Consider a prime $p$ that 
divides $\dSym$. Consider the quadratic form constructed in Equation
\ref{TInDim4:Local} for the prime $p$. Then,
\begin{align}\label{TInDim4:Parity}
e_p \bmod 2 := \maj(i_1 \bmod 2, i_2\bmod 2, i_3 \bmod 2)
\end{align}

\item Next we compute $e_2$. If $\Sym_2$ has a type II
block of $2$-order $\ell$, then we set $\ordt(t)=\ell+1$. 
Otherwise, $\Sym_2$ has only Type I blocks. 
Thus, $\Sym$ is $2$-equivalent to
a diagonal matrix $2^{j_1}\ttau_1 \oplus \cdots$, where $\ttau_1,\cdots$
are odd and $j_1\leq j_2\leq \cdots$. The value of 
$j_1, \cdots, j_4$ can be read off the symbol $\SymT$ as
the four smallest possible $2$-orders in $\Sym_2$. We set
$e_2=j_4$.

\item Once the parity of all $p \in \{-1,2\}\cup \bbP$ is known (see item 
(i)-(iii)), we define an integer $r$ as follows.
\begin{align}\label{TInDim4:r}
r = \prod_{p \in \{-1,2\} \cup \bbP} p^{e_p \bmod 2}
\end{align}

\item
Finally, we compute $e_p$ for all odd primes $p$ which divide 
$\dSym$. Consider the diagonal form constructed in Equation
\ref{TInDim4:Local} for the prime $p$. Out of $(i_1,i_2), (i_2,i_3),$
and $(i_1,i_3)$;
let $(i_a,i_b), a < b \in \{1,2,3\}$ be the pair with the same parity.
Then,
\begin{align}\label{TInDim4:ep}
e_p = \left\{ 
	\begin{array}{ll}
	i_a & \text{if } \legendre{\copp(r)}{p} = \legendre{\tau_a}{p},\\
	i_b & \text{otherwise}
	\end{array}
	\right.
\end{align}
\item We now have $e_p$ for every $p \in \{-1,2\} \cup \bbP$. We define
our integer $t$ as follows.
\begin{align}\label{TInDim4:t}
t = \prod_{p \in \{-1,2\}\cup \bbP} p^{e_p}
\end{align}
\end{enumerate}

The next step is to show that $t$ has a primitive representation in the genus
$\Sym$, or equivalently, $t$ has a $p^*$-primitive
representation in $\Sym$ for all $p \in \{-1,2\} \cup \bbP$. 
\begin{enumerate}[(i).]
\item ($p=-1$) By construction, $t$ is negative iff $\sig(\Sym)=-n$. In this
case $\Sym$ is a genus of negative definite matrices and hence must represent
every negative integer over $\bbR$. Otherwise, $t$ is a positive integer and
$\Sym$ is a genus of non-negative definite matrices i.e., 
$\Sym \eqR 1\oplus\cdots$. Hence, $\Sym$ must represent all positive integers
over $\bbR$. In either case, the constructed $t$ has a primitive representation
in $\Sym$ over $\bbR$.

\item ($p$ odd, $p$ does not divide $\dSym$) In this case, $p$ does
not divide $t$. Hence, $t$ has a $p^*$-primitive representation in $\Sym$
(Lemma \ref{lem:Siegel12}).

\item ($p=2$) If $\SymT$ has a Type II block then $\ordt(t)=\ell+1$, where
$\ell$ is the $2$-order of one of the Type II blocks. Then, the Type II
block represents every integer of $2$-order $\ell+1$ (by Lemma
\ref{lem:RepresentTTypeII}). The existence of a 
$2^*$-primitive representation now follows from 
Lemma \ref{lem:Siegel12}. Otherwise, there are only Type I blocks in $\SymT$ 
and the existence of a $2^*$-primitive representation follows from
Lemma \ref{lem:RepMod2Dim4} and Theorem \ref{thm:Siegel13}.

\item ($p$ odd, $p$ divides $\dSym$) By construction, $e_p$ has the
same parity as $i_a$ and $i_b$, see item (v) of the construction of $t$
and Equation \ref{TInDim4:ep}. Thus, $\ordp(t) \equiv \ordp(r) \pmod 2$, or
\[
\legendre{\copp(t)}{p} = \legendre{\copp(r)}{p}
\]
By Equation \ref{TInDim4:ep}, $e_p=i_a$ if 
$\legendre{\copp(t)}{p}=\legendre{\tau_a}{p}$ and $e_p=i_b$, otherwise.
If $\legendre{\copp(t)}{p}=\legendre{\tau_a}{p}$ then $t$ can be $p^*$-primitively
represented by the diagonal entry $p^{i_a}\tau_a$ in Equation 
\ref{TInDim4:Local}. Otherwise, $t$ can be $p^*$-primitively represented
by $p^{i_a}\tau_a \oplus p^{i_b}\tau_b$ (see Lemma \ref{lem:RepModPDim2}). 
In either case, $t$ has a 
$p^*$-primitive representation in $\SymP$.
\end{enumerate}
\end{proof}

\subsection{Representation: $n=3$}

In this case, we construct an integer 
$t$ with the following properties. If the input
genus $\SymT$ has a Type II block then the constructed 
$t$ divides $\dSym$. Otherwise, $t$ is of the form
$\wp \tilt$, where $\tilt$ divides $\dSym$ and $\wp$
is an odd prime that does not divide $\dSym$.

\begin{lemma}\label{lem:TInDim3}
Let $\Sym^{n=3}$ be a genus. Then, there exists an integer
$\wp \tilt$ such that $\tilt$ divides $\dSym$, 
$\wp\in\bbP\setminus\bbP_{\Sym}$ and $\wp \tilt$ 
has a primitive representation in the genus $\Sym$.
\end{lemma}
\begin{proof}
Let us suppose that $p$ is an odd prime that divides $\dSym$. In this
case, we construct a diagonal form using Lemma \ref{lem:LocalQF} as follows.
\begin{align}\label{TInDim3:Local}
\SymP \eqp p^{i_1}\tau_1 \oplus p^{i_2}\tau_2 \oplus p^{i_3}\tau_3
~~~ \tau_1,\tau_2, \tau_3 \in (\zpz)^{\times}, i_1\leq i_2 \leq i_3
\end{align}
An integer can be equivalently written as 
$\prod_{p\in\{-1,2\}\cup \bbP} p^{e_p}$, where $e_p$ is the $p$-order
of the integer. The construction
of the integer $t$ is as follows.
\begin{enumerate}[(i).]
\item if $\sig(\Sym)>-3$, then we set $e_{-1}=0$.  Otherwise, $e_{-1}=1$.

\item For every odd prime $p$ that divides $\dSym$, our first step
is to compute the value of $e_p \bmod 2$. 
Consider the quadratic form constructed in Equation
\ref{TInDim4:Local} for the prime $p$. Then,
\begin{align}\label{TInDim3:Parity}
e_p \bmod 2 := \maj\{i_1 \bmod 2, i_2\bmod 2, i_3 \bmod 2\}
\end{align}

\item If $\Sym_2$ has a type II
block of $2$-order $\ell$, then we set $\ordt(t)=\ell+1$. Also, for
all odd primes $p \not\in \PSym$, we set $e_p=0$. 

\item 
Otherwise, $\Sym_2$ has only Type I blocks. Thus,
\begin{align}\label{TInDim3:Local2}
\Sym \eqt 2^{j_1}\ttau_1 \oplus 2^{j_2}\ttau_2 \oplus 2^{j_3}\ttau_3, ~~
\ttau_1,\ttau_2,\ttau_3 \in \SGN^{\times}, i_1\leq i_2 \leq i_3 
\end{align}
The value of $j_1$ can be read off the symbol $\SymT$ as
the smallest possible $2$-orders in $\Sym_2$. We set
$e_2=j_1$. We also pick an odd prime $\wp$ not in $\PSym$ 
satisfying the following equation.
\begin{align}\label{TInDim3:wp}
\left(\prod_{p\in\{-1\}\cup (\PSym\setminus\{2\})}
p^{e_p \bmod 2}\right)\wp \equiv \ttau_1 \pmod 8
\end{align}
Such a prime can be found by rejection sampling. A random odd prime
satisfies the Equation \ref{TInDim3:wp} with probability 
$1/4$. We set $e_{\wp}=1$. Also, for all primes $p$ that do not
divide $2\wp\dSym$, we set $e_p=0$.

\item Once the parity of all $p \in \{-1,2\}\cup \bbP$ is known 
(see item (i)-(iv) of the construction), we define an integer $r$ as follows.
\begin{align}\label{TInDim3:r}
r = \prod_{p \in \bbP\cup \{-1,2\}} p^{e_p \bmod 2}
\end{align}

\item
Finally, we compute $e_p$ for all odd primes $p$ which divide 
$\dSym$. Consider the diagonal form constructed in Equation
\ref{TInDim4:Local} for the prime $p$. Out of $(i_1,i_2), (i_2,i_3),$
and $(i_1,i_3)$;
let $(i_a,i_b), a < b \in \{1,2,3\}$ be the pair which has the same parity.
Then,
\begin{align}\label{TInDim3:ep}
e_p = \left\{ 
	\begin{array}{ll}
	i_a & \text{if } \legendre{\copp(r)}{p} = \legendre{\tau_a}{p},\\
	i_b & \text{otherwise}
	\end{array}
	\right.
\end{align}
\item We now have $e_p$ for every $p \in \{-1,2\} \cup \bbP$. We define
our integer $t$ as follows.
\begin{align}\label{TInDim3:t}
t = \prod_{p \in \{-1,2\}\cup \bbP} p^{e_p}
\end{align}
\end{enumerate}

The next step is to show that $t$ has a primitive representation in the genus
$\Sym$. Equivalently, it suffices to show that $t$ has a $p^*$-primitive
representation in $\Sym$ for all $p \in \{-1,2\} \cup \bbP$. Note that
if $\SymT$ has a Type II block then the construction of $t$ in this case
is the same as the construction of $t$ in the case of 
Lemma \ref{lem:TInDim4}. The correctness of the construction also follows
from the same proof. In the rest, we assume that $\SymT$ has no Type II
block.
\begin{enumerate}[(i).]
\item ($p=-1$) By construction, $t$ is negative iff $\sig(\Sym)=-3$. In this
case $\Sym$ is a genus of negative definite matrices and hence must represent
every negative integer over $\bbR$. Otherwise, $t$ is a positive integer and
$\Sym$ is a genus of non-negative definite matrices i.e., 
$\Sym \eqR 1\oplus\cdots$. Hence, $\Sym$ must represent all positive integers
over $\bbR$. In either case, the constructed $t$ has a primitive representation
in $\Sym$ over $\bbR$.

\item ($p$ odd, $p$ does not divide $\wp\dSym$) In this case, $p$ does
not divide $t$. Hence, $t$ has a $p^*$-primitive representation in $\Sym$
(Lemma \ref{lem:Siegel12}).

\item ($p=2$) By assumption, there are only Type I blocks in $\SymT$.
By construction of $t$, $\symtk(2^{j_1}\ttau_1) = \symtk(t)$. But then,
$t$ has a $2^*$-primitive representation in $2^{j_1}\ttau_1$ 
(Lemma \ref{lem:Square}).

\item ($p$ odd, $p$ divides $\dSym$) By construction, $e_p$ has the
same parity as $i_a$ and $i_b$, see item (vi) of the construction of $t$
and Equation \ref{TInDim3:ep}. Thus, $\ordp(t) \equiv \ordp(r) \pmod 2$, or
\[
\legendre{\copp(t)}{p} = \legendre{\copp(r)}{p}
\]
By Equation \ref{TInDim3:ep}, $e_p=i_a$ if 
$\legendre{\copp(t)}{p}=\legendre{\tau_a}{p}$ and $e_p=i_b$, otherwise.
If $\legendre{\copp(t)}{p}=\legendre{\tau_a}{p}$ then $t$ can be $p^*$-primitively
represented by the diagonal entry $p^{i_a}\tau_a$ (see Equation 
\ref{TInDim3:Local}). Otherwise, $t$ can be $p^*$-primitively represented
by $p^{i_a}\tau_a \oplus p^{i_b}\tau_b$ (see Lemma \ref{lem:RepModPDim2}). 
In either case, $t$ has a 
$p^*$-primitive representation in $\SymP$.

\item ($p=\wp$) Finally, we show that $t$ has a $\wp^*$-primitive
representation in $\Sym$.
The prime $\wp$ does not divide $\dSym$ and hence by Lemma
\ref{lem:RelevantPrimes}, $\Sym \eqx{\wp^*}\dSym\oplus1\oplus1$. 
Consider the following equation.
\begin{align}\label{TInDim3:wp2}
\dSym x_1^2 + x_2^2+x_3^2 \equiv t \pmod{\wp^2}\;.
\end{align}
By Lemma \ref{lem:QR}, $x_2^2+x_3^3$ represents 
$t - \dSym$ over $\bbZ/\wp\bbZ$. Also,
$\dSym x_1^2$ represents $\dSym$
primitively over $\bbZ/\wp^2\bbZ$. Thus, Equation
\ref{TInDim3:wp2} has a primitive solution. By Theorem \ref{thm:Siegel13},
$t$ has a $\wp^*$-primitive representation in $\Sym$.
\end{enumerate}
\end{proof}

\subsection{Representation: $n=2$, basics}\label{dim2}

Finding an integer representation in dimension~2 is the most difficult. 
As with dimension~3, we may need a prime $\wp$ but it needs to satisfy
more stringent conditions. In this case, we strongly use the fact that
the input symbol $\Sym$ is reduced and valid.


Recall the definition of a reduced genus. In dimension~2
a reduced genus $\Sym$ has the following form.
\begin{align}\label{def:red2}
\Sym \left\{\begin{array}{ll}
\eqp a_p \oplus p^{i_p}b_p & \text{where }a_p, b_p \in (\zpz)^\times, 
p \text{ odd},\\
\eqt \MX \in \{a_2 \oplus 2^i b_2, \MTP, \MTP\} &
\text{where }a_2,b_2 \in \SGN^\times
\end{array}\right.
\end{align}

Note that $\Sym$ is a symbol in dimension~2 and hence 
$\sig(\Sym) \in \{2,0,-2\}$. Define the quantities $\epsilon, \rho, $
and the function $\xi:\bbZ\to\{0,1\}$, as follows.
\begin{align}\label{def:dim2}
\begin{array}{lll}
\epsilon&=&\frac{\dSym}{|\dSym|} \\
\rho &=& \left\{\begin{array}{ll}
1 & \text{if }\sig(\Sym)\in\{0,2\}\\
-1 & \text{otherwise.}
\end{array}\right.\\
\xi(x) &=&\left\{\begin{array}{ll}
1 & \text{if }{x \choose 2}=\frac{x(x-1)}{2} \text{ is odd, and}\\
0 & \text{otherwise.}
\end{array}\right.
\end{array}
\end{align}

Then, the signature and the oddity of $\Sym$ can be computed
as follows.
\begin{align}\label{dim2:sigodty}
\begin{array}{lll}
\sig(\Sym) &=& \rho(1+\epsilon) \\
\odty(\Sym) &=& \left\{\begin{array}{ll}
0 & \text{if }\SymT \eqt \MTP \text{ or }\SymT \eqt \MTM\\
a_2 + b_2 \bmod 8 & \text{if }\legendre{b_2}{2}=1, \text{ or  $i_2$ even}\\
a_2+b_2+4 \bmod 8 & \text{otherwise}
\end{array}\right.
\end{array}
\end{align}

For convenience, we define the set $S$ as the set of
odd primes $p$ for which $i_p$ is odd i.e., 
$S=\{p \in \PSym \cap \bbP \mid i_p \text{ odd}\}$.
Next, for each $d \in \{1,3,5,7\}$ and $b \in \{-,+\}$ we define 
sets,
\begin{align}\label{def:Sdb}
S_{db} &= \Big\{p \in S \mid p \equiv d \bmod 8, 
\legendre{a_p}{p} = b \Big\}
\end{align}

If we eliminate
a subscript, it means a union of the sets with all possible values
of the subscript. For example, $S_3=S_{3+}\cup S_{3-}$. The calligraphic
versions, as usual, will denote the size of the corresponding sets. For
example, $\fS_{\{3,5\}-}$ is $|S_{3-}|+|S_{5-}|$.

\begin{lemma}\label{lem:dim2pexcess}
Let $\Sym^{n=2}$ be a valid reduced genus,
$\epsilon=\frac{\dSym}{|\dSym|}$, and $m$ be the total number of 
antisquares in $\Sym$. Then,
\begin{align*}
\left(\sum_{p \in \bbP\cap\PSym} \pexcess(\Sym)\right)  
\equiv 2\fS_3+4\fS_5+6\fS_7+2(1-(-1)^m) \pmod 8
\end{align*}
where,
\[
(-1)^m = 
	\left\{\begin{array}{ll}
	(-1)^{\fS_{-}+\xi(\fS_{\{3,7\}})}\epsilon^{\fS_{\{3,7\}}}
	& \text{if }\ordt(\dSym) \text{ is even}\\
	(-1)^{\fS_{-}+\xi(\fS_{\{3,7\}})+\fS_{\{3,5\}}}\epsilon^{\fS_{\{3,7\}}}
	& \text{otherwise.}\\
	\end{array}\right.
\]
\end{lemma}
\begin{proof}
Let $p$ be an odd prime from $\PSym$. To compute the $p$-excess,
we need to compute the number of $p$-antisquares in $\SymP$.
By construction, $\Sym\eqp a_p \oplus p^{i_p}b_p$ has a $p$-antisquare 
iff $i_p$ is odd and $\legendre{b_p}{p}=-1$. But then, 
$\dSym \eqp \det(a_p\oplus p^{i_p}b_p)=a_pb_pp^{i_p}$ and so
$\legendre{b_p}{p}=\legendre{a_p\copp(\dSym)}{p}$. If $m$ is the total
number of antisquares in $\Sym$ then,
\begin{align*}
(-1)^m &= \prod_{p \in S} \legendre{b_p}{p} 
	= \prod_{p \in S} \legendre{a_p\copp(\dSym)}{p} 
	= \prod_{p \in S} \legendre{\copp(\dSym)}{p} 
	\prod_{p \in S} \legendre{a_p}{p} \\
	&= \left\{\begin{array}{ll}
	(-1)^{\fS_{-}}\prod_{p \in S}\legendre{\epsilon}{p}
	\prod_{p_i \neq p_j \in S} \legendre{p_i}{p_j}\legendre{p_j}{p_i} 
	& \text{if }\ordt(\dSym) \text{ is even}\\
	(-1)^{\fS_{-}}\prod_{p \in S}\legendre{2\epsilon}{p}
	\prod_{p_i \neq p_j \in S} \legendre{p_i}{p_j}\legendre{p_j}{p_i} 
	& \text{otherwise.}
	\end{array}\right.
\end{align*}
Note that $\legendre{-1}{p}=-1$ iff $p \equiv 3 \bmod 4$ i.e., 
$\prod_{p \in S}\legendre{\epsilon}{p}=\epsilon^{\fS_{\{3,7\}}}$. Also,
$\legendre{2}{p}=-1$ iff $p\bmod 8 \in \{3,5\}$ i.e.,
$\prod_{p\in S}\legendre{2}{p}=(-1)^{\fS_{\{3,5\}}}$. By Quadratic
Reciprocity, 
$\prod_{p_i \neq p_j} \legendre{p_i}{p_j}\legendre{p_j}{p_i}
=(-1)^{\xi(\fS_{\{3,7\}})}$. Putting it together, we have,
\[
(-1)^m = 
	\left\{\begin{array}{ll}
	(-1)^{\fS_{-}+\xi(\fS_{\{3,7\}})}\epsilon^{\fS_{\{3,7\}}}
	& \text{if }\ordt(\dSym) \text{ is even}\\
	(-1)^{\fS_{-}+\xi(\fS_{\{3,7\}})+\fS_{\{3,5\}}}\epsilon^{\fS_{\{3,7\}}}
	& \text{otherwise.}\\
	\end{array}\right.
\]

By definition, if $m$ is the total number of $p$-antisquares in $\Sym$
with $p$ odd then, 
\begin{align*}
\sum_{p\in \bbP\cap\PSym} \pexcess(\Sym) &= \sum_{p \in S} \pexcess(\Sym) 
= 4m+\sum_{p \in S} (p-1)  \bmod 8 \\
&= 2\fS_3 + 4\fS_5 + 6\fS_7 + 4m \pmod 8
\end{align*}

The expression $4m \bmod 8$ evaluates to~4 iff $m$ is odd. Equivalently,
$4m \bmod 8$ evaluates to~4 iff $(-1)^m$ evaluates to
$-1$; completing the proof.
\end{proof}

\begin{lemma}\label{lem:RepresentWP}
Let $\Sym^{n=2}$ be a reduced genus,
$r$ be an integer and $\wp$ be an
odd prime that does not divide $r\dSym$. Then, $r\wp$ has a 
$\wp^*$-primitive representation in $\Sym$ iff 
$\legendre{-\dSym}{\wp}=1$. If $\ordt(\dSym)$ is even, then
\[
\legendre{-\dSym}{\wp} = \left\{\begin{array}{ll}
\underset{p\in S}{\prod} \legendre{\wp}{p} & \text{if }\wp\equiv1\bmod4\\
(-1)^{\fS_{\{3,7\}}+1}\epsilon\underset{p\in S}{\prod} \legendre{\wp}{p} & \text{otherwise.}
\end{array}\right.
\]
and if $\ordt(\dSym)$ is odd then,
\[
\legendre{-\dSym}{\wp} = \left\{\begin{array}{ll}
\underset{p\in \{2\} \cup S}{\prod} \legendre{\wp}{p} 
& \text{if }\wp\equiv1\bmod4\\
(-1)^{\fS_{\{3,7\}}+1}\epsilon
\underset{p\in \{2\}\cup S}{\prod} \legendre{\wp}{p} & \text{otherwise.}
\end{array}\right.
\]
\end{lemma}
\begin{proof}
By Lemma \ref{lem:RelevantPrimes}, 
$\Sym \overset{\wp*}{\sim} \diag(\dSym, 1)$.  By Theorem 
\ref{thm:Siegel13}, $r\wp$ has a $\wp^*$-primitive representation 
in $\Sym$ iff the following equation has a primitive solution;
\[
\wp r \equiv \dSym x^2 + y^2 \pmod{\wp^2}\;.
\]
Both $x$ and $y$ must be units of $\bbZ/\wp\bbZ$. But then, 
\[
\legendre{-\dSym}{\wp}=\legendre{\wp r - \dSym x^2}{\wp}
=\legendre{y^2}{\wp}=1.
\]

Convsersely, if $\legendre{-\dSym}{\wp}=1$ then, by Lemma 
\ref{lem:Square}, the following equation
has a solution.
\[
x^2 \equiv -\dSym + \wp r \pmod{\wp^2}
\]
Thus, $x^2+\dSym \equiv \wp r \pmod{\wp^2}$ or $\wp r$ has a 
$\wp^*$-primitive representation in $1\oplus \dSym$.

Next, we write $\legendre{-\dSym}{\wp}$ in terms of $\epsilon, \rho, $
and $\fS_{db}$ using the Law of Quadratic Reciprocity (Equation 
\ref{QuadraticReciprocity}). If $\ordt(\dSym)$ is even then,
\begin{align*}
\legendre{-\dSym}{\wp} 
	&=\legendre{-\epsilon}{\wp}\prod_{p \in S} \legendre{p}{\wp} \\
	&=\legendre{-\epsilon}{\wp} \prod_{p \in S_{\{1,5\}}} \legendre{\wp}{p} 
	\prod_{p \in S_{\{3,7\}}} \legendre{p}{\wp}\\
	&= \left\{\begin{array}{ll}
	\underset{p \in S}{\prod} \legendre{\wp}{p}
	& \text{$\wp \equiv 1 \bmod 4$,} \\
	(-1)^{\fS_{\{3,7\}}}\legendre{-\epsilon}{\wp}
	\underset{p \in S}{\prod}\legendre{\wp}{p}
	& \text{otherwise.}\\
	\end{array}\right.
\end{align*}
On the other hand, if $\ordt(\dSym)$ is odd then 
\begin{align*}
\legendre{-\dSym}{\wp} 
	&=\legendre{-2\epsilon}{\wp}\prod_{p \in S} \legendre{p}{\wp} \\
	&= \left\{\begin{array}{ll}
	\legendre{\wp}{2}\underset{p \in S}{\prod} \legendre{\wp}{p}
	& \text{$\wp \equiv 1 \bmod 4$,} \\
	(-1)^{\fS_{\{3,7\}}}\legendre{\wp}{2}\legendre{-\epsilon}{\wp}
	\underset{p \in S}{\prod}\legendre{\wp}{p}
	& \text{otherwise.}\\
	\end{array}\right.
\end{align*}
\end{proof}


\subsection{Representation: $n=2$, Type II}

In this section, we construct an integer $t$ such that $t$ has
a primitive representation in $\Sym$, where $\SymT \eqt \MTP$ or
$\SymT \eqt \MTM$. Note that, in this situation,~2 does not divide
$\dSym$.

\begin{lemma}\label{TypeIIdim2:FindT}
Let $\Sym^{n=2}$ be a valid reduced genus with $\SymT \eqt \MTM$ or
$\SymT \eqt \MTP$. Then, there exists an integer of the form 
$2\wp r^2$ with primitive representation
in $\Sym$, where $\wp$ is an odd prime that does not divide $\dSym$
and $r^2$ is an integer that divides $\dSym$.
\end{lemma}
\begin{proof}
Recall the definition of $\epsilon, \rho \in \{-1,1\}$ as in Equation 
\ref{def:dim2}. Let us define the following set of congruences.
\begin{align}\label{TypeIIdim2:wp}
\begin{array}{ll}
\wp & \equiv 2\rho a_p  \pmod p ~\text{ for all} ~ p \in S\\
\wp &\equiv \left\{\begin{array}{ll}
1 \pmod 4 & \text{if }\rho=+, \fS_{\{3,5\}}+\fS_{-} \text{ even, or} \\
& \text{if }\rho=-, \fS_{\{5,7\}}+\fS_{-} \text{ even.} \\
3\pmod 4 &\text{if }\rho=+,\epsilon=+,\fS_{\{5,7\}}+\fS_{-}\text{ odd, or} \\
& \text{if }\rho=+,\epsilon=-,\fS_{\{5,7\}}+\fS_{-} \text{ even, or}\\
& \text{if }\rho=-,\epsilon=-,\fS_{\{3,5\}}+\fS_{-} \text{ even, or}\\
& \text{if }\rho=-,\epsilon=+,\fS_{\{3,5\}}+\fS_{-} \text{ odd.}
\end{array}\right.
\end{array}
\end{align}

Note that the set of possibilities under which we can write a 
modulo~4 congruence is not exhaustive. It is, as we show later, 
exhaustive for every valid symbol $\Sym$.

It is possible to solve the congruence in such a way that 
$\wp$ is a prime (Dirichlet's Theorem).
Consider an integer $r$ defined as follows.
\begin{align}\label{TypeIIdim2:r}
r = \prod_{p \in \PSym} p^{e_p/2} \qquad
e_p = \left\{\begin{array}{ll}
0 & \text{if }p \in \{2\}\cup S \\
0 & \text{if }p \in \PSym \setminus (\{2\}\cup S), 
\legendre{a_p}{p}=\legendre{2\rho\wp}{p} \\
i_p & \text{if }p \in \PSym \setminus (\{2\}\cup S), 
\legendre{a_p}{p}\neq \legendre{2\rho\wp}{p} 
\end{array}\right.
\end{align}

Note that, if $p$ is an odd prime not in $S$ then $i_p$ is even. Thus,
$r$ is an integer. Define $t=2\rho \wp r^2$. We next show that $t$
has a primitive representation in the genus $\Sym$. 
For this, it
suffices to show that $t$ has a primitive $p^*$-representation in $\Sym$
for all $p \in \{p \mid \ordp(2t\dSym)>0\}\cup\{-1\}$ 
(see Lemma \ref{lem:Siegel12}).
\begin{enumerate}[(i).]
\item ($p=-1$) By Equation \ref{def:dim2}, $\rho=-1$ iff 
$\sig(\Sym)=-2$. In this case, $\Sym$ is negative definite and 
represents all negative integers. Otherwise, $\Sym \eqR 1\oplus x$ and
$\rho=1$. But then, $t$ is a positive integer and hence can be
represented by $\Sym$ over $\bbR$. In either case, $t$ has a 
representation in $\Sym$ over $\bbR$.

\item ($p=2$) By construction,~2 does not divide $\wp r^2$ and so
$\ordp(t)=1$. By assumption, $\SymT \eqt \MTP$ or $\SymT \eqt \MTM$.
In either case, by Lemma \ref{lem:RepresentTTypeII} and 
Theorem \ref{thm:Siegel13}, $t$ has a 
$2^*$-primitive representation in $\SymT$.

\item ($p \in S$) By definition of $r$, $\ordp(r)=0$ for all primes
$p \in S$. By construction in Equation \ref{TypeIIdim2:wp},
$\legendre{t}{p}=\legendre{2\rho\wp}{p}=\legendre{a_p}{p}$, where
$\Sym \eqp a_p \oplus p^{i_p}b_p$. But then, by Lemma 
\ref{lem:Square}, $a_p$ (and hence, $\Sym$) 
represents $t$, $p^*$-primitively.

\item ($p$ odd, $p\in \PSym \setminus S$) If $\Sym \eqp
a_p \oplus p^{i_p}b_p$ then, $i_p$ is even. If 
$\legendre{a_p}{p}=\legendre{2\rho\wp}{p}$ then, $p$ does
not divide $2\rho\wp r^2$ and $\legendre{a_p}{p}=\legendre{t}{p}$.
Thus, $t$ has a $p^*$-primitive representation in $\Sym$ (Lemma
\ref{lem:Square} and Theorem \ref{thm:Siegel13}).
Otherwise, $\legendre{a_p}{p} \neq \legendre{2\rho\wp}{p}$. But then,
$\ordp(t)=\ordp(r^2)=i_p$, and
by Lemma \ref{lem:RepModPDim2} and Theorem \ref{thm:Siegel13}, $t$
has a $p^*$-primitive representation in $a_p \oplus p^{i_p}b_p$.

\item ($p=\wp$) Finally, it remains to show that $t$ has a 
$\wp^*$-primitive representation in $\Sym$. By Lemma 
\ref{lem:RepresentWP}, one needs to show that 
$\legendre{-\dSym}{\wp}=1$. Recall the Quadratic Reciprocity
Laws in Equation \ref{QuadraticReciprocity}. Also, note that
$\legendre{-1}{p}=1$ iff $p \equiv 1 \bmod 4$. The computation
of $\legendre{-\dSym}{\wp}$ can be done using Lemma 
\ref{lem:RepresentWP}, as follows.
\begin{align*}
\prod_{p\in S}\legendre{\wp}{p} &= \prod_{p \in S} \legendre{2\rho a_p}{p} 
	= (-1)^{\fS_{-}}\rho^{\fS_{\{3,7\}}}\prod_{p \in S} \legendre{2}{p} \\
	&=(-1)^{\fS_{-}}\rho^{\fS_{\{3,7\}}}\prod_{p \in S} \legendre{p}{2} 
	= (-1)^{\fS_{-}+\fS_{\{3,5\}}}\rho^{\fS_{\{3,7\}}}\\
\legendre{-\dSym}{\wp} 
	&= \left\{\begin{array}{ll}
	(-1)^{\fS_{-}+\fS_{\{3,5\}}}\rho^{\fS_{\{3,7\}}} 
	& \text{$\wp \equiv 1 \bmod 4$,} \\
	(-1)^{\fS_{\{5,7\}}+\fS_{-}+1}\epsilon\rho^{\fS_{\{3,7\}}}
	& \text{otherwise.}\\
	\end{array}\right. \\
	&= \left\{\begin{array}{ll}
	(-1)^{\fS_{-}+\fS_{\{3,5\}}}
	& \rho=+, \wp \equiv 1 \bmod 4 \\
	(-1)^{\fS_{-}+\fS_{\{5,7\}}}
	& \rho=-, \wp \equiv 1 \bmod 4\\
	(-1)^{\fS_{-}+\fS_{\{5,7\}}+1} & \rho=+, \epsilon=+, \wp\equiv 3\bmod 4\\
	(-1)^{\fS_{-}+\fS_{\{5,7\}}} & \rho=+, \epsilon=-, \wp\equiv 3\bmod 4\\
	(-1)^{\fS_{-}+\fS_{\{3,5\}}} & \rho=-, \epsilon=-, \wp\equiv 3\bmod 4\\
	(-1)^{\fS_{-}+\fS_{\{3,5\}}+1} & \rho=-, \epsilon=+, \wp\equiv 3\bmod 4\\
	\end{array}\right. \\
\end{align*}
It turns out that $\wp \bmod 4$ was defined to satisfy exactly
this equation (see Equation \ref{TypeIIdim2:wp}).
\end{enumerate}

This completes the proof of the claim that $t$ has a primitive
representation in the genus $\Sym$.

Finally, we show that the set of possibilities under the modulo~4
congruence in Equation \ref{TypeIIdim2:wp} is exhaustive, if the
input symbol $\Sym$ is valid.
The proof of this statement is computer assisted and the code
can be found in Appendix \ref{sec:Code}.

We design the test program as follows. For all possible choices of
$\epsilon, \rho \in \{1,-1\}$, and $\fS_{db} \in \{0,1,2,3\}$, we
compute $\sig(\Sym), \odty(\Sym)$ by Equation \ref{dim2:sigodty}. 
We also compute $\sum_{p \in \bbP}\pexcess(\Sym)$ by Lemma 
\ref{lem:dim2pexcess}. Then, we check the oddity condition i.e,
\[
\sig(\Sym)+\sum_{p \in \bbP}\pexcess(\Sym)\equiv\odty(\Sym) \pmod 8
\]
If the oddity condition is satisfied then we check if at least one
of these conditions hold.
\begin{align}\label{TypeIIdim2:good}
\begin{array}{ll}
(\rho=+, \fS_{\{3,5\}}+\fS_{-} \text{ even}) &
(\rho=-, \fS_{\{5,7\}}+\fS_{-} \text{ even}) \\
(\rho=+,\epsilon=+,\fS_{\{5,7\}}+\fS_{-}\text{ odd}) &
(\rho=+,\epsilon=-,\fS_{\{5,7\}}+\fS_{-} \text{ even})\\
(\rho=-,\epsilon=-,\fS_{\{3,5\}}+\fS_{-} \text{ even}) &
(\rho=-,\epsilon=+,\fS_{\{3,5\}}+\fS_{-} \text{ odd})
\end{array}
\end{align}
In each of these cases, a $\wp$ and hence $t$ exists by Equation
\ref{TypeIIdim2:wp}. The test program never finds itself in
a situation when none of the conditions in Equation \ref{TypeIIdim2:good} 
are true. This completes the proof of existence of a primitively
representable $t$.

\end{proof}

\subsection{Representation: $n=2$, Type I, Even}

This section deals with the case when $\SymT \eqt a_2 \oplus 2^{i_2} b_2$,
where $i_2$ is even, and $a_2, b_2 \in \{1,3,5,7\}$.

\begin{lemma}\label{lem:TypeIEven}
Let $\Sym^{n=2}$ be valid reduced genus with
$\SymT\eqt a_2\oplus 2^{i_2}b_2$, where $i_2$ is even and 
$a_2, b_2 \in \SGN^\times$. Then,
there exists an integer of the form $\wp r^2$ with primitive representation
in $\Sym$, where $\wp$ is an odd prime that does not divide $\dSym$
and $r^2$ is an integer that divides $\dSym$.
\end{lemma}
\begin{proof}
Recall the definition of $\epsilon, \rho \in \{-1,1\}$ as in Equation 
\ref{def:dim2}. Let us define the following set of congruences.
\begin{align}\label{TypeIEven:wp}
\begin{array}{ll}
\wp &\equiv \rho a_p \pmod{p} ~~ \text{for all }p \in S\\
\wp &\equiv \left\{\begin{array}{ll}
x \pmod 8 & \text{if }\rho=+, \fS_{-} \text{ even and } 
x \in X \cap \{1,5\}\text{ or} \\
& \text{if }\rho=-, \fS_{\{3,7\}}+\fS_{-} \text{ even and } 
x \in X \cap \{1,5\}\text{ or} \\
y\pmod 8 &\text{if }\rho=+,\epsilon=+,\fS_{\{3,7\}}+\fS_{-}\text{ odd and }
y \in X \cap \{3,7\}\text{ or}\\
& \text{if }\rho=+,\epsilon=-,\fS_{\{3,7\}}+\fS_{-} \text{ even and }
y \in X \cap \{3,7\}\text{ or}\\
& \text{if }\rho=-,\epsilon=-,\fS_{-} \text{ even and }
y \in X \cap \{3,7\}\text{ or}\\
& \text{if }\rho=-,\epsilon=+,\fS_{-} \text{ odd and }
y \in X \cap \{3,7\}
\end{array}\right.\\
&\text{where, }X = \{\rho a_2 \bmod 8, \rho b_2 \bmod 8\}
\end{array}
\end{align}
A few word on notation. The modulo~8 congruences should be read as
follows. Consider the first congruence ``$\rho \equiv x \bmod 8$,
if $\rho=+, \fS_{-}$ even and $x \in X\cap\{1,5\}$''. If $X\cap\{1,5\}$
is empty then this statement is false. Otherwise, we pick any
element $x$ from the intersection. Also, $x \equiv 1 \bmod 4$
and $y \equiv 3 \bmod 4$.

Note that the set of possibilities under which we can write a 
modulo~8 congruence is not exhaustive. It is, as we show later, 
exhaustive for every valid symbol $\Sym$.

It is possible to solve the congruence in such a way that 
$\wp$ is a prime (Dirichlet's Theorem).
Consider an integer $r$ defined as follows.
\begin{align}\label{TypeIEven:r}
r = \prod_{p \in \PSym} p^{e_p/2} \qquad
e_p = \left\{\begin{array}{ll}
0 & \text{if }(p \in S) \text{ or }(p=2, \wp \equiv a_2 \bmod 8)\\
i_2 & \text{if }p=2, \wp \equiv b_2 \bmod 8\\
0 & \text{if }p \in \PSym \setminus (\{2\}\cup S), 
\legendre{a_p}{p}=\legendre{\rho\wp}{p} \\
i_p & \text{if }p \in \PSym \setminus (\{2\}\cup S), 
\legendre{a_p}{p}\neq \legendre{\rho\wp}{p} 
\end{array}\right.
\end{align}
The exponent $e_p$ is always even and hence $r$ is an integer.
Define $t=\rho \wp r^2$. We next show that $t$
has a primitive representation in the genus $\Sym$. For this, it
suffices to show that $t$ has a primitive $p^*$-representation in $\Sym$
for all $p \in \{p \mid \ordp(2t\dSym)>0\}\cup\{-1\}$ 
(see Lemma \ref{lem:Siegel12}).

\begin{enumerate}[(i).]
\item ($p=-1$) By Equation \ref{def:dim2}, $\rho=-1$ iff 
$\sig(\Sym)=-2$. In this case, $\Sym$ is negative definite and 
represents all negative integers. Otherwise, $\Sym \eqR 1\oplus x$ and
$\rho=1$. But then, $t$ is a positive integer and hence can be
represented by $\Sym$ over $\bbR$. In either case, $t$ has a 
representation in $\Sym$ over $\bbR$.

\item ($p \in S$) By definition of $r$, $\ordp(r)=0$ for all primes
$p \in S$. By construction in Equation \ref{TypeIEven:wp},
$\legendre{t}{p}=\legendre{\rho\wp}{p}=\legendre{a_p}{p}$, where
$\Sym \eqp a_p \oplus p^{i_p}b_p$. But then, by Lemma 
\ref{lem:Square}, $a_p$ (and hence, $\Sym$) 
represents $t$, $p^*$-primitively.

\item ($p$ odd, $p\in \PSym \setminus S$) If $\Sym \eqp
a_p \oplus p^{i_p}b_p$ then, $i_p$ is even. If 
$\legendre{a_p}{p}=\legendre{\rho\wp}{p}$ then, $p$ does
not divide $\rho\wp r^2$ and $\legendre{a_p}{p}=\legendre{t}{p}$.
Thus, $t$ has a $p^*$-primitive representation in $\Sym$ (Lemma
\ref{lem:Square} and Theorem \ref{thm:Siegel13}).
Otherwise, $\legendre{a_p}{p} \neq \legendre{\rho\wp}{p}$. But then,
$\ordp(t)=\ordp(r^2)=i_p$, and
by Lemma \ref{lem:RepModPDim2} and Theorem \ref{thm:Siegel13}, $t$
has a $p^*$-primitive representation in $a_p \oplus p^{i_p}b_p$.

\item ($p=\wp$) Next, we show that $t$ has a 
$\wp^*$-primitive representation in $\Sym$. By Lemma 
\ref{lem:RepresentWP}, one needs to show that 
$\legendre{-\dSym}{\wp}=1$. Recall the Quadratic Reciprocity
Laws in Equation \ref{QuadraticReciprocity}. Also, note that
$\legendre{-1}{p}=1$ iff $p \equiv 1 \bmod 4$. The computation
of $\legendre{-\dSym}{\wp}$ can be done using Lemma 
\ref{lem:RepresentWP}, as follows.
\begin{align*}
\prod_{p\in S}\legendre{\wp}{p} &= \prod_{p \in S} \legendre{\rho a_p}{p} 
	= (-1)^{\fS_{-}}\rho^{\fS_{\{3,7\}}}\\
\legendre{-\dSym}{\wp} 
	&= \left\{\begin{array}{ll}
	(-1)^{\fS_{-}}\rho^{\fS_{\{3,7\}}} 
	& \text{$\wp \equiv 1 \bmod 4$,} \\
	(-1)^{\fS_{\{3,7\}}+\fS_{-}+1}\epsilon\rho^{\fS_{\{3,7\}}}
	& \text{otherwise.}\\
	\end{array}\right. \\
	&= \left\{\begin{array}{ll}
	(-1)^{\fS_{-}}
	& \rho=+, \wp \equiv 1 \bmod 4 \\
	(-1)^{\fS_{-}+\fS_{\{3,7\}}}
	& \rho=-, \wp \equiv 1 \bmod 4\\
	(-1)^{\fS_{-}+\fS_{\{3,7\}}+1} & \rho=+, \epsilon=+, \wp\equiv 3\bmod 4\\
	(-1)^{\fS_{-}+\fS_{\{3,7\}}} & \rho=+, \epsilon=-, \wp\equiv 3\bmod 4\\
	(-1)^{\fS_{-}} & \rho=-, \epsilon=-, \wp\equiv 3\bmod 4\\
	(-1)^{\fS_{-}+1} & \rho=-, \epsilon=+, \wp\equiv 3\bmod 4\\
	\end{array}\right. \\
\end{align*}
It turns out that $\wp \bmod 4$ was defined to satisfy exactly
this equation (see Equation \ref{TypeIIdim2:wp}).

\item ($p=2$) In this case, $\SymT \eqt a_2\oplus2^{i_2}b_2$, where
$a_2, b_2 \in \{1,3,5,7\}$. 
From Equation \ref{TypeIEven:wp}, either $a_2 \eqt t$ or 
$2^{i_2}b_2\eqt t$. In either case, $t$ has a $2^*$-primitive 
representation in $\Sym$.
\end{enumerate}

This completes the proof of the claim that $t$ has a primitive
representation in the genus $\Sym$.

Finally, we show that the set of possibilities under the modulo~8
congruence in Equation \ref{TypeIEven:wp} is exhaustive, if the
input symbol $\Sym$ is valid.
The proof of this statement is computer assisted and the code
can be found in Appendix \ref{sec:Code}.

We design the test program as follows. For all possible choices of
$\epsilon, \rho \in \{1,-1\}$, $a_2,b_2 \in \{1,3,5,7\}$ 
and $\fS_{db} \in \{0,1,2,3\}$, we
compute $\sig(\Sym), \odty(\Sym)$ by Equation \ref{dim2:sigodty}. 
We also compute $\sum_{p \in \bbP}\pexcess(\Sym)$ by Lemma 
\ref{lem:dim2pexcess}. Then, we check the oddity condition i.e,
\[
\sig(\Sym)+\sum_{p \in \bbP}\pexcess(\Sym)\equiv\odty(\Sym) \pmod 8
\]
We next check the following determinant condition for $\SymT$.
\[
\legendre{a_2b_2}{2} = (-1)^{\fS_{\{3,5\}}}
\]
If either of these conditions is not satisfied then the symbol $\Sym$
is not valid. For the others, we check if at least one of these
condition holds.
\begin{align}\label{TypeIEven:good}
\begin{array}{l}
(\rho=+, \fS_{-}\text{ even}, |X\cap\{1,5\}|>0) \\
(\rho=-, \fS_{-}+\fS_{\{3,7\}}\text{ even}, |X\cap\{1,5\}|>0) \\
(\rho=+,\epsilon=+,\fS_{\{3,7\}}+\fS_{-}\text{ odd}, |X\cap\{3,7\}|>0) \\
(\rho=+,\epsilon=-,\fS_{\{3,7\}}+\fS_{-} \text{ even}, |X\cap\{3,7\}|>0)\\
(\rho=-,\epsilon=-,\fS_{-} \text{ even}, |X\cap\{3,7\}|>0) \\
(\rho=-,\epsilon=+,\fS_{-} \text{ odd}, |X\cap\{3,7\}|>0)
\end{array}
\end{align}
In each of these cases, a $\wp$ and hence $t$ exists by Equation
\ref{TypeIEven:wp}. The test program never finds itself in
a situation when none of the conditions in Equation \ref{TypeIEven:good} 
are true. This completes the proof of existence of a primitively
representable $t$.

\end{proof}

\subsection{Representation: $n=2$, Type I, Odd}

This section deals with the case when $\SymT \eqt a_2 \oplus 2^{i_2} b_2$,
where $i_2$ is odd and $a_2,b_2 \in \{1,3,5,7\}$.

\begin{lemma}\label{lem:TypeIOdd}
Let $\Sym^{n=2}$ be a valid reduced genus with $\SymT \eqt 
a_2\oplus 2^{i_2}b_2$, where $i_2$ is odd and $a_2,b_2 \in \SGN^\times$.
Then, there exists an integer of the form $\wp r^2$ or $2^{i_2}\wp r^2$ with 
primitive representation in $\Sym$, where $\wp$ is an odd prime that 
does not divide $\dSym$ and $r^2$ is an integer that divides 
$\dSym$.
\end{lemma}
\begin{proof}
By assumption $i_2$ is odd and hence an odd power of~2 divides
$\dSym$. 

Consider the following set of congruences, along with the construction
of the candidate primitively representable integer $t$.
\begin{align}\label{TypeIOdd:wp}
\begin{array}{l}
\textbf{if } \Big(\rho a_2 \equiv 1 \bmod 4 \textbf{ and } 
(-1)^{\fS_{-}}\rho^{\fS_{\{3,7\}}}\legendre{a_2}{2}=1\Big) \textbf{ or } \\
\Big(\rho a_2 \equiv 3 \bmod 4 \textbf{ and } 
(-1)^{\fS_{-}+\fS_{\{3,7\}}+1}\rho^{\fS_{\{3,7\}}}
\epsilon\legendre{a_2}{2}=1\Big) \textbf{ then}\\
\begin{array}{ll}
~~\wp &\equiv \rho a_p \bmod p  ~~\text{ for all }p\in S\\
~~\wp &\equiv \rho a_2 \bmod 8 \\
~~e_p &= \left\{\begin{array}{ll}
0 & \text{if }p \in \{2\}\cup S \\
0 & \text{if }p \in \PSym \setminus (\{2\}\cup S), 
\legendre{a_p}{p}=\legendre{\rho\wp}{p} \\
i_p & \text{if }p \in \PSym \setminus (\{2\}\cup S), 
\legendre{a_p}{p}\neq \legendre{\rho\wp}{p} 
\end{array}\right.\\
~~t &= \rho \wp \left(\underset{p \in \PSym}{\prod} p^{e_p}\right)
\end{array}\\
\textbf{elif } \Big(\rho b_2 \equiv 1 \bmod 4 \textbf{ and } 
(-1)^{\fS_{\{3,5\}}+\fS_{-}}\rho^{\fS_{\{3,7\}}}\legendre{b_2}{2}=1\Big)
\textbf{ or}\\
\Big(\rho b_2 \equiv 3 \bmod 4 \textbf{ and }
(-1)^{\fS_{-}+\fS_{\{5,7\}}+1}\rho^{\fS_{\{3,7\}}}
\epsilon\legendre{b_2}{2}=1\Big) \textbf{ then}\\
\begin{array}{lll}
~~\wp &\equiv & 2\rho a_p \bmod p ~~ \text{ for all }p\in S\\
~~\wp &\equiv& \rho b_2 \bmod 8 \\
~~e_p &=& \left\{\begin{array}{ll}
0 & \text{if }p \in S \\
0 & \text{if }p \in \PSym \setminus (\{2\}\cup S), 
\legendre{a_p}{p}=\legendre{2\rho\wp}{p} \\
i_p & \text{if }p \in \PSym \setminus (\{2\}\cup S), 
\legendre{a_p}{p}\neq \legendre{2\rho\wp}{p} 
\end{array}\right.\\
~~t &=& \rho 2^{i_2}\wp \left(\underset{p\in\PSym\cap\bbP}{\prod}p^{e_p}\right)
\end{array}
\end{array}
\end{align}

Note that the set of possibilities under which we can write  
the congruence for $\wp$ is not exhaustive. It is, as we show later, 
exhaustive for every valid symbol $\Sym$.

We show that $t$ has a primitive representation in
$\Sym$, or equivalently, $t$ has a $p^*$-primitive representation
in $\Sym$ for all $p \in \{-1,2\}\cup \bbP$.
For this, it
suffices to show that $t$ has a primitive $p^*$-representation in $\Sym$
for all $p \in \{p \mid \ordp(2t\dSym)>0\}\cup\{-1\}$ 
(see Lemma \ref{lem:Siegel12}).
\begin{enumerate}[(i).]
\item ($p=-1$) The value of $\rho=-1$ iff $\sig(\Sym)=-2$. Thus,
$t$ has a representation over $\bbR$ in $\Sym$.

\item (odd $p \in \PSym \setminus S$) In this case, 
$\SymP \eqp a_p \oplus p^{i_p}b_p$ and $\ordp(t)=i_p$, where
$i_p$ is even. By Lemma \ref{lem:RepModPDim2},
$t$ has a primitive representation in $a_p \oplus p^{i_p} b_p$.

\item ($p=2$) By construction, either $t \eqt a_2$ or $t\eqt 2^{i_2}b_2$.
In either case, $t$ has a primitive $2^*$-representation in $\Sym$ 
(Lemma \ref{lem:Square}).

\item ($p=\wp$) In this case, we need to show that 
$\legendre{-\dSym}{\wp}=1$. We split the proof into 
two sub-cases.
\begin{enumerate}[(a).]
\item ($\wp \equiv \rho a_2 \bmod 8$) 
We first compute the value of 
$\prod_{p \in \{2\}\cup S} \legendre{\wp}{p}$.
\begin{align*}
\prod_{p \in \{2\}\cup S} \legendre{\wp}{p} 
&= \prod_{p\in\{2\}\cup S}\legendre{\rho a_p}{p}
=(-1)^{\fS_{-}}\rho^{\fS_{\{3,7\}}}\legendre{a_2}{2}
\end{align*}
And then, we insert it into the computation 
of $\legendre{-\dSym}{\wp}$ in Lemma \ref{lem:RepresentWP}.
\begin{align*}
\legendre{-\dSym}{\wp} &= \left\{\begin{array}{ll}
(-1)^{\fS_{-}}\rho^{\fS_{\{3,7\}}}\legendre{a_2}{2} 
& \text{if }\wp\equiv 1 \bmod 4\\
(-1)^{\fS_{-}+\fS_{\{3,7\}}+1}\rho^{\fS_{\{3,7\}}}\epsilon\legendre{a_2}{2} 
& \text{otherwise}
\end{array}\right.
\end{align*}

\item ($\wp \equiv \rho b_2 \bmod 8$) Similarly, we compute,
\begin{align*}
\prod_{p \in \{2\}\cup S} \legendre{\wp}{p} 
&= \legendre{\rho b_2}{2}\prod_{p \in S}\legendre{2\rho a_p}{p}
=(-1)^{\fS_{-}+\fS_{\{3,5\}}}\rho^{\fS_{\{3,7\}}}\legendre{b_2}{2}
\end{align*}
And then, we insert it into the computation 
of $\legendre{-\dSym}{\wp}$ in Lemma \ref{lem:RepresentWP}.
\begin{align*}
\legendre{-\dSym}{\wp} &= \left\{\begin{array}{ll}
(-1)^{\fS_{-}+\fS_{\{3,5\}}}\rho^{\fS_{\{3,7\}}}\legendre{b_2}{2} 
& \text{if }\wp\equiv 1 \bmod 4\\
(-1)^{\fS_{-}+\fS_{\{5,7\}}+1}\rho^{\fS_{\{3,7\}}}\epsilon\legendre{b_2}{2} 
& \text{otherwise}
\end{array}\right.
\end{align*}
\end{enumerate}
In either case, $\legendre{-\dSym}{\wp}=1$, proving the
primitive $\wp^*$-representativeness.
\end{enumerate}

This completes the proof of the claim that $t$ has a primitive
representation in the genus $\Sym$.

Finally, we show that the set of possibilities when we can write
a congruence (see Equation \ref{TypeIOdd:wp}) is exhaustive, if the
input symbol $\Sym$ is valid.
The proof of this statement is computer assisted and the code
can be found in Appendix \ref{sec:Code}.

We design the test program as follows. For all possible choices of
$\epsilon, \rho \in \{1,-1\}$, $a_2,b_2 \in \{1,3,5,7\}$ 
and $\fS_{db} \in \{0,1,2,3\}$, we
compute $\sig(\Sym), \odty(\Sym)$ by Equation \ref{dim2:sigodty}. 
We also compute $\sum_{p \in \bbP}\pexcess(\Sym)$ by Lemma 
\ref{lem:dim2pexcess}. Then, we check the oddity condition i.e,
\[
\sig(\Sym)+\sum_{p \in \bbP}\pexcess(\Sym)\equiv\odty(\Sym) \pmod 8
\]
We next check the following determinant condition for $\SymT$.
\[
\legendre{a_2b_2}{2} = (-1)^{\fS_{\{3,5\}}}
\]
If either of these conditions is not satisfied then the symbol $\Sym$
is not valid. For the others, we check if at least one of these
condition holds.
\begin{align}\label{TypeIOdd:good}
\begin{array}{l}
(\rho a_2 \equiv 1 \bmod 4,
(-1)^{\fS_{-}}\rho^{\fS_{\{3,7\}}}\legendre{a_2}{2}=1)\\
(\rho a_2 \equiv 3 \bmod 4,
(-1)^{\fS_{-}+\fS_{\{3,7\}}+1}\rho^{\fS_{\{3,7\}}}
\epsilon\legendre{a_2}{2}=1)\\
(\rho b_2 \equiv 1 \bmod 4,
(-1)^{\fS_{\{3,5\}}+\fS_{-}}\rho^{\fS_{\{3,7\}}}\legendre{b_2}{2}=1)\\
(\rho b_2 \equiv 3 \bmod 4,
(-1)^{\fS_{-}+\fS_{\{5,7\}}+1}\rho^{\fS_{\{3,7\}}}
\epsilon\legendre{b_2}{2}=1)
\end{array}
\end{align}
In each of these cases, a $\wp$ and hence $t$ exists as in Figure~5.1.
The test program never finds itself in
a situation when none of the conditions in Equation \ref{TypeIOdd:good} 
are true. This completes the proof of existence of a primitively
representable $t$.

\end{proof}

\subsection{Representation: putting it together}\label{dim2end}

\begin{proof}(Theorem \ref{thm:FindT})
The construction follows from the constructive nature of 
Lemma \ref{lem:TInDim4}, Lemma \ref{lem:TInDim3} and the constructions
for the case of dimension~2. The only remaining task in case of
dimension $2$, is to find a $\wp$ which satisfies the
given set of congruence relations.

Assuming ERH, one can find $\sigma_p$ i.e., the smallest non-residue
modulo $p$ in $O(\log^3 p)$ ring operations over $\zpz$. 
If done for every prime that divides
$\dSym$, this takes $O(|\PSym|\log ^3 \dSym)$ ring
operations over $\bbZ/\dSym\bbZ$.

Let $p_1,\cdots,p_s$ be the primes which appear with odd parity in the
symbol and $\alpha=8p_1\cdots p_s$. Then, we form the required set of 
congruent equations i.e.,
\[
x \equiv \left\{ 
	\begin{array}{ll}
	x_{p_i} \bmod{p_i} & \text{if $\epsilon_i=-1$, 
	where $x_{p_i} \in (\zpz)^\times$} \\
	\tau \bmod8 & \text{where $\tau \in \{1,3,5,7\}$}
	\end{array}\right.
\]
Solve this set of 
congruence using the Chinese Remainder and let $a$ be a solution.
Pick a $b$ uniformly
at random from the range $[0,\alpha^2]$. If 
$S=\{a+z\alpha \mid z \in \bbZ, z \leq \alpha^2\}$, then $a+b\alpha$ is
a uniformly random element of $S$. By Theorem \ref{thm:ERH} with probability
$\frac{1}{\log |S|}$ the number $a+b\alpha$ is prime. One then sets
$\wp = a+b\alpha$. If repeated $O(\log^2|S|)$ times, one can find $\wp$ with
overwhelming probability. The time complexity of the algorithm follows from
the fact that $|S| \leq \alpha^2 \leq \dSym^2$.
\end{proof}

The next step is to devise an algorithm that given the local form $\MS$,
positive integer $q$ and the generated
$t$ finds a primitive $\Vx$ such that $\Vx'\MS\Vx\equiv t\bmod{q}$.
Instead, we find primitive representations $\Vx_p$ for all $p$ that divides $q$
such that $\Vx_p'\MS_p\Vx_p \equiv t \bmod{p^k}$, where $\MS_p$ is 
the $p^*$-equivalent
form, $k = \ordp(q)$, and then combining them using Chinese Remainder.

The construction of $t$ used at most~4 diagonal entries of $\MS_p$
and so to find $\Vx_p$ we use Theorem \ref{thm:PolyRep}
and construct
$\Vx$ by filling the rest of the dimensions with $0$. The time taken
by this algorithm does not not depend on $n$ and is 
$\poly(k,\log p)$, for each prime factor of $q$.

\section{Polynomial Time Algorithm}\label{sec:QFGenPoly}

In this section, we give the main contribution of this thesis.

\begin{theorem}\label{thm:Main}
Let $\Sym^n$ be a valid genus. Then, there exists a randomized
$\poly(n,\log \dSym)$ algorithm that 
outputs a quadratic form $\MQ^n \in \Sym$ with constant probability.
\end{theorem}
\begin{proof}
Recall definition of the
reduced genus. By Lemma \ref{lem:ReducedGenus}, it follows that
finding a quadratic form in $\Sym^*=\red(\Sym)$ suffices for generating
a quadratic form in $\Sym$. 

The algorithm described in Section 
\ref{sec:QFGen} is correct; but is not polynomial as it is.
The analysis of the time 
complexity will be done on a different algorithm, the correctness
of which will follow from the proof of correctness of the algorithm
in Section \ref{sec:QFGen}. We now describe the algorithm.

\paragraph{}\textsc{QFGenPoly}({\em input:} valid symbol $\Sym^n$) 
{\em output:} $\MQ^n \in \Sym$
\begin{enumerate}[i.]
\addtolength{\itemsep}{-4pt}
\item If $n<4$ then return \textsc{QFGen}($\Sym$).
\item Compute $\gcd(\Sym)$ and let $\Sym^*=\Sym/\gcd(\Sym)$.
\item Find $t$ which has a primitive representation in $\Sym^*$.
Let $q=\overline{t^{n-1}\det(\Sym^*)}$ and $K_p=\ordp(q)$. 
\item For every $p\in\bbP_{\Sym}$, we construct a block diagonal matrix
$\MS_p$ and a matrix $[\Vx_p,\MA_p] \in \bbZ/p^{K_p}\bbZ$ as follows. 
Use Theorem \ref{thm:LocalQF} to find a block diagonal matrix $\MD_p\eqp\SymP^*$. 
Recall the construction of $t$ in Lemma \ref{lem:TInDim4}.

If $p$ is odd then by construction, $t$ has a primitive representation in 
$\bbZ/p^{K_p}\bbZ$ of two different types; (a) $t$ has a primitive representation
by the first entry of $\MD_p$. Let $x$ be the primitive representation.
Then, set $\MS_p=\MD_p$
and $[\Vx_p,\MA_p]=\begin{pmatrix}x & 0\\0 & x^{-1}\end{pmatrix}\oplus\MI^{n-2}$
, (b) otherwise, $t$ has a primitive representation
by two of the first three entries of $\MD_p$, say $d_1,d_2$ where 
$\ordp(d_1)\geq \ordp(d_2)$. If $(x_1,x_2)$ is a primitive $p^{K_p}$ 
representation of $t$ then
in this case $x_1$ is primitive. Let $(\MD_p)_{3+}$ be the rest of the blocks
in $\MD_p$ then we set $\MS_p=d_1\oplus d_2 \oplus (\MD_p)_{3+}$ and 
$[\Vx_p,\MA_p]=\begin{pmatrix}x_1 & 0\\x_2 & x_1^{-1}\end{pmatrix}\oplus\MI^{n-2}$.

On the other hand, when $p=2$ then too $t$ has a primitive representation in
$\bbZ/2^{K_2}\bbZ$ of two different kinds; (a) when $\SymT^*$ has a type II
block then if $x_1,x_2$ be the primitive representation with $x_1$ odd then we set
$\MS_2$ as the block diagonal form equivalent to $\MD_p$ where the first 
block is the type II block which was used to represent $t$. Then, we set
$[\Vx_2,\MA_2]=\begin{pmatrix}x_1 & 0\\x_2 & x_1^{-1}\end{pmatrix}\oplus\MI^{n-2}$,
(b) otherwise the first four Type I entries of $\MD_2$ were used to represent
$t$ and $\ordt((\MD_2)_4)=\ordp(t)$, by construction. Also, if $x_1,x_2,x_3,x_4$
is the primitive representation of $t$ then $x_4$ is primitive. In this case, we set
$\MS_2$ as $(\MD_2)_4\oplus\cdots\oplus(\MD_2)_1\oplus(\MD_2)_{4+}$, and
$[\Vx_2,\MA_2]=\begin{pmatrix}x_4 & 0 & 0 &0\\x_3 & x_4^{-1}&0&0\\
x_2&0&1&0\\x_1&0&0&1\end{pmatrix}\oplus\MI^{n-4}$.

{\bf Property.} The construction satisfies the property that for each 
$p \in \SymP$, $[\Vx_p,\MA_p] \in \gln(\bbZ/p^{K_p}\bbZ)$ and 
$\Vx_p'\MS_p\Vx_p\equiv t\bmod{p^{K_p}}$.

\item For each $p\in\PSym$, let $\tSym_p=\SYM_p(\MH_p)$, where $\MH_p$
is defined as follows.
\begin{align}\label{HALG:dHP}
\begin{array}{ll}
\Vd_p=\Vx_p'\MS_p\MA_p \bmod{p^{K_p}}  & 
\MH_p = (t\MA_p'\MS_p\MA_p - \Vd_p'\Vd_p) \bmod{p^{K_p}}\\
\end{array}
\end{align}

\item Let $\tSym=\{\tSym_p:p\in\PSym\}$, and $\tSym^*=\tSym/\gcd(\tSym)$.
\item Call this algorithm recursively with input
$\tSym^*$. Let us suppose that the algorithm returns $\tMH^* \in \gen(\tSym^*)$.
Then, set $\tMH=\gcd(\tSym)\tMH^*$.
\item Use Chinese Remaindering to compute $[\Vx,\MA]$ from
$\{[\Vx_p,\MA_p] \bmod{p^{K_p}}:p\in \PSym\}$, $\MS$ from 
$\{\MS_p\bmod{p^{K_p}}:p\in\PSym\}$
and $\MH$ from $\{\MH_p\bmod{p^{K_p}}:p\in\PSym\}$.
\item Canonicalize both $\MH$ and $\tMH$ over $\zqz$, i.e., we find
$\tMU \in \gl_{n-1}(\zqz)$ such that $\tMH\equiv \tMU'\MH\tMU \bmod q$.
\item Output the following quadratic form.
\begin{align}\label{ALGO:Q}
\MQ = \gcd(\Sym)\begin{pmatrix}t & \Vd\tMU \\ 
(\Vd\tMU)' & \frac{\tMH+\tMU'\Vd'\Vd\tMU}{t}\end{pmatrix}
\end{align}
\addtolength{\itemsep}{4pt}
\end{enumerate}

The correctness of this algorithm follows from the proof before.

Let us compute the time complexity of this algorithm.

Our first step is to show that the recursions do not blow up the 
size of the symbol. Notice that to calculate the $(n-1)$-dimensional
symbol, we multiply by $t$ in Equation \ref{HALG:dHP}. The analysis 
is done below, separately for odd primes and $p=2$.

\paragraph{Odd primes.}
In this case, we show that $\ordp(\tSym^*)=\ordp(\Sym^*)$. 
\medskip

For each odd $p\in\PSym$, $\Vx_p'\MS_p\Vx_p\equiv t\bmod{p^{K_p}}$
and $[\Vx_p,\MA_p] \in \gln(\bbZ/p^{K_p}\bbZ)$. Recall the two 
cases discussed in the algorithm while constructing $\MS_p$ and
$[\Vx_p,\MA_p]$. 
\begin{enumerate}
\item[Case 1:] Suppose $t$ is primitively
representable by the first entry of $\MS_p$. Let $x$ be the primitive
representation. Note
that $x$ is primitive and $(\MS_p)_1$ has $p$-scale~0 because
$\Sym^*$ is reduced. 
The computation for $\MH_p$ (and $\ordp(\tSym))$
is as follows. 
\begin{gather*}
[\Vx_p,\MA_p] \equiv \begin{pmatrix}x & 0 \\ 0 & x^{-1}\end{pmatrix}\oplus 
\MI^{n-2}\bmod{p^{K_p}} \qquad \MS_p=s_1\oplus\cdots \oplus s_n\\
[\Vx_p,\MA_p]'\MS_p[\Vx_p,\MA_p] \equiv x^2s_1 \oplus s_2x^{-2} 
\diag(s_3,\cdots) \bmod{p^{K_p}}\\
\MH_p \equiv t\MA_p'\MS_p\MA_p - \Vd_p'\Vd_p\equiv 
s_1s_2 \oplus x^2s_1\diag(s_3,\cdots,s_n) \bmod{p^{K_p}} \\
\ordp(\tSym)=\ordp(\MH_p)=\ordp(x^2s_1s_n)=\ordp(s_n)=\ordp(\Sym^*)
\end{gather*}

\item[Case 2:] Otherwise, the first two entries of $\MS_p$
represent $t$. In this case, if $\MS_p=s_1\oplus\cdots \oplus s_n$,
then by construction, $\ordp(s_1)\geq \ordp(s_2)$ and $\ordp(x_1)=0$,
where $x_1,x_2$ is the primitive representation of $t$ in $\MS_p$.
Then,
\begin{gather*}
[\Vx_p,\MA_p] = \begin{pmatrix}x_1 & 0 \\ x_2 & x_1^{-1}\end{pmatrix}\oplus \MI^{n-2}
\qquad \MS_p = s_1\oplus\cdots\oplus s_n\\
[\Vx_p,\MA_p]'\MS_p[\Vx_p,\MA_p] \equiv 
\begin{pmatrix}s_1x_1^2+s_2x_2^2 & s_2x_1^{-1}x_2\\
s_2x_1^{-1}x_2 & s_2x_1^{-2}\end{pmatrix} \oplus \diag(s_3,\cdots,s_n)\bmod{p^{K_p}}\\
\MH_p \equiv t\MA_p'\MS_p\MA_p - \Vd_p'\Vd_p 
= s_1s_2\oplus(s_1x_1^2+s_2x_2^2)\diag(s_3 \cdots s_n) \bmod{p^{K_p}}
\end{gather*}
By construction, $s_1x_1^2+s_2x_2^2\equiv t\bmod{p^{K_p}}$, $x_1$ is primitive
and $\ordp(s_2)=\ordp(t)$. Thus, each entry of $\MH_p$ is divisible by
$p^{\ordp(t)}$. When we reduce $\tSym$ to $\tSym^*$, we have
$\ordp(\tSym^*)=\ordp(s_n)=\ordp(\Sym^*)$.
\end{enumerate}
Thus, for odd prime $p$, in either case, $\ordp(\tSym^*)=\ordp(\Sym^*)$.

\paragraph{Prime $p=2$.} Recall the two cases discussed in the algorithm.
The first case is when $t$ has a primitive representation using a type II block.
By construction of $\MS_2$ in Theorem \ref{thm:LocalQF}, the block is
either $\MTP$ or $\MTM$. Let $i$ be the $2$-order of the block, then
\begin{gather*}
[\Vx_2,\MA_2] = \begin{pmatrix}x_1 & 0 \\ x_2 & x_1^{-1}\end{pmatrix}\oplus \MI^{n-2}
\qquad \MS_2 = 2^i\begin{pmatrix}2&1\\1&2c\end{pmatrix}\oplus(\MS_p)_{3+}\\
[\Vx_2,\MA_2]'\MS_2[\Vx_2,\MA_2] \equiv 2^i\begin{pmatrix}2x_1^2+2x_1x_2+2cx_2^2 
& 1+2cx_2x_1^{-1}\\
1+2cx_2x_1^{-1} & 2cx_1^{-2}\end{pmatrix} \oplus (\MS_p)_{3+}\\
\MH_2 = 4^i(4c-1) \oplus 2^i(2x_1^2+2x_1x_2+2cx_2^2)(\MS_p)_{3+}
\end{gather*}
Hence, $\ordt(\tSym^*)\leq 1+\ordt(\Sym^*)$. There could be at most~n Type II blocks
in a quadratic form of dimension~n, which can be generated during the recursion.
Thus, the $2$-order of the recursively generated reduced $2$-symbols remain
bounded by $\ordp(\Sym^*)+n$.

Otherwise, $t$ has a primitive representation using the first four type I blocks 
of $\MS_2$. In this case, the calculations are as follows.
\begin{gather*}
[\Vx_2,\MA_2] = \begin{pmatrix}x_4 & 0 & 0 & 0 \\ x_3 & x_4^{-1} & 0 & 0\\
x_2 & 0 & 1 & 0\\ x_1 & 0 & 0 & 1\end{pmatrix}\oplus \MI^{n-4}
\qquad \MS_2 = \diag(d_4,\cdots,d_1,s_5,\cdots,s_n)\\
t \equiv d_4x_4^2+\cdots+d_1x_1^2 \bmod{2^{K_2}}\\
[\Vx_2,\MA_2]'\MS_2[\Vx_2,\MA_2] = 
\begin{pmatrix}t & d_3x_3x_4^{-1} &d_2x_2 & d_1x_1 \\
d_3x_3x_4^{-1} & d_3x_4^{-2} & 0 & 0\\
d_2x_2 & 0 & d_2 & 0\\
d_1x_1 & 0 & 0 & d_1\\
\end{pmatrix} \oplus (\MS_2)_{5+} \bmod{2^{K_2}}\\
\MH_2 = \begin{pmatrix}d_3\frac{d_1x_1^2+d_2x_2^2+d_4x_4^2}{x_4^2}
& -\frac{x_3d_3x_2d_2}{x_4} & -\frac{x_3d_3x_1d_1}{x_4}\\
-\frac{x_3d_3x_2d_2}{x_4} & (d_1x_1^2+d_3x_3^2+d_4x_4^2)d_2 & -x_2d_2x_1d_1\\
-\frac{x_3d_3x_1d_1}{x_4} & -x_2d_2x_1d_1 & (d_2x_2^2+d_3x_3^2+d_4x_4^2)d_1 
\end{pmatrix}\\
\oplus t(\MS_2)_{5+} \bmod{2^{K_2}}
\end{gather*}
Recall Lemma \ref{lem:RepMod2Dim4}. By 
construction of $x_1,\cdots,x_4$
it follows that for each $i\in[4]$, 
$\ordt(d_ix_i^2)\geq \ordt(d_4)=\ordt(t)$,
$\ordt(d_4)\geq \cdots \geq \ordt(d_1)$ and $\ordt(x_4)=0$.
This implies that $\ordt(x_1)\geq \cdots\geq \ordt(x_4)=0$ and
by inspection, every entry in the first $3\times3$ submatrix
of $\MH_2$ is divisible by $2^{\ordt(t)}$. Thus, 
$\ordt(\tSym^*)=\ordt(\Sym^*)$.

To recapitulate, $\ordp(\tSym^*)$ is equal to $\ordp(\Sym^*)$
unless we use a Type II block to represent $t$ modulo $\bbZ/2^{K_2}\bbZ$,
in which case it increases by exactly~1.

The step by step calculation of the time taken by the algorithm is
as follows.

\begin{enumerate}[(i.)]
\item After calculating the reduced symbol $\Sym^*$, the algorithm 
starts by computing a positive
integer $t$ which is primitively representable in $\Sym^*$. For 
$n\geq4$ such an integer can be found by looking at the first~4 
dimensions of the symbol $\Sym$, see Lemma \ref{lem:TInDim4}.
This takes time linear in the number of relevant primes of $\Sym$
i.e., $O(|\PSym|(\log\det(\Sym^*))^2)$.
\item Next, we find a quadratic form $\MS$ which is equivalent
to $\Sym^*$ over the ring $\zqz$, for $q=\overline{t^{n-1}\det(\Sym^*)}$. 
By Lemma \ref{lem:TInDim4}, the integer $t$ has the property that
$t$ divides $\det(\Sym^*)$. Thus, we do not introduce any new primes and 
for every prime $p \in \PSym$;
\[
\ordp(q) \leq n\ordp(\det(\Sym^*))+k_p\;.
\]
By Theorem \ref{thm:LocalQF}, finding such an integral quadratic form
$\MS_p$ takes time $\poly(n,\log \dSym , \log p)$. There are $|\PSym|$
relevant primes and hence the total time in this step is.
\begin{align}\label{Time:FindS}
\poly(|\PSym|, n,\log\det(\Sym^*))
\end{align}
\item Then, we find a primitive representation $\Vx_p$ of $t$ in $\MS_p$ over
$\zpkz$, $k=\ordp(q)\leq n\log\det(\Sym^*)$,
for all $p \in \PSym$. Note that
the representation is done by 
Theorem \ref{thm:PolyRep} on a $4\times 4$ submatrix, which takes time 
$\poly(k,\log p)$. By the bound on $k$, we get the following
expression.
\begin{align}\label{Time:Findx}
O(|\PSym|,n,\log\det(\Sym^*))\;.
\end{align}
\item Then, we Chinese Remainder the matrices $[\Vx_p,\MA_p]$, $\MS_p$
and $\MH_p$ entry-by-entry ($ \leq n^2$ entries in each matrix)
to get $[\Vx,\MA]$, $\MS$ and $\MH$, respectively.
The modulus of the Chinese Remainder is $q$. This takes time 
$\poly(|\PSym|,n, \log q)$.
\item Finally, we canonicalize both $\MH$ and $\tMH$ modulo $q$. This
is again done by canonicalizing for each prime that divides $q$ and then
Chinese Remaindering the results. For each $p$, 
$\ordp(q) \leq n\log\det(\Sym^*)$. Thus, the time taken for each $p$ is 
bounded by $\poly(|\PSym|, n, \log\det(\Sym^*))$.
\end{enumerate}

The next step is to calculate the reduced form $\tSym^*$ and recurse.
By the discussion of the blowup above, it follows that
$\det(\tSym^*) \leq 2^{n-2}\det(\Sym^*)$. Or, 
$\log \det(\tSym^*)\leq (n-2)\log 2+\log \dSym$. Thus the total
time complexity of the algorithm can be written recursively
as 
\begin{align*}
T(n,\det(\Sym^*)) &= T(n-1,2^{n-2}\det(\Sym^*)) 
+ \poly(|\PSym|,n, \log\det(\Sym^*))
\end{align*}

Although the blowup in the determinant is exponential, all our 
algorithms run in $\poly(\log d, |\PSym|)$, where $d$ is the 
determinant of the input genus. For $n \leq 3$, 
$t \leq \wp d$, and $\wp \leq d^2$. Thus, for any constant $\delta > 0$ 
the generation algorithm runs in time 
$\poly(n,\log d, \log \frac{1}{\delta})$ and succeeds with 
probability at least $1-\delta$.
\end{proof}

\bibliographystyle{alpha}
\bibliography{quadraticforms.bib}

\begin{thebibliography}{DH14b}

\bibitem[Ank52]{Ankeny52}
NC~Ankeny.
\newblock The least quadratic non residue.
\newblock {\em Annals of mathematics}, pages 65--72, 1952.

\bibitem[Cas78]{Cassels78}
John~WS Cassels.
\newblock Rational quadratic forms.
\newblock {\em London and New York}, 1978.

\bibitem[CS99]{CS99}
John Conway and Neil~JA Sloane.
\newblock {\em Sphere packings, lattices and groups}, volume 290.
\newblock Springer, 1999.

\bibitem[DH14a]{DH14Can}
Chandan Dubey and Thomas Holenstein.
\newblock Computing the $p$-adic canonical quadratic form in polynomial time.
\newblock {\em arXiv preprint arXiv:1409.6199}, 2014.

\bibitem[DH14b]{DH14}
Chandan Dubey and Thomas Holenstein.
\newblock Sampling a uniform random solution of a quadratic equation modulo
  $p^k$.
\newblock {\em arXiv preprint arXiv:1404.0281}, 2014.

\bibitem[Die03]{Dietmann03}
Rainer Dietmann.
\newblock Small solutions of quadratic diophantine equations.
\newblock {\em Proceedings of the London Mathematical Society},
  86(03):545--582, 2003.

\bibitem[Gau86]{Gauss86}
Carl~Friedrich Gau{\ss}.
\newblock Disquisitiones arithmeticae, 1801. english translation by arthur a.
  clarke, 1986.

\bibitem[Har08]{Hartung08}
Rupert Hartung.
\newblock {\em Computational problems of quadratic forms: complexity and
  cryptographic perspectives}.
\newblock PhD thesis, Ph. D. thesis, Goethe-Universit{\"a}t Frankfurt a. M.,
  2008, http://publikationen. ub. uni-frankfurt.
  de/volltexte/2008/5444/pdf/HartungRupert. pdf, 2008.

\bibitem[HR07]{HR07}
Ishay Haviv and Oded Regev.
\newblock Tensor-based hardness of the shortest vector problem to within almost
  polynomial factors.
\newblock In {\em Proceedings of the thirty-ninth annual ACM symposium on
  Theory of computing}, pages 469--477. ACM, 2007.

\bibitem[HR14]{HR13}
Ishay Haviv and Oded Regev.
\newblock On the lattice isomorphism problem.
\newblock {\em SODA}, pages 391--404, 2014.

\bibitem[IK04]{Kowalski04}
Henryk Iwaniec and Emmanuel Kowalski.
\newblock {\em Analytic number theory}, volume~53.
\newblock American Mathematical Society Providence, 2004.

\bibitem[Jon44]{Jones44}
Burton~W Jones.
\newblock A canonical quadratic form for the ring of 2-adic integers.
\newblock {\em Duke Math. J}, 11(715):e727, 1944.

\bibitem[Jon50]{Jones50}
Burton~Wadsworth Jones.
\newblock {\em The arithmetic theory of quadratic forms}, volume~10.
\newblock Mathematical Association of America, distributed by Wiley [New York,
  1950.

\bibitem[Kit99]{Kitaoka99}
Yoshiyuki Kitaoka.
\newblock {\em Arithmetic of quadratic forms}, volume 106.
\newblock Cambridge University Press, 1999.

\bibitem[Kne02]{Kneser02}
Martin Kneser.
\newblock {\em Quadratische formen}.
\newblock Springer DE, 2002.

\bibitem[MH73]{MH73}
John~Willard Milnor and Dale Husem{\"o}ller.
\newblock {\em Symmetric bilinear forms}.
\newblock Springer, 1973.

\bibitem[Min10]{Minkowski10}
Hermann Minkowski.
\newblock {\em Geometrie der zahlen}.
\newblock Berlin, 1910.

\bibitem[O'M73]{OMeara73}
Onorato~Timothy O'Meara.
\newblock {\em Introduction to quadratic forms}, volume 117.
\newblock Springer, 1973.

\bibitem[Pal65]{Pall65}
Gordon Pall.
\newblock The weight of a genus of positive n-ary quadratic forms.
\newblock In {\em Proc. Sympos. Pure Math}, volume~8, pages 95--105, 1965.

\bibitem[Sie35]{Siegel35}
Carl~Ludwig Siegel.
\newblock {\"U}ber die analytische theorie der quadratischen formen.
\newblock {\em The Annals of Mathematics}, 36(3):527--606, 1935.

\bibitem[Sie72]{Siegel72}
Carl~Ludwig Siegel.
\newblock {\em Zur theorie der quadratischen formen}.
\newblock Vandenhoeck und Ruprecht, 1972.

\bibitem[Wat60]{Watson60}
George~Leo Watson.
\newblock {\em Integral quadratic forms}.
\newblock Cambridge, 1960.

\bibitem[Wat76]{Watson76}
GL~Watson.
\newblock The 2-adic density of a quadratic form.
\newblock {\em Mathematika}, 23(01):94--106, 1976.

\bibitem[Wed01]{Wedeniwski01}
Sebastian Wedeniwski.
\newblock {\em Primality Tests on Commutator Curves}.
\newblock PhD thesis, Eberhard-Karls-Universit\"at T\"ubingen, 2001.

\end{thebibliography}

\appendix

\section{Diagonalizing a Matrix}\label{sec:BlockDiagonal}

In this section, we provide a proof of Theorem \ref{thm:BlockDiagonal}.

\paragraph{Module.} There are quadratic forms which have no associated 
lattice e.g., negative
definite quadratic forms. To work with these, we define the concept of
free modules (henceforth, called module) which behave as vector 
space but have no associated realization
over the Euclidean space $\bbR^n$.

If $M$ is finitely generated $\Ring$-module with generating set
$\Vx_1,\cdots,\Vx_n$ then the elements $\Vx \in M$ can
be represented as $\sum_{i=1}^n r_i \Vx_i$, such that
$r_i \in \Ring$ for every $i \in [n]$. By construction,
for all
$a,b \in R$, and $\Vx,\Vy \in M$;
\[
a(\Vx+\Vy)=a\Vx+a\Vy \qquad (a+b)\Vx=a\Vx+b\Vx \qquad a(b\Vx)=(ab)\Vx
\qquad 1\Vx=\Vx
\]
Note that, if we replace $\Ring$ by a field in the definition 
then we get a vector space (instead of a module). 
Any inner product
$\beta:M\times M \to \Ring$ gives rise to a quadratic form 
$\MQ\in\Ring^{n\times n}$ as follows;
\[
\MQ_{ij} = \beta(\Vx_i,\Vx_j) \;.
\]
Conversely, if $R=\bbZ$ then by definition, every symmetric matrix 
$\MQ \in \bbZ^{n\times n}$ gives rise to an inner product $\beta$ 
over every $\bbZ$-module $M$; as follows.
Given $n$-ary integral quadratic form $\MQ$ and a $\bbZ$-module
$M$ generated by the basis $\{\Vx_1,\cdots,\Vx_n\}$ we define the 
corresponding inner product $\beta:M\times M \to \bbZ$ as;
\[
\beta(\Vx,\Vy)=\sum_{i,j}c_id_j\MQ_{ij}
\text{ where, }\Vx=\sum_{i}c_i\Vx_i ~~ \Vy=\sum_{j}d_j\Vx_j\;.
\]
In particular, any integral quadratic form $\MQ^n$ can be interpreted 
as describing an inner product over a free module of dimension $n$.

For studying quadratic forms over $\zpkz$, where $p$ is a prime and $k$
is a positive integer; the first step is to find equivalent quadratic 
forms which have as few mixed terms as possible (mixed terms are terms
like $x_1x_2$).

\begin{proof}(Theorem \ref{thm:BlockDiagonal})
The transformation of the matrix $\MQ$ to a block diagonal form involves
three different kinds of transformation. We first describe these 
transformations on $\MQ$ with small dimensions (2 and~3).

\begin{enumerate}[(1)]
\item Let $\MQ$ be a $2\times 2$ integral quadratic form. Let us also 
assume that
the entry with smallest $p$-order in $\MQ$ is a diagonal entry, say
$\MQ_{11}$. Then, $\MQ$ is of the following form; where $\alpha_1,\alpha_2$
and $\alpha_3$ are units of $\zpz$.
\[
\MQ = \begin{pmatrix}p^i \alpha_1 & p^j \alpha_2 \\ p^j \alpha_2 & p^s\alpha_3 
\end{pmatrix} \qquad i \leq j,s
\]
The corresponding $\MU \in \text{SL}_2(\zpkz)$, that diagonalizes $\MQ$ 
is given below. The number
$\alpha_1$ is a unit of $\zpz$ and so $\alpha_1$ has an inverse in $\zpkz$.  
\[
\MU = \begin{pmatrix} 1 & -\frac{p^{j-i}\alpha_2}{\alpha_1} \bmod{p^k} \\ 
0 & 1\end{pmatrix} 
\qquad
\MU'\MQ\MU \equiv \begin{pmatrix} p^i\alpha_1 & 0 \\ 0 & p^s\alpha_3 - 
p^{2j-i}\frac{\alpha_2^2}{\alpha_1}\end{pmatrix}\pmod{p^k}
\]

\item If $\MQ^2$ does not satisfy the condition of item (1) i.e., 
the off diagonal entry is the one with smallest $p$-order, then we start
by the following transformation $\MV \in \SL_2(\zpkz)$.
\begin{gather*}
\MV = \begin{pmatrix}1 & 0 \\ 1 & 1\end{pmatrix} \qquad 
\MV'\MQ\MV = \begin{pmatrix} \MQ_{11}+2\MQ_{12}+\MQ_{22} & \MQ_{12}+\MQ_{22} \\
\MQ_{12}+\MQ_{22} & \MQ_{22}\end{pmatrix} 
\end{gather*}
If $p$ is an odd prime then $\ordp(\MQ_{11}+2\MQ_{12}+\MQ_{22})=\ordp(\MQ_{12})$, 
because $\ordp(\MQ_{11}),$ $\ordp(\MQ_{22}) >\ordp(\MQ_{12})$. By definition,
$\MS=\MV'\MQ\MV$ is equivalent to $\MQ$ over the ring $\zpkz$. But now, $\MS$
has the property that $\ordp(\MS_{11}) = \ordp(\MS_{12})$, and it can be 
diagonalized using the transformation in (1). The final transformation
in this case is the product of $\MV$ and the subsequent transformation
from item (1). The product of two matrices from $\SL_2(\zpkz)$ is also
in $\SL_2(\zpkz)$, completing the diagonalization in this case.

\item If $p=2$, then the transformation in item (2) fails. In this case,
it is possible to subtract a linear combination of these two rows/columns
to make everything else on the same row/column equal to zero over $\ztkz$.
The simplest such transformation is in dimension~3. The situation is as 
follows. Let $\MQ^3$ be a quadratic form whose off diagonal entry has the 
lowest possible power of
$2$, say $2^{\ell}$ and all diagonal entries are divisible by at least
$2^{\ell+1}$. In this case, the matrix $\MQ$ is of the
following form.
\[
\MQ = \begin{pmatrix}2^{\ell + 1} a & 2^{\ell}b & 2^id \\
2^{\ell}b & 2^{\ell+1}c & 2^je \\
2^id & 2^je & 2^{\ell+1}f \end{pmatrix} \qquad b \text{ odd}, \ell\leq i,j
\]
In such a situation, we consider the matrix $\MU \in \SL_3(\ztkz)$
of the form below such that if $\MS=\MU'\MQ\MU \pmod{2^k}$ then 
$\MS_{13}=\MS_{23}=0$.
\begin{gather*}
\MU = \begin{pmatrix}1 & 0 & -r \\ 0 & 1 & -s \\ 0 & 0 & 1\end{pmatrix}\\
(\MU'\MQ\MU)_{13} \equiv 0 \pmod{2^k} \implies 
r 2a + s b \equiv  2^{i-\ell}d \pmod{2^{k-\ell}}\\
(\MU'\MQ\MU)_{23} \equiv 0 \modtk \implies 
r b + s 2c \equiv 2^{j-\ell}e  \pmod{2^{k-\ell}}
\end{gather*}
For $i,j \geq \ell$ and $b$ odd, the solution $r$ and $s$ can be found by 
the Cramer's rule, as below. The solutions exist because the matrix 
$\begin{pmatrix}2a & b \\ b & 2c\end{pmatrix}$ has determinant $4ac-b^2$, which
is odd and hence invertible over the ring $\bbZ/2^{k-\ell}\bbZ$.
\[
r = \frac{\det\begin{pmatrix}2^{i-\ell}d & s \\ 2^{j-\ell}e & 2c\end{pmatrix}}
{\det\begin{pmatrix}2a & b \\ b & 2c\end{pmatrix}} \pmod{2^{k-\ell}} ~~
s = \frac{\det\begin{pmatrix}2a & 2^{i-\ell}d \\ b & 2^{j-\ell}e\end{pmatrix}}
{\det\begin{pmatrix}2a & b \\ b & 2c\end{pmatrix}}  \pmod{2^{k-\ell}}
\]
\end{enumerate}

This completes the description of all the transformations we are going
to use, albeit for $n$-dimensional $\MQ$ they will be a bit technical.
The full proof for the case of odd prime follows.

Our proof will be a reduction of the problem of diagonalization from
$n$ dimensions to $(n-1)$-dimensions, for the odd primes $p$. We now
describe the reduction.

Given the matrix $\MQ^n$, let $M$ be the corresponding $(\zpkz)$-module
with basis $\MB=[\Vb_1,\cdots,\Vb_n]$ i.e., $\MQ=\MB'\MB$. We first find 
a matrix entry with the smallest $p$-order, say $\MQ_{i^*j^*}$. The 
reduction has two cases: (i) there is a diagonal entry in $\MQ$ with
the smallest $p$-order, and (ii) the smallest $p$-order occurs on an
off-diagonal entry.

We handle case (i) first. Suppose it is possible to pick $\MQ_{ii}$
as the entry with the smallest $p$-order. Our first transformation
$\MU_1 \in \sln(\zpkz)$ is the one which makes the following 
transformation i.e., swaps $\Vb_1$ and $\Vb_i$.
\begin{align}\label{BlockDiagonal:U1}
[\Vb_1,\cdots,\Vb_n] \underset{\MU_1, p^k}{\to} [\Vb_i,\Vb_2,\cdots,
\Vb_{i-1},\Vb_1,\Vb_{i+1},\cdots, \Vb_n]
\end{align}

Let us call the new set of elements $\MB_1=[\Vv_1,\cdots,\Vv_n]$ and
the new quadratic form $\MQ_1=\MB_1'\MB_1 \bmod{p^k}$. Then,
$\Vv_1'\Vv_1$ has the smallest $p$-order in $\MQ_1$ and
$\MU_1'\MQ\MU_1\equiv \MQ_1 \bmod{p^k}$. The next transformation
$\MU_2 \in \sln(\zpkz)$ is as follows. 
\begin{equation}\label{BlockDiagonal:U2}
\Vw_i = \left\{ \begin{array}{ll}
\Vv_1 & \text{if $i=1$}\\
\Vv_i - \frac{\Vv_1'\Vv_i}{p^{\ordp((\MQ_1)_{11})}} \cdot 
\left(\frac{1}{\copp((\MQ_1)_{11})} \bmod{
p^k}\right) \cdot \Vv_1 & \text{otherwise\,.}
\end{array}\right.
\end{equation}
By assumption, $(\MQ_1)_{11}$ is the matrix entry with 
the smallest $p$-order and so $p^{\ordp((\MQ_1)_{11})}$ divides 
$\Vv_{1}'\Vv_i$. Furthermore, $\copp((\MQ_1)_{11})$ is invertible 
modulo $p^k$. Thus, the transformation in Equation 
\ref{BlockDiagonal:U2} is well defined. Also note that it is
a basis transformation, which maps one basis of $\MB_1=[\Vv_1,\cdots,\Vv_n]$
to another basis $\MB_2=[\Vw_1,\cdots,\Vw_n]$. Thus, the 
corresponding basis transformation $\MU_2$ is a 
unimodular matrix over integers, and so $\MU_2\in\sln(\zpkz)$.
Let $\MQ_2=\MU_2'\MQ_1\MU_2 \bmod{p^k}$. Then, we show that the
non-diagonal entries in the entire first row and first column of
$\MQ_2$ are~0. 
\begin{align*}
(\MQ_2&)_{1i(\neq 1)}=(\MQ_2)_{{i1}}=\Vw_1'\Vw_i \bmod{p^k}\\
	&\overset{(\ref{BlockDiagonal:U2})}{\equiv} 
	\Vv_1'\Vv_i - \frac{\Vv_1'\Vv_i}{p^{\ordp((\MQ_1)_{11})}} \cdot 
	\left(\frac{1}{\copp((\MQ_1)_{11})} \bmod{p^k}\right) 
	\cdot \Vv_1'\Vv_1 \\
	&\equiv\Vv_1'\Vv_i - \frac{\Vv_1'\Vv_i}{p^{\ordp((\MQ_1)_{11})}} \cdot 
	\left(\frac{1}{\copp((\MQ_1)_{11})} \bmod{p^k}\right) 
	\cdot p^{\ordp((\MQ_1)_{11})}\copp((\MQ_1)_{11}) \\
	&\equiv 0 \bmod{p^k} 
\end{align*}

Thus, we have reduced the problem to $(n-1)$-dimensions.
We now recursively call this algorithm with the quadratic form
$\MS=[\Vw_2,\cdots,\Vw_{n}]'[\Vw_2,\cdots,\Vw_{n}] \bmod{p^k}$
and let $\MV \in \SL_{n-1}(\zpkz)$ be the output of
the recursion. Then, $\MV'\MS\MV \bmod{p^k}$ is a diagonal
matrix. Also, by consruction $\MQ_2=\diag((\MQ_2)_{11},\MS)$.
Let $\MU_3=1\oplus\MV$, and $\MU=\MU_1\MU_2\MU_3$, then,
by construction, $\MU'\MQ\MU \bmod{p^k}$ is a diagonal
matrix; as follows.
\begin{gather*}
\MU'\MQ\MU \equiv \MU_3'\MU_2'\MU_1'\MQ\MU_1\MU_2\MU_3\equiv
\MU_3'\MQ_2\MU_3\equiv(1\oplus\MV)'\diag((\MQ_2)_{11})(1\oplus\MV)\\
\equiv \diag((\MQ_2)_{11},\MV'\MS\MV) \bmod{p^k}
\end{gather*}

Otherwise, we are in case (ii) i.e., the entry with smallest 
$p$-order in $\MQ$ is an off diagonal entry, say $\MQ_{i^*j^*},
i^*\neq j^*$. Then, we make the following basis transformation
from $[\Vb_1,\cdots,\Vb_n]$ to $[\Vv_1,\cdots,\Vv_n]$ as follows.
\begin{equation}\label{BlockDiagonal:U0}
\Vv_i = \left\{ \begin{array}{ll}
\Vb_{i^*}+\Vb_{j^*} & \text{if $i=i^*$}\\
\Vb_i & \text{otherwise\,.}
\end{array}\right.
\end{equation}
The transformation matrix $\MU_0$ is from $\sln(\zpkz)$.
Recall, 
$\ordp(\MQ_{i^*j^*}) < \ordp(\MQ_{i^*i^*}), \ordp(\MQ_{j^*j^*})$, and so
$\ordp(\Vv_{i^*}'\Vv_{i^*})=\ordp(\Vb_{i^*}'\Vb_{j^*})$. Furthermore, 
$\ordp(\Vv_i'\Vv_j)\geq \ordp(\Vb_{i^*}'\Vb_{j^*})$, and so the minimum 
$p$-order does not change after the transformation in Equation 
(\ref{BlockDiagonal:U0}). This transformation reduces the problem to the 
case when the matrix entry with minimum $p$-order appears on the 
diagonal. This completes the proof of the theorem for odd primes
$p$.

For $p=2$, exactly the same set of transformations works, unless the 
situation in item (3) arises. In such a case, we use the type II block
to eliminate all other entries on the same rows/columns as the type II
block. Thus, in this case, the problem reduces to one in dimension $(n-2)$.

The algorithm uses $n$ iterations, reducing the dimension by~1 in each 
iteration. In each iteration, we have to find the minimum $p$-order, costing
$O(n^2\log k)$ ring operations and then~3 matrix multiplications costing $O(n^3)$ operations
over $\zpkz$. Thus, the overall complexity is $O(n^4+n^3\log k)$ or 
$O(n^4\log k)$ ring operations.
\end{proof}

\section{Missing Proofs}\label{sec:Proofs}

\begin{proof}(proof of Lemma \ref{lem:Square})
We split the proof in two parts: for odd primes $p$ and for the prime~2.
\begin{description}
\item [Odd Prime.]
If $0 \neq t \in \zpkz$ then $\ordp(t)<k$. If
$t$ is a square modulo $p^k$ then there exists a $x$
such that $x^2\equiv t \pmod{p^k}$. Thus, there exists $a \in \bbZ$
such that $x^2=t+ap^k$. But then, $2\ordp(x)=\ordp(t+ap^k)=\ordp(t)$.
This implies that $\ordp(t)$ is even and $\ordp(x)=\ordp(t)/2$.
Substituting this into $x^2=t+ap^k$ and dividing the entire equation
by $p^{\ordp(t)}$ yields that $\copp(t)$ is a
quadratic residue modulo $p$; as follows.
\[
\copp(x)^2 = \copp(t)+ap^{k-\ordp(t)}\equiv \copp(t) \pmod{p}
\]

Conversely, if $\copp(t)$ is a quadratic residue modulo $p$ 
then there exists a $u \in \zpkz$ such that $u^2\equiv \copp(t) \pmod{p^k}$,
by Lemma \ref{lem:Square}. If $\ordp(t)$ is even then 
$x=p^{\ordp(t)/2}u$ is a solution to the equation 
$x^2\equiv t \pmod{p^k}$.

\item [Prime $2$.] 
If $0 \neq t \in \ztkz$ then $\ordt(t)<k$. If $t$ is a square modulo 
$2^k$ then there exists an integer $x$ such that $x^2\equiv t \modtk$.
Thus, there exists an integer $a$ such that $x^2=t+a2^k$. But then,
$2\ordt(x)=\ordt(t+a2^k)=\ordt(t)$. This implies that $\ordt(t)$ is 
even and $\ordt(x)=\ordt(t)/2$. Substituting this into the equation
$x^2=t+a2^k$ and dividing the entire equation by $2^{\ordt(t)}$ 
yields,
\[
\copt(x)^2 = \copt(t) + a2^{k-\ordt(t)} \qquad \copt(t) < 2^{k-\ordt(t)}\;.
\]
But $\copt(x)$ is odd and hence
$\copt(x)^2 \equiv 1 \pmod{8}$. If $k-\ordt(t)>2$, then 
$\copt(t)\equiv 1 \pmod{8}$. Otherwise, if $k-\ordt(t)\leq 2$ 
then $\copt(t) < 2^{k-\ordp(t)}$ implies that $\copt(t)=1$.

Conversely, if $\copt(t)\equiv 1 \pmod{8}$ then there exists a 
$u \in \ztkz$ such that $u^2\equiv \copt(t) \modtk$, by Lemma
\ref{lem:Square}. If $\ordt(t)$ is even then $x=2^{\ordt(t)/2}u$
is a solution to the equation $x^2\equiv t\modtk$.
\end{description}
\end{proof}

\begin{proof}(Theorem \ref{thm:Siegel13})
We do the proof in two steps: (i) if $t$ has a primitive $p^k$-
representation in $\MQ$ then $t$ has a primitive $p^*$-representation
in $\MQ$, and (ii) if $t$ has a primitive $p^*$-representation in
$\MQ$ then $t$ has a primitive $p^*$-representation in $\tMQ$
for all $\tMQ$ such that $\tMQ \eqp \MQ$.

The proof of (i) follows.
By assumption, there exists a primitive $\Vx \in (\bbZ/p^k\bbZ)^n$ 
such that $\Vx'\MQ\Vx\equiv t \pmod{p^k}$. Let $a=\Vx'\MQ\Vx$ be 
an integer, then by definition of symbols $a$ and $t$ have the 
same $p^k$-symbol. This implies that for all 
$i\geq k$ there exists a unit $u_i \in \bbZ/p^i\bbZ$ such that 
$u_i^2a\equiv t \pmod{p^i}$. It follows that $u_i\Vx$ is a 
primitive representation of $t$ in $\bbZ/p^i\bbZ$. But,
if $\Vx$ is a primitive representation of $t$ by $\MQ$ over 
$\bbZ/p^i\bbZ$ then $\Vx$ is also a primitive representation 
of $t$ by $\MQ$ over $\bbZ/p^j\bbZ$, for all positive integers 
$j \leq i$. This completes the proof of (i).

The proof of (ii) follows. Let $K$ be an arbitrary positive 
integer and $\Vx \in (\bbZ/p^K\bbZ)^n$ be a primitive vector 
such that $\Vx'\MQ\Vx\equiv t \bmod{p^K}$. As $\tMQ \eqp \MQ$, there exists
$\MU \in \gln(\bbZ/p^K\bbZ)$ such that $\MQ\equiv\MU'\tMQ\MU \bmod{p^K}$.
Thus, $(\MU\Vx)'\tMQ(\MU\Vx) \equiv t \bmod{p^K}$ and $\MU\Vx$
is a $p^K$-representation of $t$ in $\tMQ$. If $\Vx$ is primitive then
so is $\MU\Vx$. As $K$ is arbitrary, the proof of (ii) and hence
the theorem is complete.
\end{proof}

\section{Computer Assisted Proofs}\label{sec:Code}

In this section, we provide the \textsc{Maple} code for the 
computer Assisted proofs. The
procedure $\FXI$ computes the function $\xi$ and the names
of the other procedures are self-explanatory.

Following are the names of the variables that we use.
\begin{align*}
\small
\begin{array}{ll}
\begin{array}{ll}
rh &= \rho \\
a2 &= a_2 \\
sdp &= \fS_{d+}\\
s37 &= \fS_{\{3,7\}}\\
sm37 &= \fS_{-}+\fS_{\{3,7\}}\\ 
lega &= \legendre{a_2}{2}
\end{array} &
\begin{array}{ll}
eps &= \epsilon \\
b2 &= b_2 \\
sdm &= \fS_{d-}\\
s35 &= \fS_{\{3,5\}}\\
sm35 &= \fS_{-}+\fS_{\{3,7\}}\\
legb &= \legendre{b_2}{2}
\end{array} 
\end{array}
\end{align*}

When run on \textsc{Maple}, none of these codes output ``FAIL!''.

\newpage
\begin{figure}[h]
\small
\setlength{\unitlength}{0.175in}
\begin{picture}(30,4)
\put(0,4){\textsc{fxi}:=\PROC($s::$integer)::integer;}
\put(0.5,3){$f=1$;}
\put(0.5,2){\IF $s \bmod 4 = 1$ \OR $s \bmod 4 = 0$ \THEN $~f:=0;$ \ENDIF:}
\put(0.5,1){\RET $f$:}
\put(0,0){\END \PROC:}
\end{picture}
\end{figure}

\begin{figure}[h]
\small
\setlength{\unitlength}{0.175in}
\begin{picture}(30,24)
\put(0,24){\textsc{TypeIIBruteForce} := \PROC();}
\put(0.5,23){\FOR $rh$ \IN $\{-1,1\}$ \DO:\FOR $eps$ \IN $\{-1,1\}$ \DO:}
\put(1,22){$sig := rh\cdot (1+eps);odty := 0;pexs := (odty - sig) \bmod 8;$}
\put(1,21){\FOR $s1p$ \IN $\{0,1\}$ \DO: \FOR $s1m$ \IN $\{0,1\}$ \DO:}
\put(1,20){\FOR $s5p$ \IN $\{0,1\}$ \DO: \FOR $s5m$ \IN $\{0,1\}$ \DO:}
\put(1,19){\FOR $s3p$ \IN $\{0,1,2,3\}$ \DO: \FOR $s3m$ \IN $\{0,1,2,3\}$ \DO:}
\put(1,18){\FOR $s7p$ \IN $\{0,1,2,3\}$ \DO: \FOR $s7m$ \IN $\{0,1,2,3\}$ \DO:}
\put(1.5,17){$s3:=s3p+s3m;s5:=s5p+s5m;$}
\put(1.5,16){$sm:=s1m+s3m+s5m+s7m;s7:=s7p+s7m;$}
\put(1.5,15){$s37:=s3+s7;sm35:=sm+s3+s5;sm57:=sm+s5+s7;$}
\put(1.5,14){$sx:=2\cdot s3+4\cdot s5+6 \cdot s7;$}
\put(1.5,13){\IF $pexs = sx+2\cdot(1-eps^{s37}\cdot(-1)^{(sm+\FXI(s37))}) \bmod 8$ \THEN}
\put(2,12){\IF \NOT ($rh=1$ \AND $\TYPE(sm35, \EVEN)$) \AND}
\put(2,11){\NOT ($rh=-1$ \AND $\TYPE(sm57, \EVEN)$) \AND}
\put(2,10){\NOT ($rh=1$ \AND $eps=1$ \AND $\TYPE(sm57, \ODD)$) \AND}
\put(2,9){\NOT ($rh=1$ \AND $eps=-1$ \AND $\TYPE(sm57, \EVEN)$) \AND}
\put(2,8){\NOT ($rh=-1$ \AND $eps=-1$ \AND $\TYPE(sm35, \EVEN)$) \AND}
\put(2,7){\NOT ($rh=-1$ \AND $eps=1$ \AND $\TYPE(sm35, \ODD)$) \THEN}
\put(2.5,6){print(``FAIL!'');}
\put(2,5){\ENDIF:}
\put(1.5,4){\ENDIF:}
\put(1,3){\END \DO: \END \DO:\END \DO: \END \DO:}
\put(1,2){\END \DO: \END \DO:\END \DO: \END \DO:}
\put(0.5,1){\END \DO: \END \DO:}
\put(0,0){\END \PROC:}
\end{picture}
\end{figure}

\newpage

\begin{figure}[h]
\small
\setlength{\unitlength}{0.175in}
\begin{picture}(30,31)
\put(0,31){\textsc{TypeIEvenBruteForce} := \PROC();}
\put(0.5,30){\FOR $rh$ \IN $\{-1,1\}$ \DO:\FOR $eps$ \IN $\{-1,1\}$ \DO:}
\put(0.5,29){\FOR $a2$ \IN $\{1,3,5,7\}$ \DO:\FOR $b2$ \IN $\{1,3,5,7\}$ \DO:}
\put(1,28){$sig := rh\cdot (1+eps);odty := a2+b2 \bmod 8;pexs := (odty - sig) \bmod 8;$}
\put(1,27){$leg:=$numtheory[legendre]($a2\cdot b2$,$2$);
$X:=\{rh\cdot a2 \bmod 8, rh\cdot b2 \bmod 8\}$;}
\put(1,26){\FOR $s1p$ \IN $\{0,1\}$ \DO: \FOR $s1m$ \IN $\{0,1\}$ \DO:}
\put(1,25){\FOR $s5p$ \IN $\{0,1\}$ \DO: \FOR $s5m$ \IN $\{0,1\}$ \DO:}
\put(1,24){\FOR $s3p$ \IN $\{0,1,2,3\}$ \DO: \FOR $s3m$ \IN $\{0,1,2,3\}$ \DO:}
\put(1,23){\FOR $s7p$ \IN $\{0,1,2,3\}$ \DO: \FOR $s7m$ \IN $\{0,1,2,3\}$ \DO:}
\put(1.5,22){$s3:=s3p+s3m;s5:=s5p+s5m;sx:=2\cdot s3+4\cdot s5+6 \cdot s7;$}
\put(1.5,21){$sm:=s1m+s3m+s5m+s7m;s7:=s7p+s7m;$}
\put(1.5,20){$s37:=s3+s7;sm37:=sm+s37;$}
\put(1.5,19){\IF $pexs = sx+2\cdot(1-eps^{s37}\cdot(-1)^{sm+\FXI(s37)}) \bmod 8$ \AND
$leg=(-1)^{s35}$ \THEN}
\put(2,18){\IF \NOT ($rh=1$ \AND $\TYPE(sm, \EVEN)$ \AND}
\put(2.25,17){$\NOPS(X$\INT$\{1,5\})>0$) \AND}
\put(2.25,16){\NOT ($rh=-1$ \AND $\TYPE(sm37, \EVEN)$ \AND}
\put(2.25,15){$\NOPS(X$\INT$\{1,5\})>0$) \AND}
\put(2.25,14){\NOT ($rh=1$ \AND $eps=1$ \AND $\TYPE(sm37, \ODD)$ \AND}
\put(2.25,13){$\NOPS(X$\INT$\{3,7\})>0$) \AND}
\put(2.25,12){\NOT ($rh=1$ \AND $eps=-1$ \AND $\TYPE(sm37, \EVEN)$ \AND}
\put(2.25,11){$\NOPS(X$\INT$\{3,7\})>0$) \AND}
\put(2.25,10){\NOT ($rh=-1$ \AND $eps=-1$ \AND $\TYPE(sm, \EVEN)$ \AND}
\put(2.25,9){$\NOPS(X$\INT$\{3,7\})>0$) \AND}
\put(2.25,8){\NOT ($rh=-1$ \AND $eps=1$ \AND $\TYPE(sm, \ODD)$ \AND}
\put(2.25,7){$\NOPS(X$\INT$\{3,7\})>0$) \THEN}
\put(2.5,6){print(``FAIL!'');}
\put(2,5){\ENDIF:}
\put(1.5,4){\ENDIF:}
\put(1,3){\END \DO: \END \DO:\END \DO:\END \DO:}
\put(1,2){\END \DO: \END \DO:\END \DO:\END \DO:}
\put(0.5,1){\END \DO: \END \DO:\END \DO:\END \DO:}
\put(0,0){\END \PROC:}
\end{picture}
\end{figure}

\newpage
\begin{figure}
\small
\setlength{\unitlength}{0.175in}
\begin{picture}(30,26)
\put(0,26){\textsc{TypeIOddBruteForce} := \PROC();}
\put(0.5,25){\FOR $rh$ \IN $\{-1,1\}$ \DO:\FOR $eps$ \IN $\{-1,1\}$ \DO:}
\put(0.5,24){\FOR $a2$ \IN $\{1,3,5,7\}$ \DO:\FOR $b2$ \IN $\{1,3,5,7\}$ \DO:}
\put(1,23){$legb:=$numtheory[legendre]$(b2,2);lega:=$numtheory[legendre]$(a2,2);$}
\put(1,22){\IF $legb=-1$ \THEN $~odty:=odty+4 \bmod 8$; \ENDIF:}
\put(1,21){$sig := rh\cdot (1+eps);pexs := (odty - sig) \bmod 8;$}
\put(1,20){$leg:=$numtheory[legendre]$(a2\cdot b2,2)$;}
\put(1,19){\FOR $s1p$ \IN $\{0,1\}$ \DO: \FOR $s1m$ \IN $\{0,1\}$ \DO:}
\put(1,18){\FOR $s5p$ \IN $\{0,1\}$ \DO: \FOR $s5m$ \IN $\{0,1\}$ \DO:}
\put(1,17){\FOR $s3p$ \IN $\{0,1,2,3\}$ \DO: \FOR $s3m$ \IN $\{0,1,2,3\}$ \DO:}
\put(1,16){\FOR $s7p$ \IN $\{0,1,2,3\}$ \DO: \FOR $s7m$ \IN $\{0,1,2,3\}$ \DO:}
\put(1.5,15){$s3:=s3p+s3m;s5:=s5p+s5m;sx:=2\cdot s3+4\cdot s5+6 \cdot s7;$}
\put(1.5,14){$sm:=s1m+s3m+s5m+s7m;s7:=s7p+s7m;$}
\put(1.5,13){$s37:=s3+s7;s35:=s3+s5;sm37:=sm+s37;sm35:=sm+s35;$}
\put(1.5,12){\IF $pexs = sx+2\cdot(1-eps^{s37}\cdot(-1)^{(sm+\FXI(s37))}) \bmod 8$ \AND}
\put(1.75,11){$leg=(-1)^{s35}$ \THEN}
\put(2,10){\IF \NOT $(rh\cdot a2 \bmod 4 = 1$ \AND $(-1)^{sm}\cdot rh^{s37}\cdot lega=1)$ \AND }
\put(2,9){\NOT $(rh\cdot a2 \bmod 4 = 3$ \AND $(-1)^{(sm37+1)}\cdot rh^{s37}\cdot 
eps\cdot lega=1)$ \AND}
\put(2,8){\NOT $(rh\cdot b2 \bmod 4 = 1$ \AND $(-1)^{sm35}\cdot rh^{s37}\cdot legb=1)$ \AND }
\put(2,7){\NOT $(rh\cdot b2 \bmod 4 = 3$ \AND $(-1)^{(sm57+1)}\cdot rh^{s37} 
\cdot eps\cdot legb=1)$ \THEN}
\put(2.5,6){print(``FAIL!'');}
\put(2,5){\ENDIF:}
\put(1.5,4){\ENDIF:}
\put(1,3){\END \DO: \END \DO:\END \DO:\END \DO:}
\put(1,2){\END \DO: \END \DO:\END \DO:\END \DO:}
\put(0.5,1){\END \DO: \END \DO:\END \DO:\END \DO:}
\put(0,0){\END \PROC:}
\end{picture}
\end{figure}

\end{document}